\documentclass[letterpaper,12pt]{article}

\usepackage{graphicx}                
\usepackage{float}                   
\usepackage[hypertexnames=false]{hyperref}               
\usepackage[margin=1in]{geometry}  
\usepackage[short]{datetime}         
\usepackage[labelfont=bf]{caption}   
\usepackage[bottom]{footmisc}        
\usepackage{lineno}
\usepackage{ragged2e}
\usepackage[protrusion=true,expansion=true]{microtype} 
\usepackage[T1]{fontenc}                               
\usepackage{lmodern}

\usepackage[compact]{titlesec}
\titleformat*{\section}{\large\bfseries}
\titlespacing*{\section}{0pt}{1ex}{1ex}
\titleformat*{\subsection}{\normalsize\bfseries}
\titlespacing*{\subsection}{0pt}{0ex}{0ex}
\titleformat*{\subsubsection}{\normalsize\bfseries}
\titlespacing*{\subsubsection}{0pt}{0ex}{0ex}

\usepackage{amsfonts,amssymb,amsmath, mathtools}
\usepackage{xcolor, color, soul}
\usepackage{graphicx}
\usepackage{multicol}
\usepackage{url}
\usepackage{faktor}
\usepackage{ulem}
\usepackage{pdfpages}
\usepackage{soul}
\usepackage{relsize}
\usepackage{enumitem}
\usepackage{setspace}
\DeclareMathAlphabet\mathbfcal{OMS}{cmsy}{b}{n}
\usepackage{bm}

\usepackage{amsthm}
\usepackage{natbib}
\usepackage{tikz}
\usepackage{circledsteps}
\usepackage{bibunits}
\defaultbibliographystyle{apalike} 
\defaultbibliography{refs}

\def\indep{{\,\perp \!\!\! \perp\,}}
\newcommand{\Var}[0]{\text{Var}}

\newcommand{\Cov}[0]{\text{Cov}}

\newtheorem{theorem}{Theorem}

\newtheorem{lemma}{Lemma}
\newtheorem{supptheorem}[lemma]{Theorem}
\newtheorem{supplemma}[lemma]{Lemma}
\newtheorem{suppcoro}[lemma]{Corollary}

\newtheorem*{assumption*}{\assumptionnumber}
\providecommand{\assumptionnumber}{}


\begin{document}

\center{
\textbf{\Large Causal Inference for Functional Treatments with Stochastic Policies} \\}
\vspace{-.05in}
\center{\large Martha Barnard, Jared D. Huling, Julian Wolfson}
\vspace{-.1in}
\center{Division of Biostatistics and Health Data Science, School of Public Health, University of Minnesota}
\begin{abstract}
    Wearable devices can accurately measure human behavior, providing a unique opportunity to understand how behavior impacts health. Recent studies leveraging functional regression methods have found a strong relationship between accelerometer-collected physical activity and mortality. However, to determine if physical activity patterns impact mortality it is necessary to understand the \textit{causal} effects of policies for physical activity, i.e., a function-valued treatment. Functional treatments present several challenges for causal effect estimation: 1) defining a scientifically meaningful estimand that reflects real-world policies and satisfies positivity is nontrivial; and 2) the potential for temporal confounding over continuous time.
To address these, we propose stochastic policies for functional treatments that allow estimation of causal effects of changing the treatment distribution without requiring a positivity assumption. 
We develop a novel method for such that modifies the treatment through a single basis function chosen by the analyst, allowing for clear control over treatment modification and temporal confounding feedback. 
We show asymptotic normality of our estimators and that they exhibit rate double robustness. We apply our methods to the National Health and Nutrition Examination Survey to determine the causal effect of increasing physical activity over three-hour periods on mortality.
\end{abstract}
\justifying
\noindent%
\vfill
\noindent\rule{2in}{0.4pt} \\
{\small * barna126@umn.edu}

\setstretch{1.2}
\begin{bibunit}
\section{Introduction}
The ubiquity of mobile and wearable devices has made accurate and densely measured human behavior data more accessible than ever before, providing an unprecedented opportunity to understand how human behavior affects health.
One example of this data is 
24-hour, minute-level physical activity data collected with accelerometers as part of the National Health and Nutrition Examination Survey (NHANES) \citep{chen_national_2018}.  This minute-level data can be considered a realization of a true continuous, infinite-dimensional function over the 24-hour time period and analyzed using functional data methods \citep{crainiceanu_functional_2024}.
However, there is an interest in understanding the causal, rather than associative or predictive, relationship between this data and health outcomes to identify behavior changes that improve population health. 
Specifically, we want to estimate the causal effect of changes in physical activity on all-cause mortality where:
a) physical activity is viewed as a functional treatment
b) the set of ``treatments" (i.e., potential physical activity trajectories) is both individually plausible and easily modifiable by the analyst
c) assumptions about time-varying confounding of the effect of physical activity on mortality are as minimal as possible


There has been recent methodological development in estimating the causal effect of a functional treatment, which presents unique challenges due to its infinite dimension \citep{delaigle_defining_2010}. \cite{zhang_covariate_2021}, \cite{tan_causal_2025}, and  \cite{wang_flexible_2026} develop methods for estimating the average dose response functional (ADRF), i.e., the expectation of the counterfactual outcome given that every individual has been assigned a specific, fixed functional treatment.
While these methods do address the technical challenges presented by an infinite dimensional treatment, there are a variety of limitations to estimating the ADRF for a functional treatment. Positivity is often violated for the ADRF, especially with physical activity data, as many physical activity trajectories are implausible or impossible for some individuals. For example, older individuals may not be able to achieve a high level of physical activity that younger individuals exhibit regularly. Furthermore, since the ADRF is often unrealistic for some subset of individuals, it also tends to lack a meaningful scientific interpretation. \cite{jiang_estimating_2026} address the limitations of the ADRF through modified functional treatment policies (MFTPs) in which the estimand of interest is the expected counterfactual outcome given a modification to the natural value of treatment according to a specified deterministic rule. However, it can be challenging to \textit{a priori} specify a deterministic modification rule that is scientifically relevant and satisfies positivity for all combinations of covariate values.

Another key limitation of the available methods for functional treatment causal effect estimation is that they rely on the assumption of no time-varying confounding across the functional domain. However, this assumption is unlikely to be satisfied for physical activity data. For example,  physical activity levels in the morning may impact physical activity levels in the evening as well as the health outcome of interest. 
In contrast, there are a variety of methods for adjusting for time-varying confounding for longitudinal treatments received at discrete time points, such as g-computation and marginal structural models \citep{robins_new_1986, robins_marginal_2000}. 
However, for functional treatments it is often unclear for which time periods it is necessary to adjust for time-varying confounding as there is the potential for an infinite number of treatment-confounder ``feedbacks''. \cite{ying_causality_2024-1} and  \cite{ying_causality_2024} causally identify the effect of continuous-time dynamic treatment regimes while accounting for potentially infinite treatment-confounder feedback, however they do not propose estimators for their identified estimand.

To address these challenges, we develop a method for estimating the causal effect of a modification to the functional treatment distribution over a subset of the temporal domain (e.g, three hours) with stochastic dynamic policies. By considering a policy on a subset of the temporal domain, we can evaluate the impact of behavior changes over realistic time periods and systematically address temporal confounding.
We propose representing the functional treatment through a novel basis expansion and then modifying the treatment distribution through a single basis function selected by the analyst; this allows for a clear control and interpretation of the treatment policy and for temporal confounding feedback to occur across a majority of the functional treatment. We then develop stochastic policies and corresponding estimands for functional treatments that adjust for temporal confounding across the entire temporal domain. These policies 1) do not require a positivity assumption; 2) vary with covariates implicitly in a realistic manner; and 3) reflect real-world policies by modifying the \textit{distribution} of treatment or behavior rather than the treatment itself \citep{kennedy_nonparametric_2019, diaz_causal_2020, diaz_nonparametric_2023,schindl_incremental_2026}. We propose double machine learning-style estimators for these estimands and show they are asymptotically normal and rate double robust. We validate the performance of our proposed estimators through simulation studies and apply our methods to NHANES to estimate the effect of increases of physical activity over different time periods on 5-year all-cause mortality.

\section{Background and rationale}
\label{sec:more_lit_review}

\subsection{Notation}
Let $\{\bm{X}_i, Y_i, A_i(\cdot)\}_{i=1}^n$ be an identically distributed and independent sample from a population where $\bm{X}_i$ is a $p$ length vector of pre-treatment covariates, $Y_i$ is an outcome of interest and $A_i(t)$ is a stochastic process over $t \in [0, T]$, representing a functional treatment of interest. 
 Let $Y(a(\cdot))$ be the potential outcome of $Y$ given that $a(\cdot)$ had been observed.  Let $||\cdot||_2$ be the Euclidean norm. For our application of interest, $T$ is 24 hours. However, it is unrealistic for individuals to modify behavior over a 24-hour period so we consider policies enacted on $A(t)$ over $[t_1, t_2] \subset [0, T]$.
 Thus, we partition $A(\cdot)$ into the following three functions: $A^{(1)}(t) = \{A(t); t\in T_1 = [0, t_1)\}$, $A^{(2)}(t)= \{A(t); t \in T_2 = [t_1, t_2]\}$, and $A^{(3)}(t) = \{A(t); t \in T_3 = (t_2, T]\}$. With slight abuse of notation, we also sometimes consider $A^{(k)}(t) = A(t)I(t \in T_k)$ such that these functions share the same domain, $[0, T]$.

\subsection{Time-varying confounding}
As far as we are aware, there are no available estimators that account for the potentially infinite number of treatment-confounder ``feedbacks''between a functional treatment and scalar outcome. However, by making selective assumptions about time-varying confounding, the complexity of adjusting confounding over the entire, continuous temporal domain can be reduced. While assuming no time-varying confounding is often unrealistic for treatments measured over months or years, it may be reasonable over shorter periods, such as a couple hours (e.g., subsets of the 24-hour physical activity data). 
 Therefore, to leverage finite, discrete-time longitudinal causal inference techniques, 
we consider partitioning the functional temporal domain into three time periods: $[0, t_1)$, $[t_1, t_2]$, and $(t_2, T]$.  Our approach is to construct policies that only modify the treatment over $[t_1, t_2]$ so that it is not necessary to adjust for time-varying confounding within the intervals $[0, t_1)$ and $(t_2, T]$.  However, we do place some restrictions on the time-varying confounding within the time period of interest, $[t_1, t_2]$.
For our proposed estimands, these restrictions still allow for some temporal confounding to occur within this interval such that we do not require the strong assumption of no temporal confounding over $[t_1, t_2]$ (see Section \ref{sec:estimand_id}). 
Given these restrictions, it is only necessary to adjust for the time-varying confounding that occurs \textit{between} the three time periods; this is equivalent to the discrete time longitudinal treatment setting with three time points. 


\subsection{Deterministic MFTPs and stochastic policies}
Modified functional treatment policies (MFTPs) define a deterministic policy or rule, $d(\bm{X}, A(\cdot))$, that modifies the observed or natural value of the functional treatment \citep{jiang_estimating_2026}. One common example of such a policy is $d(\bm{X}, A(\cdot)) = cA(\cdot)$; the corresponding estimand,  $E[Y(d(\bm{X}, A(\cdot))]$ is the expected potential outcome if all individuals modify their physical activity by $100(c-1)$\%. 
The effect of a deterministic policy is most interpretable when it reflects a practical, real-world policy that could be implemented (e.g., increasing the dose of a drug by $c$). For our scenario, however, it is unrealistic to envision all individuals exactly adopting a deterministic modification to their physical activity, and so the causal effects of such policies can be difficult to interpret.
Furthermore, the goal of a policy that could be implemented in practice (e.g., recommending individuals increase their physical activity) is to modify an individuals general pattern of behavior (i.e., distribution of behavior), rather than achieve an exact, modified behavior. 
We therefore propose adopting stochastic, rather than deterministic, policies for functional treatments. Stochastic dynamic policies define a modified conditional distribution of treatment such that the policy modifies the treatment distribution rather than the treatment itself \citep{kennedy_nonparametric_2019,diaz_causal_2020,diaz_nonparametric_2023,schindl_incremental_2026}.
These policies can both implicitly satisfy positivity \citep{kennedy_nonparametric_2019,schindl_incremental_2026} and vary by covariates in a manner that reflects the observed joint treatment distribution; in contrast, the analyst has to \textit{a priori} specify a potentially complex deterministic function of the observed treatment and covariates for MFTPs. Since we are only interested in policies that vary with covariates, we refer to stochastic dynamic policies as stochastic policies.
Figure \ref{fig:toy_example} provides an illustrative example.
While functions do not have a probability density, for illustrative purposes, we consider defining a distribution $q(\cdot|\bm{X}, A^{(1)}(\cdot))$ from which to select the modified function over $[t_1, t_2]$ conditional on covariates and the function over $[0, t_1)$. 


\begin{figure}[h!]
 \centering
\begin{tikzpicture}
    \node[anchor=south west, inner sep=0pt] (base_image) at (0,0) {\includegraphics[width=1\textwidth]{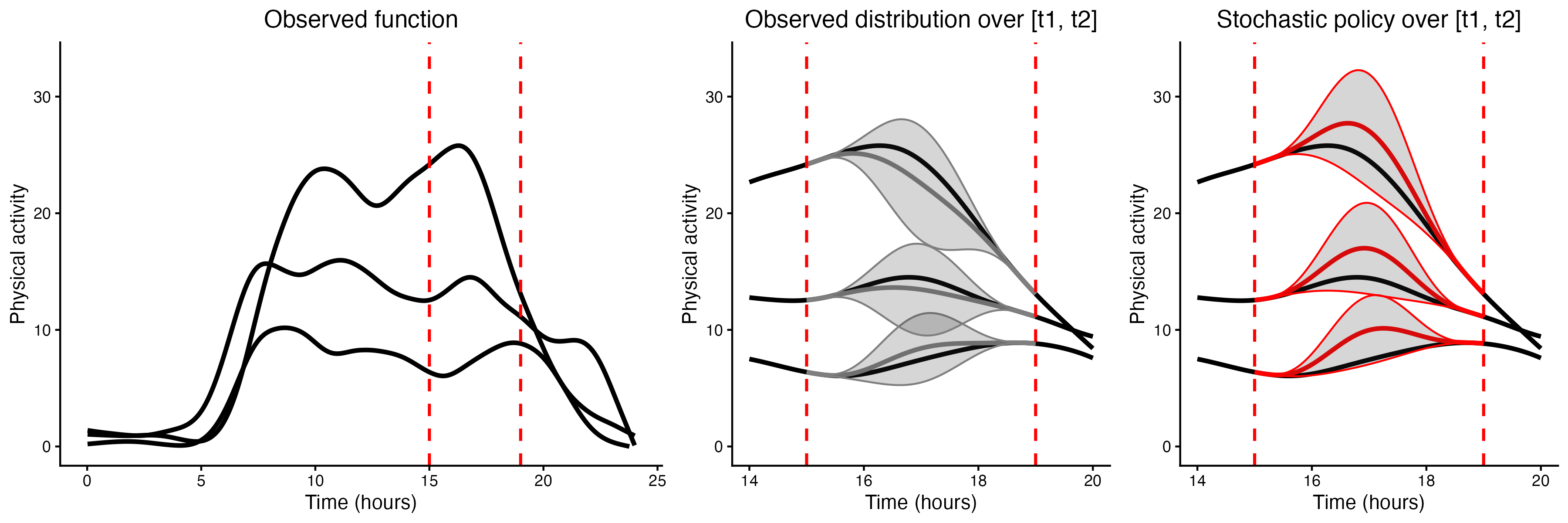}};
    \node[font=\tiny, anchor=north west, inner sep=0pt, xshift=3cm, yshift=-1.15cm] at (base_image.north west) {$A_1^{(1)}(\cdot)$};
    \node[font=\tiny, anchor=north west, inner sep=0pt, xshift=3cm, yshift=-2.1cm] at (base_image.north west) {$A_2^{(1)}(\cdot)$};
    \node[font=\tiny, anchor=north west, inner sep=0pt, xshift=3cm, yshift=-3.5cm] at (base_image.north west) {$A_3^{(1)}(\cdot)$};
    \node[font=\tiny, anchor=north west, inner sep=0pt, xshift=4.2cm, yshift=-0.9cm] at (base_image.north west) {$A_1^{(2)}(\cdot)$};
    \node[font=\tiny, anchor=north west, inner sep=0pt, xshift=4.2cm, yshift=-2.2cm] at (base_image.north west) {$A_2^{(2)}(\cdot)$};
    \node[font=\tiny, anchor=north west, inner sep=0pt, xshift=4.2cm, yshift=-3.5cm] at (base_image.north west) {$A_3^{(2)}(\cdot)$};
    \node[font=\tiny, anchor=north west, inner sep=0pt, xshift=5.5cm, yshift=-2.9cm] at (base_image.north west) {$A_1^{(3)}(\cdot)$};
    \node[font=\tiny, anchor=north west, inner sep=0pt, xshift=5.5cm, yshift=-3.5cm] at (base_image.north west) {$A_2^{(3)}(\cdot)$};
    \node[font=\tiny, anchor=north west, inner sep=0pt, xshift=5.1cm, yshift=-4.1cm] at (base_image.north west) {$A_3^{(3)}(\cdot)$};
    \node[font=\tiny, anchor=north west, inner sep=0pt, xshift=7.8cm, yshift=-3.9cm] at (base_image.north west) {$f\left(A_i^{(2)}(\cdot) \middle\vert A_i^{(1)}(\cdot)\right)$};
    \node[font=\tiny, anchor=north west, inner sep=0pt, xshift=12.1cm, yshift=-3.9cm] at (base_image.north west) {$q\left(A_i^{(2)}(\cdot) \middle\vert A_i^{(1)}(\cdot)\right)$};
\end{tikzpicture}
\caption{\footnotesize An illustrative example of stochastic policies for functional treatments. The left panel shows three observed physical activity function over 24-hours, where the red dotted line shows the time period for policy implementation (ie., $[t_1, t_2]$). The middle panel shows the observed function (black) and the expected value and 95\% bounds of the observed conditional distribution of treatment over $[t_1, t_2]$ (gray). The right panel shows the observed function (black) and the expected value and 95\% bounds of a stochastic policy distribution of treatment over $[t_1, t_2]$ (red). }
    \label{fig:toy_example}
\end{figure}
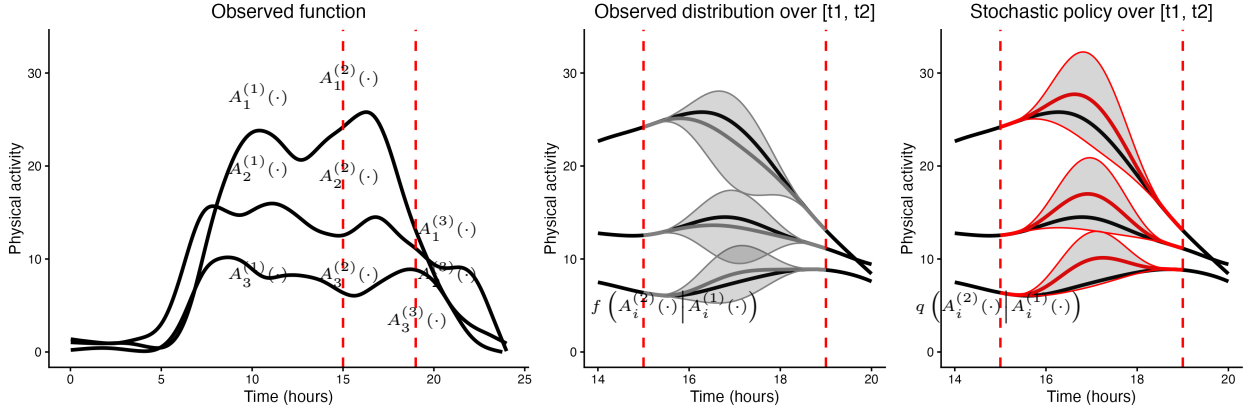

\subsection{Defining functional treatment estimands through basis approximation}
Common causal estimand identification strategies (e.g., inverse weighting by the conditional probability of receiving the observed treatment) cannot be directly applied to functional treatment estimands because they lack probability densities \citep{delaigle_defining_2010}.  It is further unclear how to construct well-defined stochastic policies for functional treatments. However, while infinite-dimensional functions do not have probability densities, their finite basis approximations do. 
Functional principal component analysis (FPCA) yields one such basis decomposition, $A(t) = a_0(t) +\sum_{j=1}^{\infty} \theta_j^{1/2}A_j\psi_j(t)$. In this decomposition,
$a_0(t)$ is the point-wise expectation of $A(t)$,  $\theta_j$ and $\psi_j(t)$ are the eigenvalues and eigenfunctions, and $A_j$
are random variables with $E[A_j] = 0$, $\Var(A_j) = 1$, and $\Cov(A_j, A_k) = 0$ for $j = 1, \ldots, \infty$, $j \neq k$. The finite approximation of $A(\cdot)$, $A_J(\cdot) = a_0(\cdot) + \sum_{j=1}^{J} \theta_j^{1/2}A_j\psi_j(\cdot)$, minimizes the total mean squared error (MSE) among all approximations of size $J$ and explains $\sum_{j=1}^J \theta_j / \sum_{j=1}^{\infty} \theta_j \%$ of the total variance in $A(\cdot)$. 
\cite{jiang_estimating_2026} demonstrate that functional treatment estimands can be defined in the limit as $J \to \infty$ of the joint distribution of $\{A_1, A_2, \ldots, A_J\}$ such that 
we could define a stochastic policy by modifying the conditional distribution of $\{A^{(2)}_1, A^{(2)}_2, \ldots, A^{(2)}_{J_2}\}$. However, it is unclear how to directly modify this distribution to yield interpretable changes to the functional treatment. Thus, we propose methods for defining the stochastic policy through a specific \textit{univariate} conditional distribution in the following sections.

\section{Causal estimands for functional treatments with stochastic policies}
\label{sec:methods}
We develop causal estimands corresponding to modifications to the functional treatment distribution. We first propose a novel procedure for functional basis construction in which one basis function can be deliberately chosen and modified by the analyst. 
We then define and identify two different classes of stochastic policies and estimands that account for temporal confounding. 
Finally, we discuss how to choose the basis function and stochastic policy distribution in practice.

\subsection{Novel functional basis construction for defining stochastic policies for functional treatments}
\label{sec:novel_basis}
We propose the following procedure which allows the analyst to explicitly choose a single basis function such that the stochastic policy can be define through only this selected function:  
\begin{enumerate}
    \item Compute the eigenfunctions $\{\psi_j(\cdot)\}_{j=1}^\infty$ from FPCA
    \item Define a function,  $\gamma_1(\cdot)$, that is continuous over $[0, T]$ and satisfies $\Vert\gamma_1\Vert_{2}= 1$
    \item Implement the Gram-Schmidt process on the set of functions $\{\gamma_1(\cdot), \psi_1(\cdot), \psi_2(\cdot), \ldots\}$ and obtain the orthonormal basis $\{\gamma_1(\cdot), \gamma_2(\cdot), \ldots \}$ 
\end{enumerate}
By Hilbert space theory for countably infinite spaces, $A(t) = a_0(t) + \sum_{j=1}^{\infty}B_j\gamma_j(t)$,
where $B_j = \int_{0}^{T} A(t)\gamma_j(t) dt$ is a random variable and $E[B_j] = 0$. In general, $\Cov(B_j, B_k) \neq 0$ and we do not obtain $\Var(B_j)$ or $\Cov(B_j, B_k)$ from the decomposition process as with FPCA.
To obtain a truncated basis from this procedure, we first select a $J$ such that $\{\psi_j(\cdot)\}_{j=1}^{J-1}$ explains a substantial proportion of the total variance of $A(\cdot) - a_0(\cdot)$ (e.g., $99\%$). Then, in step 3) we implement the Gram-Schmidt process on $\{\gamma_1(\cdot), \psi_1(\cdot), \ldots, \psi_{J-1}(\cdot)\}$, resulting in the orthonormal basis $\{\gamma_j(\cdot)\}_{j=1}^J$ which spans the same space as the $J-1$ eigenfunctions and $\gamma_1(\cdot)$. While intuitively, this seems to imply that the $\{\gamma_j(\cdot)\}_{j=1}^J$ approximation of $A(\cdot)$ is at least as good as the $J-1$ FPCA approximation this is result is not immediately clear since $\{B_j\}_{j=1}^\infty$ are correlated unlike the basis coefficients from FPCA. However, we show that the proportion of variance explained by $\{\gamma_j(\cdot)\}_{j=1}^J$ is greater than or equal to the proportion of variance explained by $\{\psi_j(\cdot)\}_{j=1}^{J-1}$ as stated in Theorem \ref{thm:novel_basis}.


\begin{theorem}
    Let $\bar{A}_{J-1}(\cdot)$ and $\tilde{A}_J(\cdot)$ be, respectively, the $\{\psi_j(\cdot)\}_{j=1}^{J - 1}$ basis (FPCA) and $\{\gamma_j(\cdot)\}_{j=1}^J$ basis approximations of $A(\cdot) - a_0(t)$. For any choice of $\gamma_1(\cdot)$, $\int_0^T \Var\{\tilde{A}_J(t)\}dt \geq \int_0^T \Var\{\bar{A}_{J-1}(t)\}dt$. 
    \label{thm:novel_basis}
\end{theorem}

 The proof of Theorem \ref{thm:novel_basis} is in supplementary Section \ref{sec:supp_theory}. With our method, the analyst can define the stochastic policy over $[t_1, t_2]$ in terms of a single function $\gamma^{(2)}_1$, which we refer to as the \textit{policy basis}; let $B^{(2)}_1$ be the basis coefficient corresponding to this function.  Our approach resolves several key technical and interpretation challenges in defining stochastic policies for functional treatments: 1) $\gamma^{(2)}_1(\cdot)$ can be explicitly selected such that modifications to this function represent realistic changes to the treatment; 2) continuity can be maintained over the entire temporal domain by choosing a $\gamma_1^{(2)}(\cdot)$ that satisfies $\gamma_1^{(2)}(t_1) = \gamma_1^{(2)}(t_2) = 0$; and 3) estimands and estimators can be defined through the univariate density of $B_1^{(2)}$. 

\subsection{Defining stochastic policies}
A majority of longitudinal causal treatment effect methods, which evaluate treatments measured over months or years, define a stochastic policy that conditions only on past treatment. 
However, when treatments are measured over 24 hours, conditioning on both the past and future treatments may be justifiable (e.g., one's evening plans may influence their physical activity in the morning). In what follows, we present two stochastic policies and estimands that reflect these different conceptual ways of accounting for time-varying confounding.
We define the stochastic policy through the basis coefficient associated with $\gamma_1^{(2)}(\cdot)$, $B_1^{(2)}$. 
For simplicity in notation, let $A^{-}(\cdot) = A(\cdot) - B_1^{(2)}\gamma_1^{(2)}(\cdot) = a_0(\cdot) + \sum_{j=1}^\infty \theta_j^{(1)^{1/2}}A_j^{(1)}\psi_j^{(1)}(\cdot) + \sum_{j=2}^\infty B^{(2)}_j\gamma_j^{(2)}(\cdot) + \sum_{j=1}^\infty \theta_j^{(3)^{1/2}}A_j^{(3)}\psi_j^{(3)}(\cdot)$ and $A^{-[0, t_2]}(\cdot) = a_0(\cdot) + \sum_{j=1}^\infty \theta_j^{(1)^{1/2}}A_j^{(1)}\psi_j^{(1)}(\cdot) + \sum_{j=2}^\infty B^{(2)}_j\gamma_j^{(2)}(\cdot)$. Let $A^{-}_{\bm{J}}(\cdot)$ and $A^{-[0, t_2]}_{\bm{J}}(\cdot)$ be the finite basis approximations of $A^{-}(\cdot)$ and $A^{-[0, t_2]}(\cdot)$ for $\bm{J} = (J_1, J_2, J_3)$. Let $Q_1^{(2)} \sim q(\cdot|\bm{X}, A^{-}(\cdot))$ and $Q^{(2)*}_1 \sim q^*(\cdot|\bm{X},  A^{-[0, t_2]}(\cdot))$ where $q$ and $q^*$ are user-specified distributions. We define the following treatments: corresponding to $q$ and $q^*$:
\begin{align*}
    A^Q(\cdot) &= A^{-}(\cdot) + Q^{(2)}_1\gamma_1^{(2)}(\cdot), \\
    A^{Q^*}(\cdot) &= A^{-[0, t_2]}(\cdot) + Q^{(2)*}_1\gamma_1^{(2)}(\cdot) + \sum_{j=1}^\infty \theta_j^{(3)^{1/2}}A_j^{(3)}\left(Q^{(2)*}_1\right)\psi_j^{(3)}(\cdot) = A^{-}(\cdot) + Q^{(2)*}_1\gamma_1^{(2)}(\cdot). 
\end{align*}
$A^Q(\cdot)$ accounts for time-varying confounding through conditioning on $A^-(\cdot)$ in $q$. $A^{Q^*}(\cdot)$ accounts for time-varying confounding through conditioning on $A^{-[0, t_2]}(\cdot)$ in $q^*$ and via $\left\{A_j^{(3)}\left(Q^{(2)*}_1\right)\right\}_{j=1}^\infty$, the counterfactual FPCA basis coefficients for the functional treatment over $(t_2, T]$; the second equality for $A^{Q^*}(\cdot)$ holds since by definition,  $Q^{(2)*}_1 \indep \{A_j^{(3)}\}_{j=1}^\infty$ and thus  $A_j^{(3)}(Q^{(2)*}_1) = A_j^{(3)}$.
This is a specific property of stochastic policies in general; under a deterministic policy, the counterfactual future treatment would not necessarily be equal to the observed future treatment. 
Thus, the only difference between $A^Q(\cdot)$ and $A^{Q^*}(\cdot)$ is the conditioning set of the distributions $q$ and $q^*$. 
Since  $A^Q(\cdot)$ and $A^{Q^*}(\cdot)$ take similar forms, we simplify notation by considering an arbitrary stochastic policy distribution, $\tilde{q}$, where  $\tilde{Q}_1^{(2)} \sim \tilde{q}(\cdot | \bm{X}, A^{-\tilde{Q}})$ and $A^{\tilde{Q}}(\cdot) = A^{-}(\cdot) + \tilde{Q}_1^{(2)}\gamma_1^{(2)}(\cdot)$. 


\subsection{Estimand definition and identification}
\label{sec:estimand_id}
Due to infinite-dimensional functions lacking probability densities, it is necessary to show that functional treatment estimands are well-defined \citep{delaigle_defining_2010}.
Thus, prior to defining and identifying stochastic policy estimands, we show that the population average over a functional treatment can be defined in terms of the basis approximations over $[0, t_1)$, $[t_1, t_2]$, and $(t_2, T]$ for both the FPCA basis and our novel basis. We specifically focus on the functional $m(\bm{x}, a(\cdot)) = E(Y|\bm{X} = \bm{x}, A(\cdot) = a(\cdot))$. Let the population average given the finite FPCA approximation over each of $[0, t_1)$, $[t_1, t_2]$, and $(t_2, T]$ be defined as $  \mu_{\bm{J}} = E\left\{m\left(\bm{X}, a_0(\cdot) +  \sum_{k=1}^3\sum_{j=1}^{J_k} \theta^{(k)^{1/2}}_jA_j^{(k)}\psi_j^{(k)}\right)\right\}$.
Then, let $E\{m\left(\bm{X}, A(\cdot)\right)\} := \lim_{\bm{J} \to \infty} \mu_{\bm{J}}$ be the population average and  $\epsilon := Y - m\left(\bm{X}, A(\cdot)\right)$. We can make the comparable definitions for our proposed basis approximation; let $\tilde{\mu}_{\bm{J}} = E\left\{m\left(\bm{X}, a_0(\cdot) +  \sum_{k=1}^3\sum_{j=1}^{J_k} B_j^{(k)}\psi_j^{(k)}\right)\right\}$.
For the population average to be well-defined through these approximations, $\lim_{\bm{J} \to \infty} \mu_{\bm{J}}$  and $\lim_{\bm{J} \to \infty} \tilde{\mu}_{\bm{J}}$ must exist. When these limits are exist, $E[Y]$, $\lim_{\bm{J} \to \infty} \mu_{\bm{J}}$, and $\lim_{\bm{J} \to \infty} \tilde{\mu}_{\bm{J}}$ are each representations of the population average that naturally should align such that $E[Y] = \lim_{\bm{J} \to \infty} \mu_{\bm{J}} = \lim_{\bm{J} \to \infty} \tilde{\mu}_{\bm{J}}$. We require the following conditions to guarantee that these limits are well-defined:
\begin{itemize}
    \item (\textbf{C1}) The eigenvalues of $A(\cdot)$ over $[0, t_1)$, $[t_1, t_2]$, and $(t_2, T]$  converge to $0$ such that $\int \Var\{A^{(k)} (t)\}dt = \sum_{j=1}^\infty \theta_j^{(k)} < \infty$ for $k = 1, 2, 3$
    \item (\textbf{C2}) There exists a constant $C$ such that for two square-integrable functions over $[0, T]$, $a_1(\cdot)$ and $a_2(\cdot)$, the function $m$ satisfies $|m\{\bm{x}, a_1(\cdot)\} -  m\{\bm{x}, a_2(\cdot)\}| \leq C\Vert a_1(\cdot) - a_2(\cdot)\Vert_2$ for all $\bm{x}$
\end{itemize}
The assumption that $\sum_{j=1}^\infty \theta_j < \infty$ is commonly made when decomposing $A(\cdot)$ over the entire temporal domain and it in general holds when $A(\cdot)$ has a square-integrable kernel over $[0, T]$ \citep{hall_properties_2006}. In contrast, condition \textbf{C1} is a slightly weaker assumption as it generally holds when $A^{(k)}(\cdot)$ has a square-integrable kernel for $k=1,2,3$, i.e., when $A(\cdot)$ is locally square integrable over these three time periods. Condition \textbf{C2} is a Lipschitz condition on $m(\cdot)$.
Given these conditions, we show that the limits of both $\mu_{\bm{J}}$ and $\tilde{\mu}_{\bm{J}}$ are well defined and equivalent to $E[Y]$ as stated in Lemma \ref{thm:pop_avg_limit} and Corollary \ref{thm:corollary_pop_avg}.


Next, we define and identify the stochastic policy estimand  $\mu^{\tilde{Q}} = E[Y(A^{\tilde{Q}}(\cdot))]$ for $\tilde{q} = q^*, q$. Similar to \cite{schindl_incremental_2026}, we assume that $B_1^{(2)}$ has bounded support, $\mathcal{B}_1^{(2)} = [-M, M]$, for some $M < \infty$. When possible, we define conditions and assumptions in terms of $\tilde{q} = q, q^*$.
To identify this estimand with the observed data we make the following causal assumptions:
\begin{itemize}
    \item (\textbf{A1}: Consistency) If $A(\cdot) = a(\cdot)$, then $Y = Y\{a(\cdot)\}$
    \item  (\textbf{A$\tilde{\bm{2}}$}: Positivity) $B_1^{(2)}|\bm{x}, a^{-\tilde{Q}}(\cdot)$ and $Q_1^{(2)}|\bm{x}, a^{-\tilde{Q}}(\cdot)$ have the same support for all $\bm{x}, a^{-\tilde{Q}}(\cdot)$
    \item (\textbf{A3}: Ignorability) $Y\{a(\cdot)\} \indep B_1^{(2)} | \bm{X}, A^{-}(\cdot)$
    \item (\textbf{A$\bm{3}^*$}: Ignorability) $(Y\{a(\cdot)\} , A^{(3)}(\cdot)) \indep B_1^{(2)}| \bm{X}, A^{-[0, t_2]}(\cdot)$
\end{itemize}
Assumption \textbf{A1} is a standard causal assumption which ensures that $Y(A(\cdot)) = Y$.
Assumption \textbf{A$\tilde{\bm{2}}$} requires the modification to the observed function, $\tilde{Q}_1^{(2)}$, to be realistic conditional on $\bm{X}, A^{-\tilde{Q}}(\cdot)$. We can choose $\tilde{q} = q, q^*$ such that \textbf{A$\tilde{\bm{2}}$} is guaranteed to hold \citep{schindl_incremental_2026}; we discuss this further in Section \ref{sec:choosing_q}. Since, the conditioning set of $q$ is larger than that of $q^*$, there will generally be a smaller set of distributional shifts that satisfy positivity for $q$ compared to $q^*$.
For $q$, Assumption \textbf{A3} states that the potential outcome is independent of $B_1^{(2)}$ given covariates $\bm{X}$ and $A^{-}(\cdot)$.
In contrast, for $q^*$, assumption \textbf{A$\bm{3}^*$} states that the joint distribution of the potential outcome and the function over $(t_2, T]$ is independent of $B_1^{(2)}$ conditional on $\bm{X}$ and $A^{-[0, t_2]}(\cdot)$. This is a substantially stronger assumption than \textbf{A3} as the conditioning set is smaller and further requires $B_1^{(2)}$ to be conditionally independent of  $A^{(3)}(\cdot)$.

In Section \ref{sec:more_lit_review}, we discuss making restrictions on the type of time-varying confounding that occurs over $[t_1, t_2]$. 
Given that the stochastic policy only applies to $B_1^{(2)}$, we can allow for time-varying confounding within $[t_1, t_2]$ for $A^{-}(\cdot)$. Thus, we only require the assumption that there is no time-varying confounding within $[t_1, t_2]$ for  $B_1^{(2)}\gamma_1^{(2)}(\cdot)$ and between $B_1^{(2)}\gamma_1^{(2)}(\cdot)$ and $A^{-}(\cdot)$. The choice of $\gamma_1^{(2)}(\cdot)$ changes the strength of this assumption; in general, it is difficult to ascertain to what extent this assumption holds for different choices in $\gamma_1^{(2)}(\cdot)$. However, for any $\gamma_1^{(2)}(\cdot)$, this is a substantial relaxation of no time-varying confounding within $[t_1, t_2]$. 
To simplify notation, let $E_{\tilde{Q}}\{m(\bm{X}, A(\cdot))|\bm{X},  A^{-}(\cdot)\} = \int_{\mathcal{B}_1^{(2)}} m(\bm{X}, b_1^{(2)}\gamma_1^{(2)}(\cdot) + A^{-}(\cdot))d\tilde{Q}(b_1^{(2)}|\bm{X}, A^{-\tilde{Q}}(\cdot))$ and $E_{\tilde{Q}}\{m(\bm{X}, A_{\bm{J}}(\cdot))|\bm{X},  A_{\bm{J}}^{-}(\cdot)\} $ \\ $= \int_{\mathcal{B}_1^{(2)}} m(\bm{X}, b_1^{(2)}\gamma_1^{(2)}(\cdot) + A_{\bm{J}}^{-}(\cdot))d\tilde{Q}(b_1^{(2)}|\bm{X}, A_{\bm{J}}^{-\tilde{Q}}(\cdot))$. This expectation is marginalized over $B_1^{(2)}$ for $\tilde{q} = q, q^*$ such that this expectation is conditional on $A^{-}(\cdot)$, rather than $A^{-\tilde{Q}}(\cdot)$. 
To ensure the causal identification of $E[Y(A^{\tilde{Q}}(\cdot))]$ is well defined, we further require the following conditions:
\begin{itemize}
    \item (\textbf{C$\tilde{\bm{3}}$}) There exists a constant $C_{\tilde{Q}}$ such that for two square-integrable functions over $[0, T]$, $a_1(\cdot)$ and $a_2(\cdot)$, the function $m$ satisfies $\left\vert E_{\tilde{Q}}\left\{m\left(\bm{x}, a(\cdot) + B_1^{(2)}\gamma_1^{(2)}(\cdot) \right)|\bm{x}, a_1(\cdot)\right\} \right. - \left. E_{\tilde{Q}}\left\{m\left(\bm{x}, a(\cdot) + B_1^{(2)}\gamma_1^{(2)}(\cdot) \right)|\bm{x}, a_2(\cdot)\right\} \right\vert \leq C_{\tilde{Q}}||a_1(\cdot) - a_2(\cdot)||_2 $ for all $\bm{x}$ and $a(\cdot)$
    \item (\textbf{C$\tilde{\bm{4}}$}) There exist constants $C_{\tilde{Q}1}$ and $C_{\tilde{Q}2}$ such that $E\left\{ \frac{\tilde{q}(B_1^{(2)}|\bm{X}, A_{\bm{J}}^{-\tilde{Q}}(\cdot))}{f(B_1^{(2)}| \bm{X}, A_{\bm{J}}^{-\tilde{Q}}(\cdot))}\right\} \leq C_{\tilde{Q}1}$ and $\Var\left\{ \frac{\tilde{q}(B_1^{(2)}|\bm{X}, A_{\bm{J}}^{-\tilde{Q}}(\cdot))}{f(B_1^{(2)}| \bm{X}, A_{\bm{J}}^{-\tilde{Q}}(\cdot))}\right\} \leq C_{\tilde{Q}2}$ for all $J_1, J_2, J_3 \in \mathbb{N}$
\end{itemize}
Condition \textbf{C$\tilde{\bm{3}}$} is similar to condition \textbf{C2}; it is a Lipschitz condition on $m(\cdot)$ marginalized over $\tilde{q}$. Condition \textbf{C$\tilde{\bm{4}}$} states that the expectation and variance of the conditional density ratios for $\tilde{q}$ are bounded for all $\bm{J}$. 
Outcome model and weighting identifications are presented in supplementary Theorems \ref{thm:outcome_reg} and \ref{thm:weighting}; in Theorem \ref{thm:dr_identification}, we present the augmented identification that aligns with the efficient influence function for continuous treatment stochastic policies derived in \cite{diaz_causal_2020} and \cite{schindl_incremental_2026}.
\begin{theorem}
\textnormal{(Augmented identification)} Given conditions \textbf{C1}-\textbf{C2}, \textbf{C$\tilde{\bm{3}}$}-\textbf{C$\tilde{\bm{4}}$} and assumptions \textbf{A1}, \textbf{A$\tilde{\bm{2}}$}, and assumption \textbf{A3} (for $\tilde{q} = q$)  or \textbf{A$\bm{3}^*$} (for $\tilde{q} = q^*$) 
\begin{align*}
    \mu^{\tilde{Q}}  &= \lim_{\bm{J} \to \infty} E\left(\frac{\tilde{q}(B_1^{(2)}|  \bm{X}, A_{\bm{J}}^{-\tilde{Q}}(\cdot))}{f(B_1^{(2)}|  \bm{X}, A_{\bm{J}}^{-\tilde{Q}}(\cdot))}\left[Y - E_{\tilde{Q}}\{m(\bm{X}, A_{\bm{J}}(\cdot))|\bm{X},  A_{\bm{J}}^{-}(\cdot)\}\right] \right.\\
    &\quad+ \left. E_{\tilde{Q}}\{m(\bm{X}, A_{\bm{J}}(\cdot))|\bm{X},  A_{\bm{J}}^{-}(\cdot)\}\right),\:\:\: \text{for}\:\: \tilde{q} = q, q^*.
\end{align*}
\label{thm:dr_identification}
\end{theorem}
The proof of Theorem \ref{thm:dr_identification} is in supplementary Section \ref{sec:supp_theory}. We note that the conditional expectation term, $E_{\tilde{Q}}$, is a function of both $m(\cdot)$ and $\tilde{q}$; this differs from standard outcome model identifications for deterministic treatments which are only a function of $m(\cdot)$. This identification also contains the conditional density ratio between $\tilde{q}$ and the observed conditional distribution of $B_1^{(2)}$, $f$. However, since the stochastic policy applies to a single basis coefficient this conditional density is univariate rather than multivariate as in  \cite{jiang_estimating_2026}, making nuisance parameter estimation more tractable.

\subsection{Choosing the stochastic policy distribution and policy basis}
\label{sec:choosing_q}
We have focused on identifying the estimand $\mu^{\tilde{Q}}$ for an arbitrary stochastic policy, $\tilde{q}$; however, we need to choose this distribution in practice. 
\cite{diaz_causal_2020} and \cite{schindl_incremental_2026} propose  stochastic policies defined as exponential tilts of the observed treatment distribution. 
Specifically, for $\delta \in \mathbb{R}$, the exponential tilt of the observed conditional treatment distribution, $f(b_1^{(2)}|\bm{x}, a^{-\tilde{Q}}(\cdot))$, is
\begin{align}
    \tilde{q}_\delta(b_1^{(2)}|\bm{x}, a^{-\tilde{Q}}(\cdot)) &= \frac{\exp(\delta b_1^{(2)})f(b_1^{(2)}|\bm{x},a^{-\tilde{Q}}(\cdot))}{E\left\{\exp(\delta B_1^{(2)})|\bm{X} = \bm{x},A^{-\tilde{Q}}(\cdot) = a^{-\tilde{Q}}(\cdot))\right\}},  \label{eq:exp_tilt_q}
\end{align}
where $\delta = 0$ yields $\tilde{q}_\delta(b_1^{(2)}|\bm{x},a^{-\tilde{Q}}(\cdot)) = f(b_1^{(2)}|\bm{x},a^{-\tilde{Q}}(\cdot))$. 
A positive $\delta$ corresponds to an increase in the likelihood of receiving higher values of $B_1^{(2)}$, while a negative $\delta$ corresponds to a decrease in the likelihood of receiving higher values of $B_1^{(2)}$. The size of the shift in the likelihood is determined both by $\delta$ and the conditional variance of $B_1^{(2)}$; the larger the conditional variance, the smaller the $\delta$ values that result in moderate shifts in the likelihood. 
Positivity is satisfied for any $\delta$ because $f(b_1^{(2)}|\bm{x},a^{-\tilde{Q}}(\cdot)) = 0$ implies $\tilde{q}_\delta(b_1^{(2)}|\bm{x},a^{-\tilde{Q}}(\cdot)) = 0$ by definition.
Thus, we can estimate the causal effect of a variety of stochastic policies simply through smoothly changing $\delta$, though there may be practical estimation challenges for some $\delta$. 
Further properties of exponentially tilted distributions are discussed by \cite{schindl_incremental_2026}.

We also need to choose the policy basis through which the function is modified under the policy. For some research questions, there may be functional form of specific interest.
However, we propose choosing the policy basis from $\gamma_{c,d}(t) = \left(\frac{t-t_1}{t_2-t_1}\right)^c\left(\frac{t_2 -t}{t_2-t_1}\right)^d/\left\Vert \left(\frac{t-t_1}{t_2-t_1}\right)^c\left(\frac{t_2 -t}{t_2-t_1}\right)^d \right\Vert_{2} $ for $ c,d > 0$, 
where $\gamma_{c,d}(t_1) = \gamma_{c,d}(t_2) = 0$ for all $c,d > 0$. 
Since $\gamma_{c,d}(\cdot)$ is convex, $\delta > 0$ corresponds to an increase in the likelihood of higher values of the function over $[t_1, t_2]$.
As $c,d$ increase, $\gamma_{c,d}(\cdot)$ becomes more convex. The relationship between $c,d$ defines where $\gamma_{c,d}(\cdot)$ is maximized; when $c=d$, $\gamma_{c,d}(\cdot)$ is maximized at $(t_1+t_2)/2$. While we can choose a single policy basis from this set, as we do in Sections \ref{sec:simulation4}-\ref{sec:nhanes_application}, 
we could also perform a sensitivity analysis for the choice of policy basis through varying $c,d$ values.
Figure \ref{fig:nhanes_walkthrough} presents an example of a single physical activity trajectory from NHANES and the corresponding policy basis, $q_\delta$, and resulting expected modified function for three $\delta$ values.

\begin{figure}[h]
    \centering
    \includegraphics[width=0.9\linewidth]{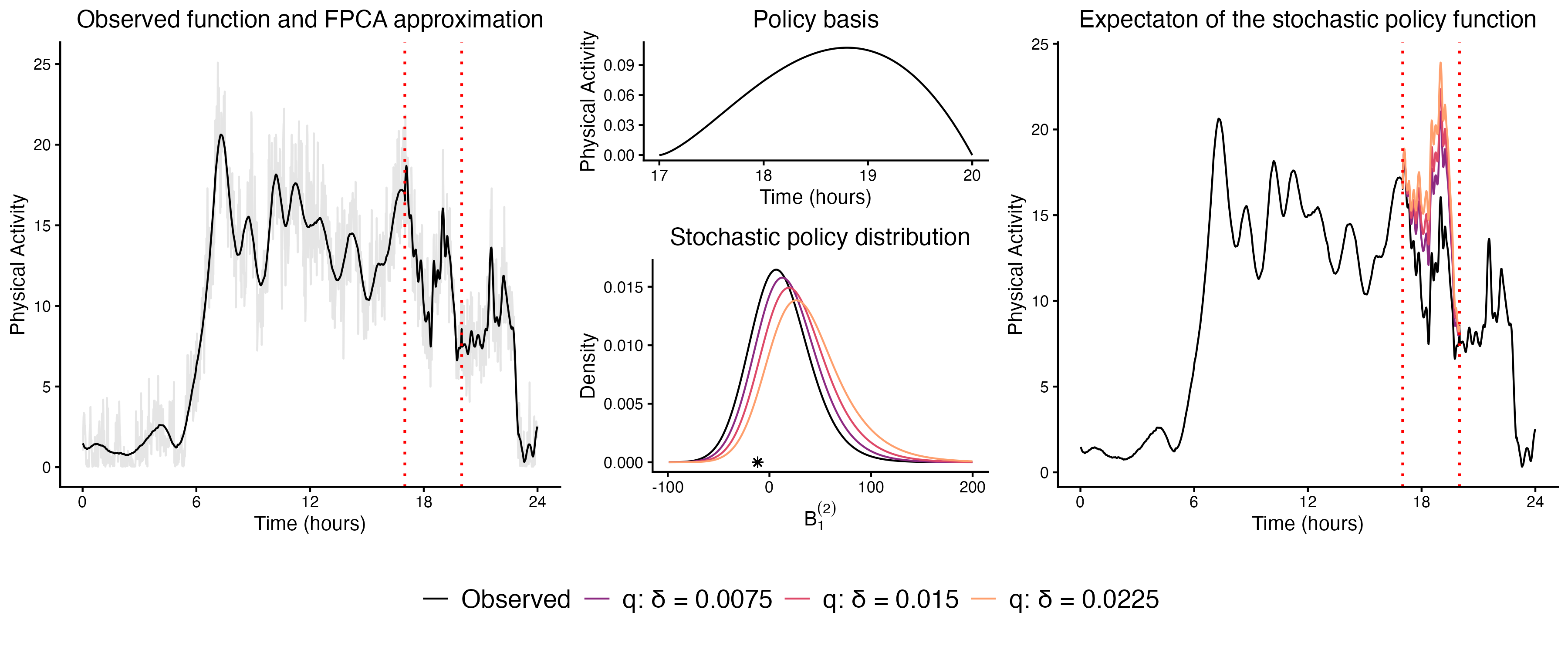}
    \caption{\footnotesize A single physical activity trajectory from NHANES and the corresponding policy basis, stochastic policy distribution, and expected functions under three stochastic policies. }
    \label{fig:nhanes_walkthrough}
\end{figure}


\section{Estimation}
\label{sec:estimation}
We propose the following plug-in estimator of the augmented estimand identification presented in Theorem \ref{thm:dr_identification}:
\begin{align*}
    \hat{\mu}^{\tilde{Q}(\delta)}_{\bm{J}}&= \frac1n\sum_{i=1}^n \frac{\hat{\tilde{q}}_\delta(B_{1,i}^{(2)}|\bm{X}_i, A_{\bm{J},i}^{-\tilde{Q}}(\cdot))}{\hat{f}(B_{1,i}^{(2)}|\bm{X}_i, A_{\bm{J},i}^{-\tilde{Q}}(\cdot))}\left\{Y_i \right.\\
    &\quad \left.- \int_{\mathcal{B}_{1}^{(2)}}\hat{m}\left(\bm{X}_i, b_1^{(2)}\gamma_1^{(2)}(\cdot) + A_{\bm{J},i}^{-}(\cdot)\right)\hat{\tilde{q}}_\delta\left(b_1^{(2)}|\bm{X}_i, A_{\bm{J},i}^{-\tilde{Q}}(\cdot))\right)db_1^{(2)}\right\} \\
    &\quad + \int_{\mathcal{B}_{1}^{(2)}}\hat{m}\left(\bm{X}_i, b_1^{(2)}\gamma_1^{(2)}(\cdot) + A_{\bm{J},i}^{-}(\cdot)\right)\hat{\tilde{q}}_\delta\left(b_1^{(2)}|\bm{X}_i, A_{\bm{J},i}^{-\tilde{Q}}(\cdot))\right)db_1^{(2)}. 
\end{align*}
In this estimator, $\hat{f}$ and $\hat{\tilde{q}}_{\delta}$ are estimates of the conditional densities and $\hat{m}$ is an estimate of the conditional expectation function. Due to the form of $\hat{\tilde{q}}_{\delta}$, it suffices to estimate $\hat{f}$ and then use numerical integration to obtain an estimate of the denominator of Equation \eqref{eq:exp_tilt_q}. The nuisance functions $m$ and $f$ must be cross-fit for both our theoretical results and estimator performance \citep{chernozhukov_doubledebiased_2018}; see supplementary Section \ref{sec:supp_cross_fit} and Algorithm 1 of \cite{schindl_incremental_2026} for details on computing the estimator with cross-fitting. We additionally recommend normalizing the conditional density ratio to reduce finite-sample estimator MSE.

\cite{schindl_incremental_2026} show that their estimator converges at a rate of $(n/\delta)^{-1/2}$
such that the effect sample size is $n/\delta$. 
For our application of interest, $\delta << 1$ (see Section \ref{sec:nhanes_application}) and thus the effective sample size is in fact larger than $n$. We derive that our proposed estimator, $\hat{\mu}^{\tilde{Q}(\delta)}_{\bm{J}}$ has the same convergence rate given some additional conditions to those required by \cite{schindl_incremental_2026} for $\tilde{q} = q, q^*$. Let $\epsilon_{\bm{J}}^f =  \hat{f}(b_{1}^{(2)}|\bm{X}, A_{\bm{J}}^{-\tilde{Q}}(\cdot)) - f(b_{1}^{(2)}|\bm{X}, A_{\bm{J}}^{-\tilde{Q}}(\cdot))$ and $\epsilon_{\bm{J}}^m = \hat{m}(\bm{X}, b_1^{(2)}\gamma_j^{(2)}(\cdot) + A_{\bm{J}}^{-}(\cdot)) - m(\bm{X}, b_1^{(2)}\gamma_j^{(2)}(\cdot) + A_{\bm{J}}^{-}(\cdot))$ be the nuisance function estimation errors where $\hat{f}$ and $\hat{m}$ are cross-fit. First, we outline the analogous conditions and assumptions to those in \cite{schindl_incremental_2026} for $\bm{J}$: 
\begin{itemize}
       \item (\textbf{A$\tilde{\bm{4}}$}: Weak positivity) Assume $\min\{\hat{f}(b_{1}^{(2)}| \bm{x},a_{\bm{J}}^{-\tilde{Q}}(\cdot)), f(b_{1}^{(2)}| \bm{x},a_{\bm{J}}^{-\tilde{Q}}(\cdot))\} \geq \pi_{\min}$ for $\pi_{\min} > 0$ for all $\bm{x},a_{\bm{J}}^{-\tilde{Q}}(\cdot)$, and  $b_1^{(2)} \in \cup_{h=1}^H [c_h, d_h]$ and $d_h - c_h \geq L > 0$ for $h = 1, \ldots H$ 
       \item (\textbf{A$\tilde{\bm{5}}$}: Consistency of models) $\left\Vert \epsilon^f_{\bm{J}} \right\Vert_{L_{b^{(2)}_1}^\infty, L^2} = o_p(1)$ and $\left\Vert \epsilon^m_{\bm{J}} \right\Vert_{L_{b^{(2)}_1}^\infty, L^2} = o_p(1)$ 
       \item (\textbf{A$\tilde{\bm{6}}$}: Rate of models) $\left\Vert \epsilon^f_{\bm{J}} \right\Vert_{L_{b^{(2)}_1}^\infty, L^2} $  $\times \left\Vert \epsilon^m_{\bm{J}} \right\Vert_{L_{b^{(2)}_1}^\infty, L^2} $ $+ \left\Vert \epsilon^f_{\bm{J}} \right\Vert^2_{L_{b^{(2)}_1}^\infty, L^2} = o_p((n/\delta)^{-1/2})$
       \item (\textbf{C$\tilde{\bm{5}}$}: Bounded treatment density) $f(b_{1}^{(2)}|\bm{x}, a_{\bm{J}}^{-\tilde{Q}}(\cdot))) \leq \pi_{\max}$ for all $\bm{x}, a_{\bm{J}}^{-\tilde{Q}}(\cdot),b_{1}^{(2)}$ 
       \item (\textbf{C$\bm{6}$}: Outcome noise) $0 < \sigma^2_{\min} \leq \Var(Y|\bm{X} = \bm{x}, A_{\bm{J}}(\cdot) =  a_{\bm{J}}(\cdot))$ for all $\bm{x}, a_{\bm{J}}(\cdot)$ 
       \item (\textbf{C$\bm{7}$}: Bounded outcome) $P(|Y| \leq B) = 1$ for some $B < \infty$
       \item (\textbf{C$\bm{8}$}): $\sqrt{n/\delta} \to \infty$
   \end{itemize}
where $\Vert h_{\bm{J}}\Vert^2_{L_{b^{(2)}_1}^\infty, L^2} = \int \left(\sup_{b_{1}^{(2)}} \left\vert h(b_{1}^{(2)}, \bm{x}, a_{\bm{J}}^{-\tilde{Q}}(\cdot))\right\vert\right)^2dF(\bm{x}, a_{\bm{J}}^{-\tilde{Q}}(\cdot))$  is a mixed $L_2$-sup norm for an arbitrary function $h$. While positivity is not required for identification, we require a weak positivity assumption  (assumption \textbf{A$\tilde{\bm{4}}$}) for the asymptotic results as in \cite{schindl_incremental_2026}; under this assumption, there can exist any finite number of places in the density where $f(b_{1}^{(2)}| \bm{x}, a_{\bm{J}}^{-\tilde{Q}}(\cdot))) = 0$. Assumptions \textbf{A$\tilde{\bm{5}}$} and \textbf{A$\tilde{\bm{6}}$} assume that $\hat{f}$ and $\hat{m}$ are consistent and that $\hat{f}$ and $\hat{m}$ converge at a rate of $o_p((n/\delta)^{-1/4})$ or faster. Since we assume both $\hat{f}$ and $\hat{m}$ are consistent, our proposed estimator is not doubly robust in terms of consistency. However, assumption \textbf{A$\tilde{\bm{6}}$} allows for flexible machine learning models to be used for estimating $\hat{f}$ and $\hat{m}$; universally consistent methods that achieve these rates satisfy \textbf{A$\tilde{\bm{5}}$} regardless of the complexity of $f$ or $m$. 
Conditions \textbf{C$\tilde{\bm{5}}$} and \textbf{C$\bm{6}$}-\textbf{C$\bm{8}$} are technical conditions that ensure that the treatment density and outcome are bounded, there is non-zero conditional outcome variance, and $n/\delta$ goes to infinity.

A key difference between our estimation problem that of \cite{schindl_incremental_2026} is that the nuisance functions in our estimator are defined in terms of an approximation of the treatment, $A_{\bm{J}}(\cdot)$; the residual variance of this approximation is upper bounded by $\sum_{k=1}^3\sum_{j=J_k}^\infty \theta_j^{(k)}$.  To extend the results of \cite{schindl_incremental_2026},
we consider that $\bm{J}$ increases with $n$ through the relationship $\bm{J}(n)$. Then, we can define the bound on the residual variance as a function of $n$, $\Delta_{\bm{J}(n)} = \sum_{k=1}^3\sum_{j=J_k(n)}^\infty \theta_j^{(k)}$. Given the relationship $\bm{J}(n)$, we have an additional assumption on the rate of $\Delta_{\bm{J}(n)}$: 
\begin{itemize}
    \item (\textbf{A7}: Rate of residual variance) $\Delta_{\bm{J}(n)} = o(\delta/n)$
\end{itemize}
See \cite{jiang_estimating_2026} for details on which $\bm{J}(n)$ satisfy this assumption for a variety of function classes. To simplify notation, let $w^{\tilde{Q}(\delta)}(\bm{X}, A_{\bm{J}(n)}^{-\tilde{Q}(\delta)}(\cdot)) = \frac{\tilde{q}_\delta(B_{1}^{(2)}|\bm{X}, A_{\bm{J}(n)}^{-\tilde{Q}(\delta)}(\cdot))}{f(B_{1}^{(2)}|\bm{X}, A_{\bm{J}(n)}^{-\tilde{Q}(\delta)}(\cdot))}$,  $v(\bm{X}, A_{\bm{J}(n)}(\cdot)) = \Var(Y|\bm{X}, A_{\bm{J}(n)}(\cdot))$ and $d^{\tilde{Q}(\delta)}(\bm{X}, A_{\bm{J}(n)}(\cdot)) = m(\bm{X}, A_{\bm{J}(n)}(\cdot)) $ $- E_{\tilde{Q}(\delta)}\{m(\bm{X}, A_{\bm{J}(n)}(\cdot))|\bm{X}, A_{\bm{J}(n)}^{-}(\cdot)\}$. Then, by the Lindeberg-Feller Central Limit Theorem our proposed estimators achieve asymptotic normality and root-$n/\delta$ consistency as stated in Theorem \ref{thm:estimator_asymp}.



\begin{theorem}
   Assume the conditions and assumptions of Theorem \ref{thm:dr_identification}, assumption \textbf{A7}, and conditions  \textbf{C7}-\textbf{C8} hold. Given that assumptions \textbf{A$\tilde{\bm{4}}$}-\textbf{A$\tilde{\bm{6}}$} and conditions \textbf{C$\tilde{\bm{5}}$}-\textbf{C$\bm{6}$} hold for all $\bm{J}$, $
   \frac{\sqrt{n}}{\sigma^{\tilde{Q}(\delta)}_{\bm{J}(n)}}( \hat{\mu}^{\tilde{Q}(\delta)}_{\bm{J}(n)} - \mu^{\tilde{Q}(\delta)}) \to N(0,1)$,
    where $\sigma^{\tilde{Q}(\delta)^2}_{\bm{J}(n)} = E\left[w^{\tilde{Q}(\delta)}(\bm{X}, A_{\bm{J}(n)}^{-\tilde{Q}(\delta)}(\cdot))^2\left\{v(\bm{X}, A_{\bm{J}(n)}(\cdot)) \right.\right.$ $\left.\left.+  d^{\tilde{Q}(\delta)}(\bm{X}, A_{\bm{J}(n)}(\cdot))^2 \right\}\right] $ $+ \Var\left[E_{\tilde{Q}(\delta)}\{m(\bm{X}, A_{\bm{J}(n)}(\cdot))|\bm{X}, A_{\bm{J}(n)}^{-}(\cdot)\}\right]$ for $\tilde{q} = q, q^*$.
    \label{thm:estimator_asymp}
\end{theorem}

All proofs of results in this section are in supplementary Section \ref{sec:supp_theory}. The asymptotic variance of our proposed estimator aligns with the nonparametric efficiency bound derived in \cite{schindl_incremental_2026}. In this asymptotic variance equation, the term that varies the most with $\delta$ is $d^{\tilde{Q}(\delta)}(\bm{X}, A_{\bm{J}(n)}(\cdot))^2$; this term gets larger and dominates the asymptotic variance as $\delta$ moves away from zero. 
We can estimate $\hat{f}$ and $\hat{m}$ with any flexible model (e.g, random forest, KNN, etc.) with the basis coefficients and covariates as model predictors. While functional data methods could also be used to estimate these nuisance functions, there are fewer flexible methods for functional data that can accurately model complex nuisance functions. 
However, note that since the conditional density $f$, is univariate, $f(B_1^{(2)} |A^{-\tilde{Q}}(\cdot))$ is well-defined and we can also consider the equivalent estimator to $\hat{\mu}^{\tilde{Q}(\delta)}_{\bm{J}(n)}$ which conditions on the true functional treatment, rather than an approximation, $\hat{\mu}^{\tilde{Q}(\delta)}$. For this estimator, it is necessarily to estimate $\hat{f}$ and $\hat{m}$ with functional data methods. However, there are no assumptions required on $\Delta_{\bm{J}(n)}$ for these estimators to achieve root-$n/\delta$ consistency as stated Corollary \ref{thm:coro_estimator_asymp}.
Thus, while we identify and define the estimands in terms of a basis approximation, we can propose estimators and derive asymptotic results that require no assumptions on the basis approximation used. 
Either of $\hat{\mu}^{\tilde{Q}(\delta)}_{\bm{J}(n)}$ or $\hat{\mu}^{\tilde{Q}(\delta)}$ could be used; $\hat{\mu}^{\tilde{Q}(\delta)}_{\bm{J}(n)}$ requires more assumptions for root-$n/\delta$ consistency but allows more options for flexible modeling of the nuisance functions than $\hat{\mu}^{\tilde{Q}(\delta)}$.

In addition to estimating $\mu^{\tilde{Q}(\delta)}$, we are interested in the treatment \textit{effect} of the stochastic policy, $\tau^{\tilde{Q}(\delta)} = \mu^{\tilde{Q}(\delta)} - \mu$. We propose the estimator $\hat{\tau}^{\tilde{Q}(\delta)}_{\bm{J}(n)} = \hat{\mu}^{\tilde{Q}(\delta)}_{\bm{J}(n)} - \frac1n\sum_{i=1}^n Y_i$. Since $\frac1n\sum_{i=1}^n Y_i$ is a root$-n$ consistent estimator of $\mu$, $\hat{\tau}^{\tilde{Q}(\delta)}_{\bm{J}(n)}$ is a root$-\min(n, n/\delta)$ consistent estimator as stated in Theorem \ref{thm:contrast_asymp}. 

\begin{theorem}
   Assume the conditions and assumptions of Theorem \ref{thm:estimator_asymp} hold. Then, $\frac{\sqrt{n}}{\sigma^{\tau^{\tilde{Q}(\delta)}}_{\bm{J}(n)}}( \hat{\tau}^{\tilde{Q}(\delta)}_{\bm{J}(n)} - \tau^{\tilde{Q}(\delta)}) \to N(0,1)$,
    where $\sigma^{\tau^{\tilde{Q}(\delta)^2}}_{\bm{J}(n)} = E\left[\{w^{\tilde{Q}(\delta)}(\bm{X}, A_{\bm{J}(n)}^{-\tilde{Q}(\delta)}(\cdot)) - 1\}^2v(\bm{X}, A_{\bm{J}(n)}(\cdot))\right.$ \\ $\left. +  w^{\tilde{Q}(\delta)}(\bm{X}, A_{\bm{J}(n)}^{-\tilde{Q}(\delta)}(\cdot))^2d^{\tilde{Q}(\delta)}(\bm{X}, A_{\bm{J}(n)}(\cdot))^2 \right.\left.-2\Var_{\tilde{Q}(\delta)}\{m(\bm{X}, A_{\bm{J}(n)}(\cdot))| \bm{X}, A_{\bm{J}(n)}^{-}(\cdot)\}\right] $ \\ $+ \Var\{d^{\tilde{Q}(\delta)}(\bm{X}, A_{\bm{J}(n)}(\cdot))\}$ for $\tilde{q} = q, q^*$.
    \label{thm:contrast_asymp}
\end{theorem}
Similar to the asymptotic variance of $\hat{\mu}^{\tilde{Q}(\delta)}_{\bm{J}(n)}$, the terms of $\sigma^{\tau^{\tilde{Q}(\delta)^2}}_{\bm{J}(n)}$ containing $d^{\tilde{Q}(\delta)}(\bm{X}, A_{\bm{J}(n)}(\cdot))$ vary the most with $\delta$ and dominate the asymptotic variance as $\delta$ moves farther away from zero. Furthermore, we have an equivalent result to Corollary \ref{thm:coro_estimator_asymp}, for the estimator $\hat{\tau}^{\tilde{Q}(\delta)}$ whose nuisance functions condition on the true functional treatment, as stated in Corollary \ref{thm:coro_contrast_asymp}. For all estimators, we obtain confidence intervals through plug-in estimators of the asymptotic variance. 

\section{Simulation studies}
\label{sec:simulation4}
We focus our simulations on evaluating the performance of $\hat{\mu}^{\tilde{Q}(\delta)}_{\bm{J}(n)}$ and $\hat{\tau}^{\tilde{Q}(\delta)}_{\bm{J}(n)}$ for $\tilde{q} = q, q^*$ because we can estimate their nuisance parameters with a wide variety of flexible models. We aim to validate the asymptotic results of Section \ref{sec:estimation}, evaluate estimator performance across various data generating scenarios, and explore which factors influence the bias of $\hat{\mu}^{Q^*(\delta)}_{\bm{J}(n)}$ when assumption \textbf{A$\bm{3}^*$} does not hold.

\subsection{Simulation methods}
We evaluate the rate of our estimators, consistency of the variance estimators, and coverage by computing these estimators for $n=1000,2000,4000,8000,1600,32000$.
We compare our asymptotic variance estimate to the empirical variance of the estimator since all terms in the asymptotic variance must be estimated to get a Monte Carlo approximation of the true asymptotic variance. For $n=1000$, we compare the absolute percent bias, average estimated variance, and coverage of our proposed estimators to evaluate how finite-sample estimator performance varies by data generating scenario. 
We generate scenarios that vary by the number of true basis coefficients over $[0, t_1)$ ($J_1=2,5,8$) and $(t_2, T]$ ($J_3=2,5,8$), the standard deviation of the observed conditional distribution of $B_1^{(2)}$ ($\sigma = 10, 12$), and the proportion of variance explained ($99\%$ and $99.9\%$) by the selected FPCA basis functions. The full data generation process is described in supplementary Section \ref{sec:supp_sim_methods}. We explore $\delta = 0.03, 0.045, 0.06, 0.09, 0.12, 0.16$ for $\sigma = 10$ and $\delta = 0.02, 0.03, 0.04, 0.06, 0.08, 0.1$ for $\sigma = 12$. These $\delta$ ranges where chosen such that this range encapsulates minimal to substantial practical positivity violations for a sample size of $n=1000$.  We use a cross-fitting with two folds to compute the proposed estimators. We estimate the conditional expectation with random forests and the conditional density with LinCDE
\citep{gao_lincde_2022}. We generate 1,000 independent datasets for all scenarios and compute the true estimand values using 1 million Monte Carlo samples. 

\subsection{Simulation results}
As sample size increases, the root mean squared error (RMSE) of $\hat{\mu}^{Q(\delta)}_{\bm{J}(n)}$ and $\hat{\tau}^{Q(\delta)}_{\bm{J}(n)}$ decreases (Figure \ref{fig:sim_rate} and supplementary Figure \ref{fig:sim_rate_tau}). Specifically,  $\log_2(\text{RMSE})$ decreases linearly with a slope $\frac12$; this confirms that estimator error decreases at a scaled root-$n$ rate as derived in Theorems \ref{thm:estimator_asymp} and \ref{thm:contrast_asymp}. Across all sample sizes, as $\delta$ increases so does estimator RMSE. Furthermore, as sample size increases the asymptotic variance estimate approaches the scaled empirical variance of the estimator, which is a close approximation of the true asymptotic variance. At lower sample sizes, we are substantially underestimating the asymptotic variance for $\delta= 0.09,0.12, 0.16$; this is likely due to practical, finite-sample positivity violations. Coverage remains consistent at the nominal level or increases towards the nominal level as sample size increases and variance estimation improves.

\begin{figure}[h]
    \centering
    \includegraphics[width=0.9\linewidth]{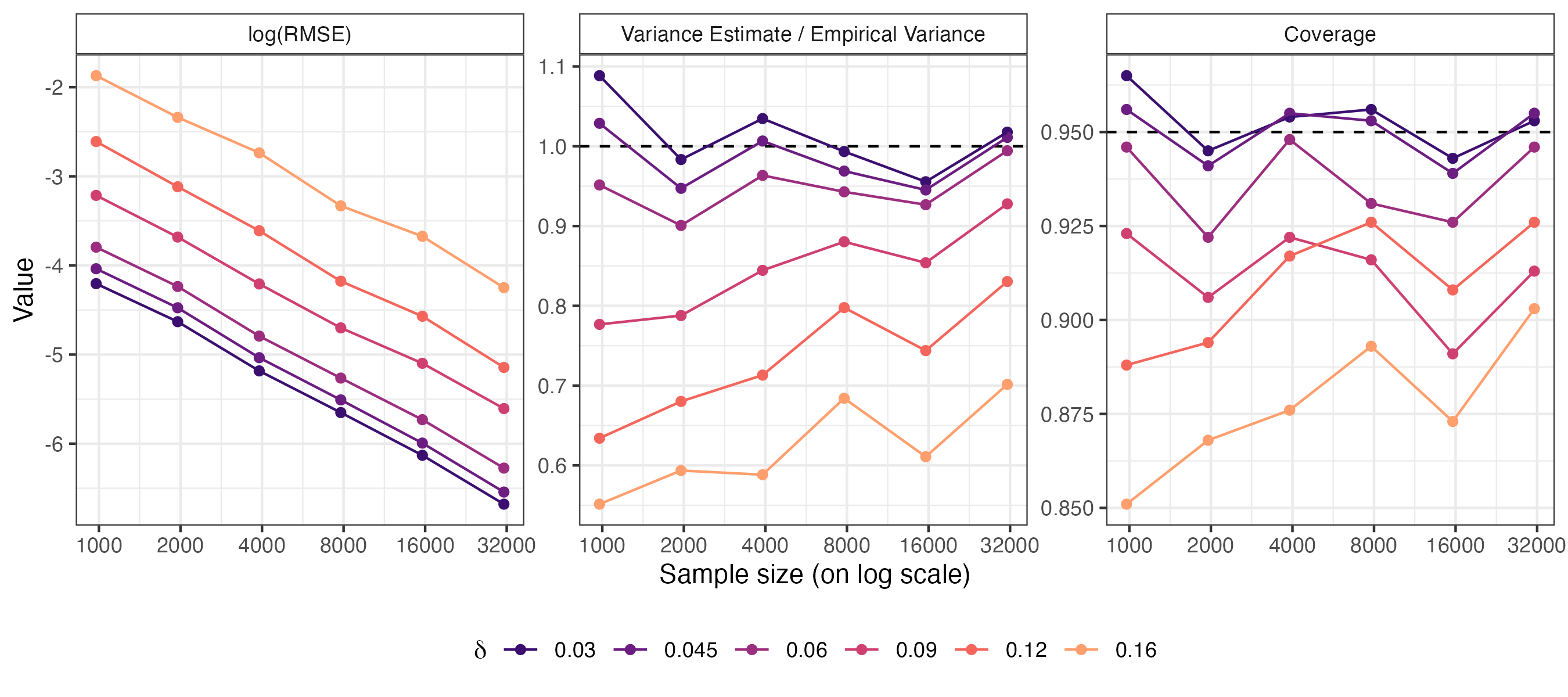}
    \caption{\footnotesize Estimator log(RMSE), the ratio of the average variance estimate and empirical estimator variance, and coverage for $\hat{\mu}^{Q(\delta)}_{\bm{J}(n)}$ by sample size and $\delta$. Dashed lines indicate the desired or nominal values for the right two panels.}
    \label{fig:sim_rate}
\end{figure}

For lower $\delta$ values and $J_3 = 2$, the absolute percent bias of $\hat{\mu}^{Q(\delta)}_{\bm{J}(n)}$ tends to be low ($<1.5\%$), the asymptotic variance tends to be estimated well, and coverage is at the nominal level (Figure \ref{fig:sim_q}). However, as $J_3$ increases the conditional density estimation problem becomes more challenging, resulting in increased bias and lower coverage. In addition, as $\delta$ increases, there is increased estimator bias and poor variance estimation likely due to finite-sample positivity violations. When the number of true of basis functions for $[0, t1)$ increases from two (Figure \ref{fig:sim_q}) to eight (supplementary Figure \ref{fig:sim_q10_b1_8}), estimator performance slightly poorer as estimating conditional expectation function is more challenging. These results emphasize the importance of the nuisance parameter estimation at smaller sample sizes. 
We see similar trends for $\hat{\tau}^{Q(\delta)}_{\bm{J}(n)}$, although absolute percent bias tends to be higher as the true treatment effects are small (supplementary Figure \ref{fig:sim_q_q10_tau}). 
When there is higher conditional density variance, $\sigma = 12$, estimator performance is slightly better for all $\delta$ values (supplementary Figure \ref{fig:sim_q_q12}).
Finally, selecting FPCA bases that explain $99.9\%$ rather than $99\%$ of the total variance does not substantially change estimation performance, though coverage is slightly improved (supplementary Figure \ref{fig:sim_q_q10_fpca999}).

\begin{figure}[h]
    \centering
    \includegraphics[width=0.9\linewidth]{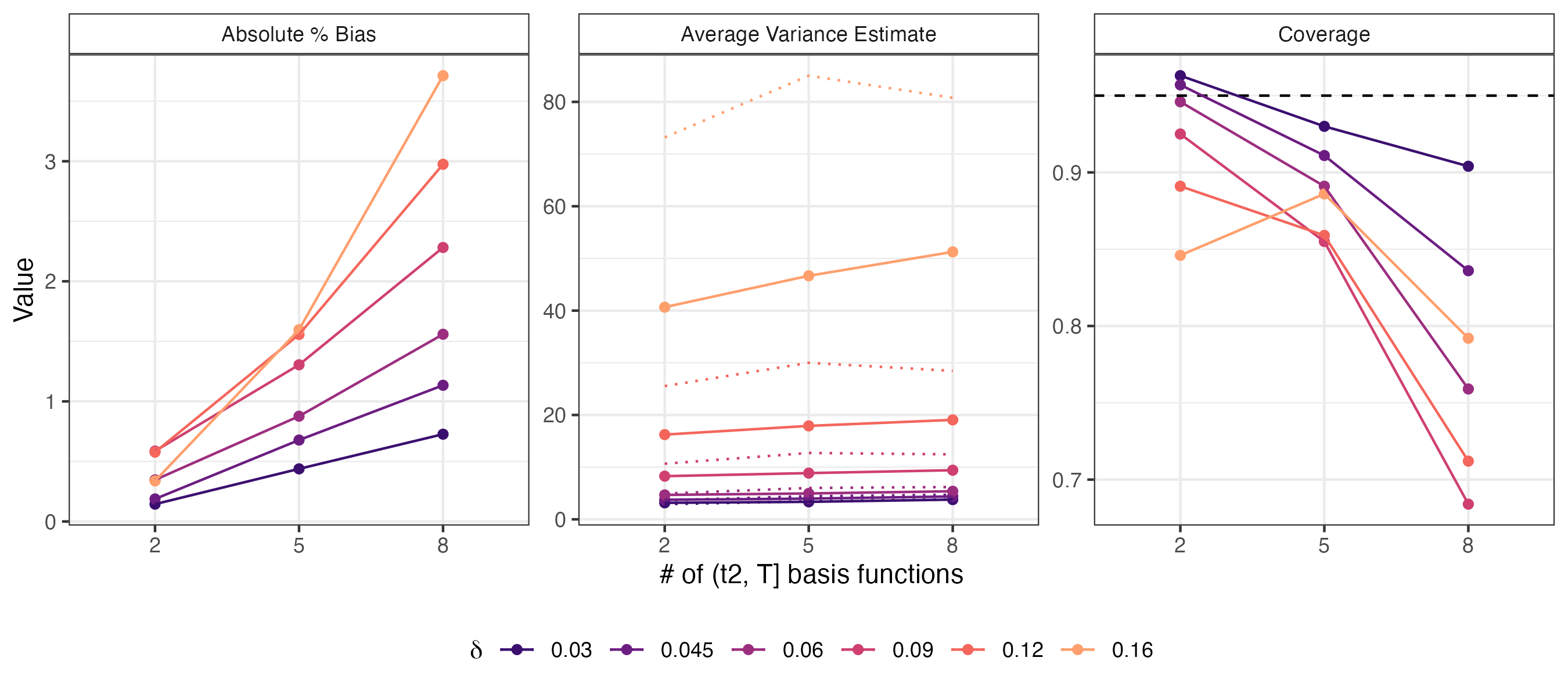}
    \caption{\footnotesize Estimator absolute percent bias, average variance estimate, and coverage by $J_3$ and $\delta$ for $\sigma = 10$, $J_1 = 2$, and 99\% of variance explained by the FPCA approximation for $\hat{\mu}^{Q(\delta)}_{\bm{J}(n)}$. The dotted line in the middle panel is the empirical estimator variance and the dashed line in the right panel is the nominal coverage level.}
    \label{fig:sim_q}
\end{figure}

When assumption \textbf{A}$\bm{3}^*$ holds, $\hat{\mu}^{Q^*(\delta)}_{\bm{J}(n)}$ has similar performance to $\hat{\mu}^{Q(\delta)}_{\bm{J}(n)}$ across simulation scenarios (supplementary Figure \ref{fig:sim_qstar_b3_0}). 
However, when assumption \textbf{A}$\bm{3}^*$ does not hold, there can be extreme bias and severe undercoverage (supplementary Figure \ref{fig:sim_qstar}). Specifically, when the function over $[0, t_1)$ is not highly predictive of the function over $(t_2, T]$, absolute percent bias can be as high as 25\%. Yet, when the function over $[1, t_1)$ is highly predictive of the function over $(t_2, T]$, bias is low and coverage remains high. Thus, when assumption \textbf{A}$\bm{3}^*$ is violated, the bias of $\hat{\mu}^{Q^*(\delta)}_{\bm{J}(n)}$ is a function of both 1) how related the function over $(t_2, T]$ is to the distribution of $B_1^{(2)}$; and 2) how predictive the function over $[0, t_1)$ is of the function over $(t_2, T]$.


\section{NHANES physical activity data application}
\label{sec:nhanes_application}
The National Health and Nutrition Examination Survey (NHANES), conducted by the U.S. Center for Disease Control, annually samples and collects demographic, health, and nutrition data on approximately 5,000 individuals in the United States through questionnaires, medical exams, and lab tests. We focus on data from the 2011-2012 and 2013-14 NHANES surveys \citep{chen_national_2018}, where participants were asked to wear a tri-axial wrist-worn accelerometer for 7 days to collect objective, 24-hour physical activity data. 
We are interested determining the causal relationship between realistic, population-level changes in physical activity and 5-year all-cause mortality \citep{smirnova_predictive_2020,leroux_quantifying_2021,leroux_nhanes_2024}.
NHANES transforms the raw accelerometer data to minute-level physical activity data in the Monitor Independent Movement Summary (MIMS) units \citep{john_open-source_2019}, which is an open-source, device-agnostic measure of physical activity. We aggregate the physical activity across the 7 collection days for these individuals to obtain an average summary of each individual's weekly physical activity at the minute-level.
We use 14 traditional predictors of mortality as covariates \citep{smirnova_predictive_2020,leroux_quantifying_2021,leroux_nhanes_2024}: age, gender, race, educational level, BMI, diabetes, coronary heart
disease, congestive heart failure, heart attack, stroke,
cancer, alcohol use, cigarette smoking, and mobility problems. For the resulting sample of 7,504 observations, the estimated 5-year all-cause mortality rate is 6.10\%

Previous studies primarily relate aggregates of the minute-level physical activity (e.g., total daily MIMS) to all-cause mortality with standard regression methods such as logistic and Cox regression. However, we can also use scalar-on-function logistic regression (SoFR) \citep{goldsmith_penalized_2011}, from which we obtain a functional coefficient estimate for the relationship between physical activity and the odds of mortality (supplementary Figure \ref{fig:nhanes_sofr}).
This coefficient and its corresponding 95\% confidence interval vary over the 24 hours; a one-unit increase physical activity is significantly related a reduction in mortality from around 9AM-9PM and the coefficient is the most negative around 1PM-4PM. A natural question that this analysis prompts is: does the smooth coefficient illustrate a causal relationship between physical activity and mortality at different times of the day?
To address this question, we compute our proposed treatment effect estimator ($\hat{\tau}_{\bm{J}(n)}^{Q(\delta)}$) for stochastic policies that reflect increases in physical activity over three time periods: 7AM-10AM, 1PM-4PM, and 5PM-8PM. Further details on data cleaning, method implementation, connecting the MIMS units to more interpretable step counts \citep{PhysioNet-minute-level-step-count-nhanes-1.0.0} are described in supplementary Section \ref{sec:supp_real_data}.


We first explore how the expected step count increase under $q$ for a single $\delta$ value differs by key demographic factors (supplementary Figure \ref{fig:nhanes_cov}). 
In general, the median step count increase under the policy decreases and the step count distribution becomes more narrow as age increases. Furthermore, for almost every age category and time period, the median step count increase under the policy is lower and the distribution is narrower for those with mobility problems. Thus, the policy implicitly varies as expected along these two demographic factors; for a deterministic MFTP this variation would have to be deliberately specified. 
Treatment effect point estimates for all time periods indicate that an increase in physical activity is related to a decrease in the mortality rate (Figure \ref{fig:nhanes_te} and  supplementary Figure \ref{fig:nhanes_te_full_delta}). While the treatment effects become more negative as $\delta$ increases across all time periods, the confidence intervals are large for $\delta > 0.0225$ such that these treatment effect estimates are not significant at the $\alpha = 0.05$ level for any time period (supplementary Figure \ref{fig:nhanes_te_full_delta}). 
For 1PM-4PM, the treatment effect estimate is only significant at the smallest $\delta$ value which corresponds $-0.08\% \:\: (-0.16, 0.002)$ additive decrease in mortality. 
For 7AM-10AM, the treatment effect estimates are significant for smaller $\delta$ values; the largest significant estimated effect is $-0.62\% \:\: (-1.22, -0.03)$ which corresponds to a median expected step count increase of 123 steps under the policy. The treatment effect estimates for 5PM-8PM are significant for almost all $\delta < 0.0225$; the largest significant estimated effect is $	
-0.90 \:\: (-1.66, -0.14)$ which corresponds to a median expected step count increase of 234 steps under the policy. 

\begin{figure}[h!]
    \centering
    \includegraphics[width=0.8\linewidth]{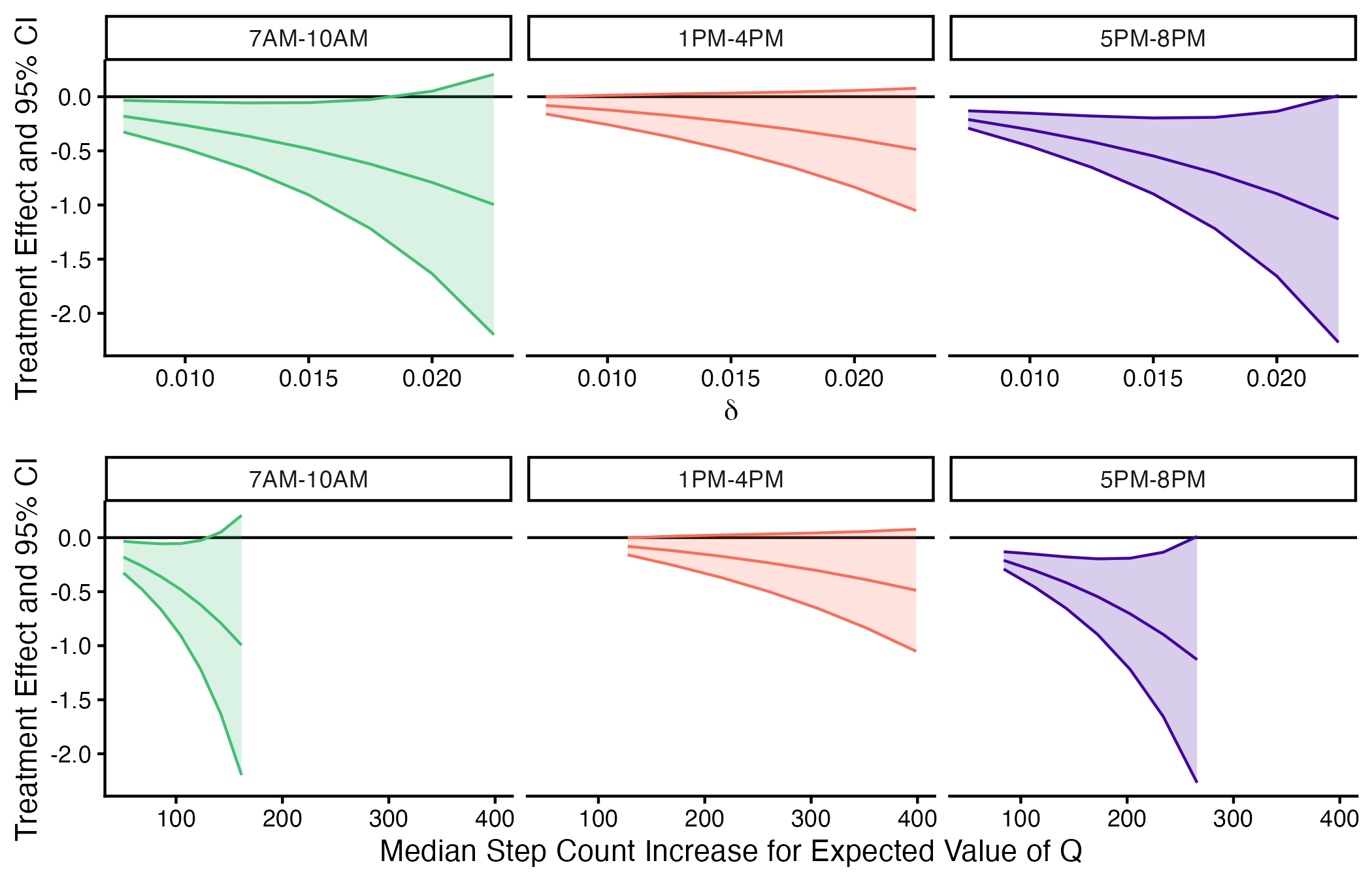}
    \caption{\footnotesize Treatment effect estimates and 95\% confidence bounds for the effect of stochastic policies that define increases in physical activity on 5-year all-cause mortality. The columns indicate the time period the policy is implemented over and the horizontal axis is $\delta$ for the top panel and the median expected step count increase under the policy for the bottom panel.}
    \label{fig:nhanes_te}
\end{figure}

These results do not align with the SoFR basis coefficient, although this is unsurprising as SoFR is an associative method. However, there are a variety of reasons we may see that increasing physical activity over 7AM-10AM or 5PM-8PM corresponds to significant and relatively large decreases in the 5-year all cause mortality rate. One potential reason is that physical activity over 7AM-10AM or 5PM-8PM may tend to be more intentional or vigorous than over 1PM-4PM such and this type of activity may be more effective at decreasing mortality. Another factor may be that meals are often eaten over 7AM-10AM and 5PM-8PM; physical activity shortly after eating is related health benefits such as increased weight loss and better hyperglycemia management \citep{hijikata_walking_2011,hashimoto_positive_2025}. While additional studies are needed to explore the mechanisms through which the effect of physical activity on mortality differs by time period, these results suggest mortality rates can be reduced at a population level through realistic increases in physical activity over 7AM-10AM and 5PM-8PM.

\section{Discussion}
\label{sec:discussion4}

There are substantial challenges in defining and estimating causal effects for functional treatments while accounting for time varying confounding. To define meaningful functional treatment estimands that satisfy positivity, we adapt stochastic policies for continuous treatments to functional treatments. As part of this adaptation, we develop a novel functional basis construction method that allows for one basis function to be explicitly chosen without sacrificing the proportion of total variance explained by the corresponding basis approximation. We leverage discrete-time longitudinal causal inference techniques to adjust for time-varying confounding and define two estimands of interest that reflect different manners of accounting for temporal relationships in the functional data. We then propose double machine learning-style estimators and show that they achieve rate double robustness and perform well in finite-samples across a wide variety of data generating scenarios. By applying our method to NHANES physical activity data, we find that increasing physical activity over 7AM-10AM or 5PM-8PM at the population-level significantly decreases 5 year all-cause mortality rates.

One characteristic of our proposed methodology, and stochastic policies in general, is that the counterfactual of the future treatment (e.g, function over $(t_2, T]$) given the stochastic policy is always equal to the observed future treatment. Thus, to explore the causal effect of both the stochastic policy of interest and potential changes to future treatment, it is necessary to hypothesize additional stochastic policies for that future treatment, i.e., construct sequential stochastic policies. While sequential stochastic policy estimands and estimators have been proposed for binary treatments \citep{kennedy_nonparametric_2019}, additional work is needed to develop sequential stochastic policy estimands and estimators for both continuous and functional treatments. In addition, our proposed estimators require estimating conditional density functions, which tends to be challenging especially at smaller sample sizes. Future work may focus on developing estimators for stochastic policies that do not require conditional density estimation to improve finite-sample estimator performance. 
An additional limitation of our method is that the time period of interest for the stochastic policy must be specified \textit{a priori} and be the same for all individuals in the data. Future work includes developing causal effect estimation methods for stochastic policies applied to different time periods across individuals both when the time periods are selected \textit{a priori} and dynamically.

\section{Competing interests}
No competing interest is declared.


\section{Funding}
The first author was supported by the National Science Foundation Graduate Research Fellowship Program under Grant No. 2237827.

\section{Data availability}
The data supporting this manuscript is publicly available at \url{https://wwwn.cdc.gov/nchs/nhanes/}. The code for replicating the analysis and results is openly accessible at \url{https://github.com/marthabarnard/causal_functional_stochastic}.

\putbib
\end{bibunit}

\pagebreak

\begin{bibunit}

\setcounter{theorem}{0}
\setcounter{figure}{0}
\setcounter{section}{0}

\center{
\textbf{\Large Supplementary Materials for Causal Inference for Functional Treatments with Stochastic Policies} \\}
\justifying

\section{Theoretical Results and Proofs}
\label{sec:supp_theory}
Let $\Vert \cdot \Vert_2$ be the Euclidean norm.
\subsection{Theorem \ref{thm:novel_basis}}
\begin{theorem}
    Let $\bar{A}_{J-1}(\cdot)$ and $\tilde{A}_J(\cdot)$ be, respectively, the $\{\psi_j(\cdot)\}_{j=1}^{J - 1}$ basis (FPCA) and $\{\gamma_j(\cdot)\}_{j=1}^J$ basis approximations of $A(\cdot) - a_0(t)$. For any choice of $\gamma_1(\cdot)$, $\int_0^T \Var\{\tilde{A}_J(t)\}dt \geq \int_0^T \Var\{\bar{A}_{J-1}(t)\}dt$. 
\end{theorem}
\begin{proof}
We first introduce and prove the following lemma:
\begin{supplemma}
    For any choice of $\gamma_1(\cdot)$,  $||A(\cdot)-\tilde{A}_J(\cdot)||_2^2 \leq ||A(\cdot)-\bar{A}_{J-1}(\cdot)||_2^2$.  
    \label{lemma:novel_basis}
\end{supplemma}
\begin{proof}
Since the basis $\{\gamma\}_{j=1}^\infty$ is constructed through Gram-Schmidt orthonormalization, we can represent the original FPCA basis as $\psi_j(\cdot) = \sum_{h=1}^{j+1} \langle \gamma_h (\cdot), \psi_j (\cdot) \rangle \gamma_h(\cdot)$. Then,
 \begin{align*}
        A(\cdot) &= a_0(\cdot) +\sum_{j=1}^{\infty} \theta_j^{1/2}A_j\psi_j(\cdot) \\
        &= a_0(\cdot) +\sum_{j=1}^{\infty} \theta_j^{1/2}A_j\sum_{h=1}^{j+1} \langle \gamma_h (\cdot), \psi_j (\cdot) \rangle \gamma_h(\cdot) \\
        &= a_0(\cdot) +\sum_{j=1}^{\infty} \left\{\sum_{h=\max(1, j-1)}^\infty \theta_h^{1/2}A_h\langle \gamma_j (\cdot), \psi_h (\cdot) \rangle\right\} \gamma_j(\cdot)
    \end{align*}
where third equality follows from swapping the $h$ and $j$ indices. Then, $\tilde{A}_J(\cdot)$ can be simplified to the following,
\begin{align*}
    \tilde{A}_J(\cdot) &= a_0(\cdot) +\sum_{j=1}^{J} \left\{\sum_{h=\max(1, j-1)}^\infty \theta_h^{1/2}A_h\langle \gamma_j (\cdot), \psi_h (\cdot) \rangle\right\} \gamma_j(\cdot) \\
    &= a_0(\cdot) + \sum_{h=1}^{J-1}\theta_h^{1/2}A_h \sum_{j=1}^{h+1}\langle \gamma_j (\cdot), \psi_h (\cdot) \rangle \gamma_j (\cdot) + \sum_{j=1}^J\sum_{h=J}^\infty \theta_h^{1/2}A_h\langle \gamma_j (\cdot), \psi_h (\cdot) \rangle \gamma_j(\cdot) \\
    &= a_0(\cdot) + \underbrace{\sum_{h=1}^{J-1}\theta_h^{1/2}A_h \psi_h(\cdot)}_{J-1\text{ FPCA truncation}} + \sum_{j=1}^J\sum_{h=J}^\infty \theta_h^{1/2}A_h\langle \gamma_j (\cdot), \psi_h (\cdot) \rangle \gamma_j(\cdot) 
\end{align*}
and we see $ \tilde{A}_J(\cdot)$ is equal to $\bar{A}_{J-1}(\cdot)$ plus an additional term that has expectation $0$. Then,
\begin{align*}
    A(\cdot) - \tilde{A}_J(\cdot) &= \sum_{h=J}^\infty \theta_h^{1/2}A_h \psi_h(\cdot) - \sum_{j=1}^J\sum_{h=J}^\infty \theta_h^{1/2}A_h\langle \gamma_j (\cdot), \psi_h (\cdot) \rangle \gamma_j(\cdot) \\
    &= \sum_{h=J}^\infty \theta_h^{1/2}A_h \left\{\psi_h(\cdot) - \sum_{j=1}^J\langle \gamma_j (\cdot), \psi_h (\cdot) \rangle \gamma_j(\cdot)\right\} \\
    &= \sum_{h=J}^\infty \theta_h^{1/2}A_h \left\{\sum_{j=1}^{h+1} \langle \gamma_j (\cdot), \psi_h (\cdot) \rangle \gamma_j(\cdot) - \sum_{j=1}^J\langle \gamma_j (\cdot), \psi_h (\cdot) \rangle \gamma_j(\cdot)\right\} \\
    &= \sum_{h=J}^\infty \theta_h^{1/2}A_h \left\{\sum_{j=J+1}^{h+1} \langle \gamma_j (\cdot), \psi_h (\cdot) \rangle \gamma_j(\cdot) \right\} \\
    &= \sum_{h=J}^\infty \theta_h^{1/2}A_h \left\{\sum_{j=J+1}^{\infty} \langle \gamma_j (\cdot), \psi_h (\cdot) \rangle \gamma_j(\cdot) \right\} \\
    &= \sum_{j=J+1}^{\infty} \left\langle  \sum_{h=J}^\infty \theta_h^{1/2}A_h \psi_h (\cdot) , \gamma_j(\cdot)\right\rangle \gamma_j(\cdot)
\end{align*}
where the first four equalities follow from simple computation and the fact that $\psi_j(\cdot) = \sum_{h=1}^{j+1} \langle \gamma_h (\cdot), \psi_j (\cdot) \rangle \gamma_h(\cdot)$, the fifth equality follows from the fact that for $\{\gamma_j(\cdot)\}_{j=1}^\infty$ constructed through the Gram-Schmidt process, $\sum_{j=J+1}^{h+1} \langle \gamma_j (\cdot), \psi_h (\cdot) \rangle \gamma_j(\cdot)  = \sum_{j=J+1}^{\infty},\langle \gamma_j (\cdot), \psi_h (\cdot) \rangle \gamma_j(\cdot) $, and the sixth equality follows from the linearity of the inner product and rearranging the two sums. Then,
\begin{align*}
    \Vert A(\cdot) - \tilde{A}_J(\cdot)\Vert_2^2 &= \left\Vert\sum_{j=J+1}^{\infty} \left\langle  \sum_{h=J}^\infty \theta_h^{1/2}A_h \psi_h (\cdot) , \gamma_j(\cdot)\right\rangle \gamma_j(\cdot)\right\Vert_2^2 \\
    &= \sum_{j=J+1}^{\infty} \left\langle  \sum_{h=J}^\infty \theta_h^{1/2}A_h \psi_h (\cdot) , \gamma_j(\cdot)\right\rangle^2 \\
    &\leq \left\Vert\sum_{h=J}^\infty\theta_h^{1/2}A_h \psi_h(\cdot)\right\Vert_2^2 \\
    &= \Vert A(\cdot)-\bar{A}_{J-1}(\cdot)\Vert_2^2
\end{align*}
where the second equality follows from the fact that $\{\gamma_j(\cdot)\}_{j=J+1}^\infty$ is an orthonormal set and the third inequality follows from Bessel's inequality. Thus, $\Vert A(\cdot)-\tilde{A}_J(\cdot)\Vert_2^2 \leq \Vert A(\cdot)-\bar{A}_{J-1}(\cdot)\Vert_2^2$ as desired. 
\end{proof}

We now return to a proof of Theorem \ref{thm:novel_basis}
    First, we have that 
    \begin{align*}
         \int_0^T  \Var\{\tilde{A}_J(t)\}dt &= \int_0^T E[\tilde{A}_J(t)^2]dt \\
        &= E\left[\int_0^T \tilde{A}_J(t)^2 dt\right] \\
        &= E\left[\Vert \tilde{A}_J(\cdot) \Vert_2^2\right]
    \end{align*}
    such that it suffices to show that $E[\Vert \tilde{A}_J(t) \Vert_2^2] \geq \int_0^T \Var\{\bar{A}_{J-1}(t)\}dt$.
    Then, note that since $\{\gamma_j\}_{j=1}^\infty$ is orthonormal,
    \begin{align*}
         \Vert A(\cdot) - a_0(\cdot) \Vert_2^2 &= \Vert A(\cdot) - a_0(\cdot) - \tilde{A}_J(\cdot) + \tilde{A}_J(\cdot) \Vert_2^2  \\
         &= \Vert \sum_{j=1}^{J} B_j\gamma_j(t)  + \sum_{j=J+1}^{\infty} B_j\gamma_j(t) \Vert_2^2 \\
         &= \Vert \sum_{j=1}^{J} B_j\gamma_j(t) \Vert_2^2 + \Vert \sum_{j=J+1}^{\infty} B_j\gamma_j(t) \Vert_2^2\\
         &=  \Vert \tilde{A}_J(\cdot) \Vert_2^2  + \Vert A(\cdot) - a_0(\cdot) - \tilde{A}_J(\cdot)\Vert_2^2.  \\
    \end{align*}
    Next, building upon the proof of Lemma 1 we also have that, 
    \begin{align*}
        &\Vert A(\cdot) - a_0(\cdot) \Vert_2^2 = \Vert A(\cdot) - a_0(\cdot) - \tilde{A}_J(\cdot) + \tilde{A}_J(\cdot) \Vert_2^2 \\
        &= \left\Vert\left\{\sum_{j=J+1}^{\infty} \left\langle  \sum_{h=J}^\infty \theta_h^{1/2}A_h \psi_h (\cdot) , \gamma_j(\cdot)\right\rangle \gamma_j(\cdot)\right\}  \right.\\
        &\quad+ \left. \left \{\sum_{h=1}^{J-1}\theta_h^{1/2}A_h \psi_h(\cdot) + \sum_{j=1}^{J} \left\langle  \sum_{h=J}^\infty \theta_h^{1/2}A_h \psi_h (\cdot) , \gamma_j(\cdot)\right\rangle \gamma_j(\cdot) \right\}\right\Vert_2^2  \\
        &= \sum_{j=J+1}^{\infty} \left\langle  \sum_{h=J}^\infty \theta_h^{1/2}A_h \psi_h (\cdot) , \gamma_j(\cdot)\right\rangle^2 + \sum_{j=1}^{J} \left\langle  \sum_{h=J}^\infty \theta_h^{1/2}A_h \psi_h (\cdot) , \gamma_j(\cdot)\right\rangle^2 + \Vert \sum_{h=1}^{J-1}\theta_h^{1/2}A_h \psi_h(\cdot)\Vert_2^2 \\
        &\quad + 2\left\langle \sum_{h=1}^{J-1}\theta_h^{1/2}A_h \psi_h(\cdot), \sum_{j=1}^{J} \left\langle  \sum_{h=J}^\infty \theta_h^{1/2}A_h \psi_h (\cdot) , \gamma_j(\cdot)\right\rangle \gamma_j(\cdot) + \sum_{j=J+1}^{\infty} \left\langle  \sum_{h=J}^\infty \theta_h^{1/2}A_h \psi_h (\cdot) , \gamma_j(\cdot)\right\rangle 
        \right\rangle \\
        &= \sum_{j=J+1}^{\infty} \left\langle  \sum_{h=J}^\infty \theta_h^{1/2}A_h \psi_h (\cdot) , \gamma_j(\cdot)\right\rangle^2 + \sum_{j=1}^{J} \left\langle  \sum_{h=J}^\infty \theta_h^{1/2}A_h \psi_h (\cdot) , \gamma_j(\cdot)\right\rangle^2 \\
        &\quad+ \Vert \sum_{h=1}^{J-1}\theta_h^{1/2}A_h \psi_h(\cdot)\Vert_2^2  2\left\langle \sum_{h=1}^{J-1}\theta_h^{1/2}A_h \psi_h(\cdot), \sum_{h=J}^{\infty}\theta_h^{1/2}A_h \psi_h(\cdot) \right\rangle \\
        &= \sum_{j=J+1}^{\infty} \left\langle  \sum_{h=J}^\infty \theta_h^{1/2}A_h \psi_h (\cdot) , \gamma_j(\cdot)\right\rangle^2 + \sum_{j=1}^{J} \left\langle  \sum_{h=J}^\infty \theta_h^{1/2}A_h \psi_h (\cdot) , \gamma_j(\cdot)\right\rangle^2 
        + \Vert \sum_{h=1}^{J-1}\theta_h^{1/2}A_h \psi_h(\cdot)\Vert_2^2 
    \end{align*}
    where the second equality follows from the proof of Lemma 1, the third equality follows through expansions of the squared norm and the fact that $\{\gamma_j(\cdot)\}_{j=1}^\infty$ are orthonormal, the fourth equality follows from the definition of  $\sum_{h=J}^{\infty}\theta_h^{1/2}A_h \psi_h(\cdot)$ in terms of the basis $\{\gamma_j(\cdot)\}_{j=1}^\infty$, and the fifth from the fact that $\{\psi_h(\cdot)\}_{h=1}^\infty$ is orthonormal. Then by the proof of Lemma 1, since $\Vert A(\cdot) - a_0(\cdot) - \tilde{A}_J(\cdot)\Vert_2^2 = \sum_{j=J+1}^{\infty} \left\langle  \sum_{h=J}^\infty \theta_h^{1/2}A_h \psi_h (\cdot) , \gamma_j(\cdot)\right\rangle^2$, $\Vert \tilde{A}_J(\cdot) \Vert_2^2 = \sum_{j=1}^{J} \left\langle  \sum_{h=J}^\infty \theta_h^{1/2}A_h \psi_h (\cdot) , \gamma_j(\cdot)\right\rangle^2 
        + \left\Vert \sum_{h=1}^{J-1}\theta_h^{1/2}A_h \psi_h(\cdot)\right\Vert_2^2$.
By the Karhunen–Loève Theorem, $\int_0^T \Var\{\bar{A}_{J-1}(t)\}dt = \sum_{h=1}^{J-1}\theta_h$. Then,
\begin{align*}
    \int_0^T  \Var\{\tilde{A}_J(t)\}dt &= E\left[\Vert \tilde{A}_J(\cdot) \Vert_2^2 \right] \\
    &= E\left[\left\Vert \sum_{h=1}^{J-1}\theta_h^{1/2}A_h \psi_h(\cdot)\right\Vert_2^2 + \sum_{j=1}^{J} \left\langle  \sum_{h=J}^\infty \theta_h^{1/2}A_h \psi_h (\cdot) , \gamma_j(\cdot)\right\rangle^2  \right] \\
    &=  E[\sum_{h=1}^{J-1} \theta_hA_h^2] + E\left[ \sum_{j=1}^{J} \left\langle  \sum_{h=J}^\infty \theta_h^{1/2}A_h \psi_h (\cdot) , \gamma_j(\cdot)\right\rangle^2\right] \\
    &= \sum_{h=1}^{J-1} \theta_h + E\left[ \sum_{j=1}^{J} \left\langle  \sum_{h=J}^\infty \theta_h^{1/2}A_h \psi_h (\cdot) , \gamma_j(\cdot)\right\rangle^2\right] \\
    &\geq \int_0^T \Var\{\bar{A}_{J-1}(t)\}dt 
\end{align*}
as desired. 

\end{proof}

\subsection{Lemma \ref{thm:pop_avg_limit}}

\begin{lemma}
    Given conditions \textbf{C1} and \textbf{C2},  $\lim_{\bm{J} \to \infty} \mu_{\bm{J}}$ exists and $E[Y] = \lim_{\bm{J} \to \infty} \mu_{\bm{J}}$.
    \label{thm:pop_avg_limit}
\end{lemma}
\begin{proof}
    We will prove that $\{\mu_{\bm{J}}\}_{\bm{J} = 1}^\infty$ is a Cauchy sequence, which implies that $\lim_{\bm{J} \to \infty} \mu_{\bm{J}}$ exists. Let $\bm{N} = (N, N, N)$ We must show that for any $\delta > 0$ there exists an integer $N$ such that for any $\bm{J}' > \bm{J} > \bm{N}$, $|\mu_{\bm{J}'} - \mu_{\bm{J}}| \leq \delta$. Let $A_{\bm{J}}(\cdot) = a_0(\cdot) +  \sum_{k=1}^3\sum_{j=1}^{J_k} \theta^{(k)^{1/2}}_jA_j^{(k)}\psi_j^{(k)}(\cdot)$, $\bm{A}_{d}^{(k)} = \{A_j^{(k)}\}_{j=1}^{d}$ and $\bm{A}_{c:d}^{(k)} = \{A_j^{(k)}\}_{j=c}^{d}$. By the definition of $\mu_{\bm{J}}$ we have that,
    \begin{align*}
        \mu_{\bm{J}} &= \int_{\mathcal{\bm{X}} \times \mathbb{R}^{\bm{J}}} m\left\{\bm{x}, a_{\bm{J}}(\cdot)\right\}f_{\bm{X}}(\bm{x})f_{\bm{A}^{(1)}_{J_1}, \bm{A}^{(2)}_{J_2}, \bm{A}^{(3)}_{J_3}|\bm{X}}(\bm{a}^{(1)}_{J_1}, \bm{a}^{(2)}_{J_2}, \bm{a}^{(3)}_{J_3}|\bm{x}) d\bm{x}d\bm{a}^{(1)}_{J_1}d \bm{a}^{(2)}_{J_2}d\bm{a}^{(3)}_{J_3} \\
        &=  \int_{\mathcal{\bm{X}} \times \mathbb{R}^{\bm{J}}} m\left\{\bm{x}, a_{\bm{J}}(\cdot)\right\}f_{\bm{X}}(\bm{x})f_{\bm{A}^{(1)}_{J_1}, \bm{A}^{(2)}_{J_2}, \bm{A}^{(3)}_{J_3}|\bm{X}}(\bm{a}^{(1)}_{J_1}, \bm{a}^{(2)}_{J_2}, \bm{a}^{(3)}_{J_3}|\bm{x}) \\
        \quad &\times \left[\int_{\mathbb{R}^{\bm{J}' - \bm{J}}} f_{\bm{A}^{(1)}_{J_1+1:J_1'}, \bm{A}^{(2)}_{J_2+1:J_1'}, \bm{A}^{(3)}_{J_3+1:J_3'}| \bm{A}^{(1)}_{J_1}, \bm{A}^{(2)}_{J_2}, \bm{A}^{(3)}_{J_3}, \bm{X}}(\bm{a}^{(1)}_{J_1+1:J_1'}, \bm{a}^{(2)}_{J_2+1:J_2'}, \bm{a}^{(3)}_{J_3+1:J_3'}| \bm{a}^{(1)}_{J_1}, \bm{a}^{(2)}_{J_2}, \bm{a}^{(3)}_{J_3}, \bm{x})d\bm{a}^{(1)}_{J_1+1:J_1'} \right. \\
        &\quad \times \left. d\bm{a}^{(2)}_{J_2+1:J_2'}d\bm{a}^{(3)}_{J_3+1:J_3'})\right]d\bm{x}d\bm{a}^{(1)}_{J_1} d\bm{a}^{(2)}_{J_2}d\bm{a}^{(3)}_{J_3}
    \end{align*}
    where the area in brackets integrates to one. We also have that
    \begin{align*}
        \mu_{\bm{J}'} &= \int_{\mathcal{\bm{X}} \times \mathbb{R}^{\bm{J}'}} m\left\{\bm{x}, a_{\bm{J}'}(\cdot)\right\}f_{\bm{X}}(\bm{x})f_{\bm{A}^{(1)}_{J_1'}, \bm{A}^{(2)}_{J_2'}, \bm{A}^{(3)}_{J_3'}|\bm{X}}(\bm{a}^{(1)}_{J_1'}, \bm{a}^{(2)}_{J_2'}, \bm{a}^{(3)}_{J_3'}|\bm{x}) \\
        \quad &\times d\bm{x}d\bm{a}^{(1)}_{J_1'}d \bm{a}^{(2)}_{J_2'}d\bm{a}^{(3)}_{J_3'}. \\
    \end{align*}

Without loss of generality, we assume that $a_0(t)$ is a known function. From condition \textbf{C1}, we have that $\sum_{j=1}^\infty \theta^{(k)}_j < \infty$ for all $k=1,2,3$. Thus, for any $\delta > 0$ we can find $N$ such that for any $\bm{J}' > \bm{J} > \bm{N}$,
\begin{equation*}
    \sum_{j=J_k + 1}^{J_k'} \theta^{(k)}_j \leq \frac{\delta^2}{3C^2}, \text{ for all } k=1,2,3
\end{equation*}
where $C$ is the constant defined in condition \textbf{C2}. 
Without loss of generality, consider $\psi^{(k)}(\cdot) = \psi^{(k)}(\cdot)$ for interval $k$ and $\psi^{(k)}(\cdot) = 0$ at all other times in $[0, T]$ such that $\psi^{(k)}$ has domain $[0, T]$ for all $k = 1,2,3$. Then since $\{\psi^{(k)}_j\}_{j=1}^\infty$ are orthonormal for $k=1,2,3$, we have that,
\begin{align*}
    \left \Vert \sum_{k=1}^3 \sum_{j = J_k+1}^{J_k'} \theta_j^{(k)^{1/2}}A_j^{(k)}\psi_j^{(k)}(\cdot) \right\Vert_2^2&= \left \Vert \sum_{j = J_1+1}^{J_1'} \theta_j^{(1)^{1/2}}A_j^{(1)}\psi_j^{(1)}(\cdot) \right\Vert_2^2 + \left \Vert \sum_{j = J_2+1}^{J_2'} \theta_j^{(2)^{1/2}}A_j^{(2)}\psi_j^{(2)}(\cdot) \right\Vert_2^2 \\
    &\quad+ \left \Vert \sum_{j = J_3+1}^{J_3'} \theta_j^{(3)^{1/2}}A_j^{(3)}\psi_j^{(3)}(\cdot) \right\Vert_2^2 \\
    &\quad+ 2\left\langle \sum_{j = J_1+1}^{J_1'} \theta_j^{(1)^{1/2}}A_j^{(1)}\psi_j^{(1)}(\cdot), \sum_{j = J_2+1}^{J_2'} \theta_j^{(2)^{1/2}}A_j^{(2)}\psi_j^{(2)}(\cdot)\right\rangle \\
    &\quad+ 2\left\langle \sum_{j = J_1+1}^{J_1'} \theta_j^{(1)^{1/2}}A_j^{(1)}\psi_j^{(1)}(\cdot), \sum_{j = J_3+1}^{J_3'} \theta_j^{(3)^{1/2}}A_j^{(3)}\psi_j^{(3)}(\cdot)\right\rangle  \\
    &\quad+ 2\left\langle \sum_{j = J_3+1}^{J_3'} \theta_j^{(3)^{1/2}}A_j^{(3)}\psi_j^{(3)}(\cdot), \sum_{j = J_2+1}^{J_2'} \theta_j^{(2)^{1/2}}A_j^{(2)}\psi_j^{(2)}(\cdot)\right\rangle  \\
    &= \left \Vert \sum_{j = J_1+1}^{J_1'} \theta_j^{(1)^{1/2}}A_j^{(1)}\psi_j^{(1)}(\cdot) \right\Vert_2^2 + \left \Vert \sum_{j = J_2+1}^{J_2'} \theta_j^{(2)^{1/2}}A_j^{(2)}\psi_j^{(2)}(\cdot) \right\Vert_2^2 \\
    &\quad+ \left \Vert \sum_{j = J_3+1}^{J_3'} \theta_j^{(3)^{1/2}}A_j^{(3)}\psi_j^{(3)}(\cdot) \right\Vert_2^2 
\end{align*}
where the second equality follows from the fact that $\langle \psi_{j_1}^{(k_1)}, \psi_{j_2}^{(k_2)}\rangle = 0$ for all $j_1, j_2$ and $k_1 \neq k_2$. Then,
\begin{align*}
    E\left[\left\Vert \sum_{k=1}^3 \sum_{j = J_k+1}^{J_k'} \theta_j^{(k)^{1/2}}A_j^{(k)}\psi_j^{(k)}(\cdot) \right\Vert_2^2 \right] &= E\left[\left \Vert \sum_{j = J_1+1}^{J_1'} \theta_j^{(1)^{1/2}}A_j^{(1)}\psi_j^{(1)}(\cdot) \right\Vert_2^2\right] + E\left[\left \Vert \sum_{j = J_2+1}^{J_2'} \theta_j^{(2)^{1/2}}A_j^{(2)}\psi_j^{(2)}(\cdot) \right\Vert_2^2\right] \\
    &\quad + E\left[\left \Vert \sum_{j = J_3+1}^{J_3'} \theta_j^{(3)^{1/2}}A_j^{(3)}\psi_j^{(3)}(\cdot) \right\Vert_2^2\right] \\
    &= E\left[\sum_{j = J_1+1}^{J_1'} \theta_j^{(1)}A_j^{(1)^2}\right] + E\left[\sum_{j = J_2+1}^{J_2'} \theta_j^{(2)}A_j^{(2)^2}\right]  +E\left[\sum_{j = J_3+1}^{J_3'} \theta_j^{(3)}A_j^{(3)^2}\right] \\
    &= \sum_{j = J_1+1}^{J_1'} \theta_j^{(1)} + \sum_{j = J_2+1}^{J_2'} \theta_j^{(2)} + \sum_{j = J_3+1}^{J_3'} \theta_j^{(3)} \\
    &\leq \frac{\delta^2}{3C^2} + \frac{\delta^2}{3C^2} + \frac{\delta^2}{3C^2} = \frac{\delta^2}{C^2}.
\end{align*}
Finally, by the Cauchy Schwartz inequality,
\begin{equation*}
    E\left[\left\Vert \sum_{k=1}^3 \sum_{j = J_k+1}^{J_k'} \theta_j^{(k)^{1/2}}A_j^{(k)}\psi_j^{(k)}(\cdot) \right\Vert_2 \right] \leq \sqrt{E\left[\left\Vert \sum_{k=1}^3 \sum_{j = J_k+1}^{J_k'} \theta_j^{(k)^{1/2}}A_j^{(k)}\psi_j^{(k)}(\cdot) \right\Vert_2^2 \right]} \leq \frac{\delta}{C}.
\end{equation*}
Given the above derivations, we have that,
\begin{align*}
    |\mu_{\bm{J}'} &- \mu_{\bm{J}}| = \left\vert \int_{\mathcal{\bm{X}}} \int_{\mathbb{R}^{\bm{J}'}} m\left\{\bm{x}, a_{\bm{J}'}(\cdot)\right\}f_{\bm{A}^{(1)}_{J_1'}, \bm{A}^{(2)}_{J_2'}, \bm{A}^{(3)}_{J_3'}|\bm{X}}(\bm{a}^{(1)}_{J_1'}, \bm{a}^{(2)}_{J_2'}, \bm{a}^{(3)}_{J_3'}|\bm{x})d\bm{a}^{(1)}_{J_1'}d \bm{a}^{(2)}_{J_2'}d\bm{a}^{(3)}_{J_3'} \right. f_{\bm{X}}(\bm{x})d\bm{x} \\
`   &\quad- \int_{\mathcal{\bm{X}}}\int_{\mathbb{R}^{\bm{J}}} m\left\{\bm{x}, a_{\bm{J}}(\cdot)\right\}f_{\bm{A}^{(1)}_{J_1}, \bm{A}^{(2)}_{J_2}, \bm{A}^{(3)}_{J_3}|\bm{X}}(\bm{a}^{(1)}_{J_1}, \bm{a}^{(2)}_{J_2}, \bm{a}^{(3)}_{J_3}|\bm{x})d\bm{a}^{(1)}_{J_1}, d\bm{a}^{(2)}_{J_2}d\bm{a}^{(3)}_{J_3} \\
         & \quad\times \int_{\mathbb{R}^{\bm{J}' - \bm{J}}} f_{\bm{A}^{(1)}_{J_1+1:J_1'}, \bm{A}^{(2)}_{J_2+1:J_1'}, \bm{A}^{(3)}_{J_3+1:J_3'}| \bm{A}^{(1)}_{J_1}, \bm{A}^{(2)}_{J_2}, \bm{A}^{(3)}_{J_3}, \bm{X}}(\bm{a}^{(1)}_{J_1+1:J_1'}, \bm{a}^{(2)}_{J_2+1:J_2'}, \bm{a}^{(3)}_{J_3+1:J_3'}| \bm{a}^{(1)}_{J_1}, \bm{a}^{(2)}_{J_2}, \bm{a}^{(3)}_{J_3}, \bm{x})\\
         & \left. \quad \times d\bm{a}^{(1)}_{J_1+1:J_1'}d \bm{a}^{(2)}_{J_2+1:J_2'}d\bm{a}^{(3)}_{J_3+1:J_3'}f_{\bm{X}}(\bm{x})d\bm{x} \right\vert \\
         &\leq  \int_{\mathcal{\bm{X}}} \int_{\mathbb{R}^{\bm{J}'}}\left\vert m\left\{\bm{x}, a_{\bm{J}'}(\cdot)\right\} - m\left\{\bm{x}, a_{\bm{J}}(\cdot)\right\}\right\vert  \times f_{\bm{A}^{(1)}_{J_1'}, \bm{A}^{(2)}_{J_2'}, \bm{A}^{(3)}_{J_3'}|\bm{X}}(\bm{a}^{(1)}_{J_1'}, \bm{a}^{(2)}_{J_2'}, \bm{a}^{(3)}_{J_3'}|\bm{x})d\bm{a}^{(1)}_{J_1'}d\bm{a}^{(2)}_{J_2'}d\bm{a}^{(3)}_{J_3'} f_{\bm{X}}(\bm{x})d\bm{x}  \\
         &\leq \int_{\mathcal{\bm{X}}} \int_{\mathbb{R}^{\bm{J}'}} C\left\Vert \sum_{k=1}^3\sum_{j=J_k + 1}^{J_k'} \theta^{(k)^{1/2}}_ja_j^{(k)}\psi_j^{(k)} \right\Vert_2 f_{\bm{A}^{(1)}_{J_1'}, \bm{A}^{(2)}_{J_2'}, \bm{A}^{(3)}_{J_3'}|\bm{X}}(\bm{a}^{(1)}_{J_1'}, \bm{a}^{(2)}_{J_2'}, \bm{a}^{(3)}_{J_3'}|\bm{x})d\bm{a}^{(1)}_{J_1'}d \bm{a}^{(2)}_{J_2'}d\bm{a}^{(3)}_{J_3'} f_{\bm{X}}(\bm{x})d\bm{x}  \\
         &= \int_{\mathcal{\bm{X}} \times \mathbb{R}^{\bm{J}'}}C\left\Vert \sum_{k=1}^3\sum_{j=J_k + 1}^{J_k'} \theta^{(k)^{1/2}}_ja_j^{(k)}\psi_j^{(k)} \right\Vert_2 f_{\bm{A}^{(1)}_{J_1'}, \bm{A}^{(2)}_{J_2'}, \bm{A}^{(3)}_{J_3'}, \bm{X}}(\bm{a}^{(1)}_{J_1'}, \bm{a}^{(2)}_{J_2'}, \bm{a}^{(3)}_{J_3'},\bm{x})d\bm{a}^{(1)}_{J_1'}d\bm{a}^{(2)}_{J_2'}d\bm{a}^{(3)}_{J_3'} (\bm{x})d\bm{x} \\
         &= CE\left[\left\Vert \sum_{k=1}^3\sum_{j=J_k + 1}^{J_k'} \theta^{(k)^{1/2}}_jA_j^{(k)}\psi_j^{(k)} \right\Vert_2 \right] \\
         &\leq \delta.
\end{align*}
Thus, $\{\mu_{\bm{J}}\}_{\bm{J} = 1}^\infty$ is a Cauchy sequence and $\lim_{\bm{J} \to \infty} \mu_{\bm{J}}$ exists as desired.

Next, we will prove that $E[Y] = \lim_{\bm{J} \to \infty} \mu_{\bm{J}}$ by showing that $\lim_{\bm{J} \to \infty}E[Y - \mu_{\bm{J}}] = 0$. For $\bm{J}' > \bm{J}$ let 
    \begin{align*}
        R_{\bm{J}', \bm{J}} &= E\left(E\left[m\left\{\bm{X}, A_{\bm{J}'}(\cdot)\right\} - m\left\{\bm{X}, A_{\bm{J}}(\cdot)\right\}\right.\right. + \epsilon \left.\left. | \bm{X}, \bm{A}_{J_1}^{(1)}, \bm{A}_{J_2}^{(2)}, \bm{A}_{J_3}^{(3)}\right]\right).
    \end{align*}
    Note that since $Y = m\{\bm{X}, A(\cdot)\} + \epsilon$,
    \begin{align*}
        \lim_{\bm{J}'\to \infty}R_{\bm{J}', \bm{J}} &= E\left(E\left[Y- m\left\{\bm{X}, A_{\bm{J}}(\cdot)\right\} | \bm{X}, \bm{A}_{J_1}^{(1)}, \bm{A}_{J_2}^{(2)}, \bm{A}_{J_3}^{(3)}\right]\right) \\
        &= E\left[Y- m\left\{\bm{X}, A_{\bm{J}}(\cdot)\right\}\right] \\
        &= E[Y - \mu_{\bm{J}}]
    \end{align*}
    Thus, it suffices to show that $\lim_{\bm{J}\to \infty}\lim_{\bm{J}'\to \infty}R_{\bm{J}', \bm{J}} = 0$. Next,
    \begin{align*}
       R_{\bm{J}', \bm{J}} &= E\left(E\left[m\left\{\bm{X}, A_{\bm{J}'}(\cdot)\right\} - m\left\{\bm{X}, A_{\bm{J}}(\cdot)\right\}\right.\right. \left.\left. | \bm{X}, \bm{A}_{J_1}^{(1)}, \bm{A}_{J_2}^{(2)}, \bm{A}_{J_3}^{(3)}\right]\right). \\
        &\leq E\left(E\left[\left\vert m\left\{\bm{X}, A_{\bm{J}'}(\cdot)\right\} - m\left\{\bm{X}, A_{\bm{J}}(\cdot)\right\}\right\vert\right.\right. \left.\left. | \bm{X}, \bm{A}_{J_1}^{(1)}, \bm{A}_{J_2}^{(2)}, \bm{A}_{J_3}^{(3)}\right]\right). \\
        &\leq E\left(E\left[\left\Vert \sum_{k=1}^3\sum_{j=J_k + 1}^{J_k'} \theta^{(k)^{1/2}}_jA_j^{(k)}\psi_j^{(k)} \right\Vert_2| \bm{X}, \bm{A}_{J_1}^{(1)}, \bm{A}_{J_2}^{(2)}, \bm{A}_{J_3}^{(3)}\right]\right) \\
        &= E\left(E\left[\sqrt{\sum_{k=1}^3\sum_{j=J_k + 1}^{J_k'} \theta^{(k)}_jA_j^{(k)^2}} | \bm{X}, \bm{A}_{J_1}^{(1)}, \bm{A}_{J_2}^{(2)}, \bm{A}_{J_3}^{(3)}\right]\right)  \\
        &= CE\left[\sqrt{\sum_{k=1}^3\sum_{j=J_k + 1}^{J_k'} \theta^{(k)}_jA_j^{(k)^2}}\right] \\
        &\leq C\sqrt{E\left[\sum_{k=1}^3\sum_{j=J_k + 1}^{J_k'} \theta^{(k)}_jA_j^{(k)^2}\right]} = C\sqrt{\sum_{k=1}^3\sum_{j=J_k + 1}^{J_k'} \theta^{(k)}_j}
    \end{align*}
    where the first equality follows from the fact that $E[\epsilon] = 0$ by definition, the second inequality follows from the fact $E[X] \leq E[|X|]$, the third inequality follows from condition \textbf{C2}, the fourth equality follows from earlier derivations in this proof, and the sixth inequality is due to Jensen's inequality. Then by condition \textbf{C1},
    \begin{align*}
        \lim_{\bm{J}\to \infty}\lim_{\bm{J}'\to \infty}C\sqrt{\sum_{k=1}^3\sum_{j=J_k + 1}^{J_k'} \theta^{(k)}_j} = 0
    \end{align*}
    and $\lim_{\bm{J} \to \infty}E[Y - \mu_{\bm{J}}] = 0$ as desired. 
\end{proof}

\subsection{Corollary \ref{thm:corollary_pop_avg}}
\begin{suppcoro}
    Given conditions \textbf{C1} and \textbf{C2},  $\lim_{\bm{J} \to \infty} \tilde{\mu}_{\bm{J}}$ exists and $E[Y] = \lim_{\bm{J} \to \infty} \tilde{\mu}_{\bm{J}}$.\label{thm:corollary_pop_avg}
\end{suppcoro}
\begin{proof}
    This result follows trivially from Lemma \ref{lemma:novel_basis} and the same arguments as Lemma \ref{thm:pop_avg_limit}.
\end{proof}

\subsection{Theorem \ref{thm:outcome_reg}}
\begin{supptheorem}
    \textnormal{(Outcome model identification)} Given conditions \textbf{C1}, \textbf{C2}, \textbf{C$\tilde{\bm{3}}$} and assumption \textbf{A1}, assumption \textbf{A3} (for $\tilde{q} = q$),  or \textbf{A$\bm{3}^*$} (for $\tilde{q} = q^*$)
    \begin{align*}
        \mu^{\tilde{Q}} &=  \lim_{\bm{J} \to \infty} E[E_{\tilde{Q}}[m(\bm{X}, A_{\bm{J}}(\cdot))|\bm{X},  A_{\bm{J}}^{-}(\cdot)]], \:\:\: \text{for}\:\: \tilde{q} = q, q^*.
    \end{align*}
    \label{thm:outcome_reg}
\end{supptheorem}

\begin{proof}
The proofs for $\tilde{Q} = Q, Q^*$ differ, so we will proceed by defining terms specific to each of these distributions. We begin by defining the the following conditional expectations of our potential outcomes of interest,
\begin{align*}
    m^{Q}(\bm{x}, a(\cdot)) &= E[Y(A^{Q})|\bm{X} = \bm{x}, A^{-}(\cdot) = a(\cdot)], \\
    m^{Q^*}(\bm{x}, a(\cdot)) &= E[Y(A^{Q^*})|\bm{X} = \bm{x}, A^{-[0, t_2]}(\cdot) = a(\cdot)].
\end{align*}
Next, we define the expectations of $m^{Q}(\cdot)$ and $m^{Q^*}(\cdot)$ given the truncated basis approximation of $A^{-}(\cdot)$ as $\mu^{Q}_{\bm{J}} = E\left[m^{Q}(\bm{X}, A^{-}_{\bm{J}}(\cdot))\right]$ and $\mu^{Q^*}_{\bm{J}} = E\left[m^{Q^*}(\bm{X}, A^{-[0, t_2]}_{J_1, J_2}(\cdot))\right]$. Given that conditions \textbf{C1} and \textbf{C2} hold for $m^{Q}$ and $ m^{Q^*}$, by replacing $Y$ with $Y\{A^{Q}(\cdot)\}$ or $Y\{A^{Q^*}(\cdot)\}$ in Lemma \ref{thm:pop_avg_limit} and Corollary \ref{thm:corollary_pop_avg} we have that
\begin{align*}
    \lim_{\bm{J} \to \infty}\mu^{Q}_{\bm{J}} &= E[m^{Q}(\bm{x}, a(\cdot))] = E[Y\{A^{Q}(\cdot)\}], \\
    \lim_{J_1, J_2 \to \infty}\mu^{Q^*}_{\bm{J}} &= E[m^{Q^*}(\bm{x}, a(\cdot))] = E[Y\{A^{Q^*}(\cdot)\}]. 
\end{align*}
    Thus, it remains to identify $\mu^{Q}_{\bm{J}}$ and $\mu^{Q^*}_{\bm{J}}$ with the observed data and show their limits exist.

    First, we will identify $\mu^{Q}_{\bm{J}}$ with the observed data. To simplify notation, we write the $Y\{A^{Q}(\cdot)\}$ as the following:
    \begin{align*}
        Y\{A^{Q}(\cdot)\} &= Y(Q^{(2)}_1\gamma_1^{(2)}+ A^{-}(\cdot)) \\
        &= Y(Q^{(2)}_1, A^{-}(\cdot))  \\
        Y(Q^{(2)}_1, a^{-}(\cdot)) &= Y(Q^{(2)}_1)
    \end{align*}
    Next, 
    \begin{align*}
        m^Q(\bm{x}, a(\cdot)) 
        &= \int_{\mathcal{B}^{(2)}_1} E[Y(q_1^{(2)})|Q_1^{(2)} = q_1^{(2)}, \bm{X} = \bm{x}, A^{-}(\cdot) = a(\cdot)]dQ(q_1^{(2)}|\bm{x}, a(\cdot)) \\
        &= \int_{\mathcal{B}^{(2)}_1} E[Y(q_1^{(2)})|B_1^{(2)} = q_1^{(2)}, \bm{X} = \bm{x}, A^{-}(\cdot) = a(\cdot)]dQ(q_1^{(2)}|\bm{x}, a(\cdot)) \\
        &= \int_{\mathcal{B}^{(2)}_1} E[Y(b_1^{(2)})|B_1^{(2)} = b_1^{(2)}, \bm{X} = \bm{x}, A^{-}(\cdot) = a(\cdot)]dQ(b_1^{(2)}|\bm{x}, a(\cdot))  \\
        &= \int_{\mathcal{B}^{(2)}_1} E[Y|B_1^{(2)} = b_1^{(2)}, \bm{X} = \bm{x}, A^{-}(\cdot) = a(\cdot)]dQ(b_1^{(2)}|\bm{x}, a(\cdot))  \\
        &= \int_{\mathcal{B}^{(2)}_1}m\left(\bm{x}, a(\cdot) + b_1^{(2)}\gamma_1^{(2)}(\cdot) \right)dQ(b_1^{(2)}|\bm{x}, a(\cdot))
    \end{align*}
    where the first equality follows from iterated expectations, the second equality follows from the fact that the fact that $Q_1^{(2)} \indep Y\{q^{(2)}_1)\}$ by definition and assumption \textbf{A3}, the third equality follows from replacing $q_1^{(2)}$ with $b_1^{(2)}$, the fourth equality follows from assumption \textbf{A1}, and the fifth equality from the definition of $m(\cdot)$.
    
    Then we have that, 
    \begin{align*}
        &|m^Q(\bm{x}, a_1(\cdot)) - m^Q(\bm{x}, a_2(\cdot))| \\
        &= \left\vert\int_{\mathcal{B}^{(2)}_1}m\left(\bm{x}, a_1(\cdot) + b_1^{(2)}\gamma_1^{(2)}(\cdot) \right)dQ(b_1^{(2)}|\bm{x}, a_1(\cdot))\right. - \left.\int_{\mathcal{B}^{(2)}_1}m\left(\bm{x}, a_2(\cdot) + b_1^{(2)}\gamma_1^{(2)}(\cdot) \right)dQ(b_1^{(2)}|\bm{x}, a_2(\cdot))\right\vert \\
        &= \left\vert\int_{\mathcal{B}^{(2)}_1}m\left(\bm{x}, a_1(\cdot) + b_1^{(2)}\gamma_1^{(2)}(\cdot) \right)\{dQ(b_1^{(2)}|\bm{x}, a_1(\cdot)) - dQ(b_1^{(2)}|\bm{x}, a_2(\cdot))\} \right. \\
        &\quad - \left.\int_{\mathcal{B}^{(2)}_1}\{m\left(\bm{x}, a_2(\cdot) + b_1^{(2)}\gamma_1^{(2)}(\cdot) \right) - m\left(\bm{x}, a_1(\cdot) + b_1^{(2)}\gamma_1^{(2)}(\cdot) \right)\}dQ(b_1^{(2)}|\bm{x}, a_2(\cdot))\right\vert \\
        &\leq \underbrace{\left\vert\int_{\mathcal{B}^{(2)}_1}m\left(\bm{x}, a_1(\cdot) + b_1^{(2)}\gamma_1^{(2)}(\cdot) \right)\{dQ(b_1^{(2)}|\bm{x}, a_1(\cdot)) - dQ(b_1^{(2)}|\bm{x}, a_2(\cdot))\} \right\vert}_{\Circled{1}} \\
        &\quad+ \underbrace{\left\vert\int_{\mathcal{B}^{(2)}_1}\{m\left(\bm{x}, a_2(\cdot) + b_1^{(2)}\gamma_1^{(2)}(\cdot) \right) - m\left(\bm{x}, a_1(\cdot) + b_1^{(2)}\gamma_1^{(2)}(\cdot) \right)\}dQ(b_1^{(2)}|\bm{x}, a_2(\cdot))\right\vert}_{\Circled{2}}
    \end{align*}
    Next by conditions \textbf{C2} and \textbf{C3} we have that,
    \begin{align*}
        \Circled{1} &= \left\vert E_Q\left[m\left(\bm{x}, a_1(\cdot) + B_1^{(2)}\gamma_1^{(2)}(\cdot) \right)|\bm{x}, a_1(\cdot)\right] - E_Q\left[m\left(\bm{x}, a_1(\cdot) + B_1^{(2)}\gamma_1^{(2)}(\cdot) \right)|\bm{x}, a_2(\cdot)\right] \right\vert \\
        &\leq C_Q||a_1(\cdot) - a_2(\cdot)||_2 \\
        \Circled{2} &\leq \int_{\mathcal{B}^{(2)}_1}\left\vert m\left(\bm{x}, a_2(\cdot) + b_1^{(2)}\gamma_1^{(2)}(\cdot) \right) - m\left(\bm{x}, a_1(\cdot) + b_1^{(2)}\gamma_1^{(2)}(\cdot) \right)\right\vert dQ(b_1^{(2)}|\bm{x}, a_2(\cdot)) \\
        &\leq \int_{\mathcal{B}^{(2)}_1} C||a_1(\cdot) - a_2(\cdot)||_2dQ(b_1^{(2)}|\bm{x}, a_2(\cdot)) \\
        &= C||a_1(\cdot) - a_2(\cdot)||_2,
    \end{align*}
    such that,
    \begin{align*}
        |m^Q(\bm{x}, a_1(\cdot)) - m^Q(\bm{x}, a_2(\cdot))| \leq (C + C_Q)||a_1(\cdot) - a_2(\cdot)||_2.
    \end{align*}
    Thus, $m^Q$ satisfies condition \textbf{C2} and $\lim_{\bm{J} \to \infty} \mu^{Q}_{\bm{J}}$ exists through a similar argument to the proof of Lemma 2. 
    Therefore we have that
    \begin{align*}
        E[Y\{A^Q(\cdot)\}] &= \lim_{\bm{J} \to \infty} \mu^{Q}_{\bm{J}} \\
        &= \lim_{\bm{J} \to \infty} E\left[m^{Q}\left\{\bm{X}, A^{-}_{\bm{J}}(\cdot)\right\}\right] \\
        &= \lim_{\bm{J} \to \infty} E\left[ \int_{\mathcal{B}_1^{(2)}} m\left(\bm{X},  b^{(2)}_1\gamma_j^{(2)} + A_{\bm{J}}^{-}(\cdot)\right)dQ(b^{(2)}_1|\bm{X}, A^{-}_{\bm{J}}(\cdot)) \right] 
    \end{align*}
    as desired. \\ 
    \vspace{0.5in} \\
    Now, we will identify $\mu^{Q^*}_{\bm{J}}$ with the observed data. To simplify notation, we write the $Y\{A^{Q^*}(\cdot)\}$ as the following:
    \begin{align*}
        Y\{A^{Q^*}(\cdot)\} &= Y\{A^{-[0, t_2]}(\cdot) + Q^{(2)*}_1\gamma_1^{(2)}+ A^{(3)}(\cdot)\} \\
        &= Y\{A^{-[0, t_2]}(\cdot), Q^{(2)*}_1, A^{(3)}\} \\
        Y\{a^{-[0, t_2]}(\cdot), Q^{(2)*}_1, A^{(3)}(\cdot)\} &= Y\{Q^{(2)*}_1,A^{(3)}(\cdot)\}
    \end{align*}
    Next, 
    \begin{align*}
        m^{Q^*}(\bm{x}, a(\cdot)) 
        &= \int_{\mathcal{B}^{(2)}_1} E[Y\{q^{(2)}_1, A^{(3)}(\cdot)\}|Q_1^{(2)*} = q_1^{(2)}, \bm{X} = \bm{x}, A^{-[0, t_2]}(\cdot) = a(\cdot)]dQ^*(q_1^{(2)}|\bm{x}, a(\cdot)) \\
        &= \int_{\mathcal{B}^{(2)}_1} E[Y\{b^{(2)}_1, A^{(3)}(\cdot)\}|B_1^{(2)} = b_1^{(2)}, \bm{X} = \bm{x}, A^{-[0, t_2]}(\cdot) = a(\cdot)]dQ^*(b_1^{(2)}|\bm{x}, a(\cdot)) \\
        &= \int_{\mathcal{A}^{(3)}(\cdot)}\int_{\mathcal{B}^{(2)}_1} E[Y\{b^{(2)}_1, a^{(3)}(\cdot)\}|A^{(3)}(\cdot) = a^{(3)}(\cdot), B_1^{(2)} = b_1^{(2)}, \bm{X} = \bm{x}, A^{-[0, t_2]}(\cdot)= a(\cdot)]\\
        &\quad\times dF(a^{(3)}(\cdot)| b_1^{(2)},\bm{x}, a(\cdot)))dQ^*(b_1^{(2)}|\bm{x}, a(\cdot)) \\
        &= \int_{\mathcal{A}^{(3)}(\cdot)}\int_{\mathcal{B}^{(2)}_1} E[Y|A^{(3)}(\cdot) = a^{(3)}(\cdot), B_1^{(2)} = b_1^{(2)}, \bm{X} = \bm{x}, A^{-[0, t_2]}(\cdot) = a(\cdot)]\\
        &\quad\times dF(a^{(3)}(\cdot)| b_1^{(2)},\bm{x}, a(\cdot)))dQ^*(b_1^{(2)}|\bm{x}, a(\cdot)) \\
         &:= \lim_{J_3 \to \infty} \int_{\mathcal{A}_{J_3}^{(3)}(\cdot)}\int_{\mathcal{B}^{(2)}_1} E[Y|A_{J_3}^{(3)}(\cdot) = a_{J_3}^{(3)}(\cdot), B_1^{(2)} = b_1^{(2)}, \bm{X} = \bm{x}, A^{-[0, t_2]}(\cdot) = a(\cdot)]\\
        &\quad\times dF(a_{J_3}^{(3)}(\cdot)| b_1^{(2)},\bm{x}, a(\cdot)))dQ^*(b_1^{(2)}|\bm{x}, a(\cdot)) \\
        &=\lim_{J_3 \to \infty}\int_{\mathcal{A}_{J_3}^{(3)}(\cdot)}\int_{\mathcal{B}^{(2)}_1} m(\bm{x}, b_1^{(2)}\gamma_1^{(2)} + a(\cdot) + a_{J_3}^{(3)}(\cdot)) dF(a_{J_3}^{(3)}(\cdot)| \bm{x}, a(\cdot)))dQ^*(b_1^{(2)}|\bm{x}, a(\cdot)) \\
        &:= \lim_{J_3 \to \infty} \tilde{m}_{J_3}(\bm{x}, a(\cdot)) 
    \end{align*}
    where the first equality follows from iterated expectations, the second equality follows from the fact that $Q_1^{(2)} \indep Y\{q^{(2)*}_1, A_{J_3}^{(3)}(\cdot)\}$ by definition, assumption \textbf{A$\bm{3}^*$}, and by replacing $q_1^{(2)}$ with $b_1^{(2)}$, the third equality follows from iterated expectations, the fourth equality follows from assumption \textbf{A1}, 
    and the sixth equality from the definition of $m(\cdot)$ and the fact that $A_{J_3}^{(3)}(\cdot) \indep B_1^{(2)} | \bm{X}, A^{-[0, t_2]}(\cdot)$ from assumption \textbf{A$\bm{3}^*$}. Note that in the fifth equality we define this quantity with the limit as $J_3 \to \infty$ because $A^{(3)}(\cdot)$ does not have a probability density.

     We need to show that $\lim_{J_1, J_2 \to \infty} \mu_{J_1, J_2}^{Q^*}$ exists. To do this, we will first show that $\{\tilde{m}_{J_3}\}_{J_3 = 1}^\infty$ is Cauchy in L1 such that $E[\lim_{J_3 \to \infty} \tilde{m}_{J_3}] = \lim_{J_3 \to \infty}E[ \tilde{m}_{J_3}]$. We must show that for any $\delta > 0$ there exists an integer $N$ such that for any $J' > J > N$, $E[|\tilde{m}_{J'} - \tilde{m}_{J}|] \leq \delta$. 
    By similar arguments to the proof of Lemma 2, 
    \begin{align*}
       E[|\tilde{m}_{J'}(\bm{X}, A(\cdot))  &- \tilde{m}_{J}(\bm{X}, A(\cdot)) |] \\
        &= E\left[\left\vert\int_{\mathcal{B}^{(2)}_1}\int_{\mathcal{A}_{J'}^{(3)}(\cdot)}m\left(\bm{X}, b_1^{(2)}\gamma_1^{(2)} + A(\cdot)+ A_{J'}^{(3)}(\cdot)\right) \right. \right.\\
        &\quad -  \left.\left. m\left(\bm{X}, b_1^{(2)}\gamma_1^{(2)} + A(\cdot)+ A_{J}^{(3)}(\cdot)\right) dF(A_{J'}^{(3)}(\cdot)|\bm{X}, A(\cdot)))dQ^*(b_1^{(2)}|\bm{X}, A(\cdot)) \right\vert \right]\\
        &\leq E\left[\int_{\mathcal{B}^{(2)}_1}\int_{\mathcal{A}_{J'}^{(3)}(\cdot)} C\left\Vert\sum_{j=J+1}^{J'} \theta_j^{(3)^{1/2}}A^{(3)}_j\psi_j^{(3)}(\cdot)\right\Vert_2dF(A_{J'}^{(3)}(\cdot)|\bm{X}, A(\cdot)))dQ^*(b_1^{(2)}|\bm{X}, A(\cdot)) \right]\\
        &= CE\left[E\left[ \left\Vert\sum_{j=J+1}^{J'} \theta_j^{(3)^{1/2}}A^{(3)}_j\psi_j^{(3)}(\cdot)\right\Vert_2 \middle\vert\bm{X}, A(\cdot)\right]\right] \\
        &= CE\left[ \left\Vert\sum_{j=J+1}^{J'} \theta_j^{(3)^{1/2}}A^{(3)}_j\psi_j^{(3)}(\cdot)\right\Vert_2\right] \\
        &\leq \delta/\sqrt{3} < \delta
    \end{align*}
    where the last inequality follows from the proof of Lemma 2. Then, $\{\tilde{m}_{J_3}\}_{J_3 = 1}^\infty$ is Cauchy in L1 and we have that $E[\lim_{J_3 \to \infty} \tilde{m}_{J_3}] = \lim_{J_3 \to \infty}E[ \tilde{m}_{J_3}]$. Next, we return to showing that $\lim_{J_1, J_2 \to \infty} \mu_{J_1, J_2}^{Q^*}$ exists. Given Lemma \ref{thm:pop_avg_limit} and Corollary \ref{thm:corollary_pop_avg}, it suffices show that $m^{Q^*}$ satisfies condition \textbf{C2}. However, note that,
   \begin{align*}
       \lim_{J_1, J_2 \to \infty} \mu_{J_1, J_2}^{Q^*} &= E[m^{Q^*}(\bm{X}, A(\cdot))] \\
       &= E[\lim_{J_3 \to \infty} \tilde{m}_{J_3}(\bm{X}, A(\cdot))] \\
       &= \lim_{J_3 \to \infty} E[\tilde{m}_{J_3}(\bm{X}, A(\cdot))] \\
       &= \lim_{J_3 \to \infty} E\left[\int_{\mathcal{A}_{J_3}^{(3)}(\cdot)}\int_{\mathcal{B}^{(2)}_1} m(\bm{X}, b_1^{(2)}\gamma_1^{(2)} + A(\cdot) + A_{J_3}^{(3)}(\cdot)) dF(A_{J_3}^{(3)}(\cdot)| \bm{X}, A(\cdot)))\right.\\
       &\quad \times\left. dQ^*(b_1^{(2)}|\bm{X}, A(\cdot)) \right] \\
       &= \lim_{J_3 \to \infty} E\left[\int_{\mathcal{B}^{(2)}_1} m(\bm{X}, b_1^{(2)}\gamma_1^{(2)} + A(\cdot) + A_{J_3}^{(3)}(\cdot))dQ^*(b_1^{(2)}|\bm{X}, A(\cdot)) \right] \\
       &:= E\left[\int_{\mathcal{B}^{(2)}_1} m(\bm{X}, b_1^{(2)}\gamma_1^{(2)} + A(\cdot) + A^{(3)}(\cdot))dQ^*(b_1^{(2)}|\bm{X}, A(\cdot)) \right] \\
       &:= E\left[\int_{\mathcal{B}^{(2)}_1} m(\bm{x}, b_1^{(2)}\gamma_1^{(2)} + A(\cdot) + A^{(3)}(\cdot))dQ^*(b_1^{(2)*}|\bm{X}, A(\cdot), A^{(3)}(\cdot)) \right]  \\
       &= E[\tilde{m}^{Q^*}(\bm{X}, A^{-}(\cdot))]
   \end{align*}
   where note that $\tilde{m}^{Q^*}(\bm{x}, a^{-}(\cdot)) = E_{Q^*}[Y|\bm{X} = \bm{x}, A^{-}(\cdot) =  a^{-}(\cdot)]$. The third equality follows because $\{\tilde{m}_{J_3}\}_{J_3 = 1}^\infty$ is Cauchy in L1 and the sixth equality follows from the fact that $q^*(b_1^{(2)}|\bm{X}, A(\cdot)) = q^*(b_1^{(2)}|\bm{X}, A(\cdot), A^{(3)}(\cdot))$ by the definition of $q^*$. Since we have established that $\{\tilde{m}_{J_3}\}_{J_3 = 1}^\infty$ is Cauchy in L1, the limit in fourth and fifth equality exists.
   Thus, it remains to show that $\tilde{m}^{Q^*}$ satisfies condition \textbf{C2}.
     Let $a_1(\cdot)$ and $a_2(\cdot)$ be square integral functions over $[0, t_2]$. Then we have that,
    \begin{align*}
        &|\tilde{m}^{Q^*}(\bm{x}, a_1(\cdot)) - \tilde{m}^{Q^*}(\bm{x}, a_2(\cdot))| \\
        &= \left\vert\int_{\mathcal{B}^{(2)}_1}m\left(\bm{x}, a_1(\cdot) + b_1^{(2)}\gamma_1^{(2)}(\cdot) \right)dQ^*(b_1^{(2)}|\bm{x}, a_1(\cdot))\right. - \left.\int_{\mathcal{B}^{(2)}_1}m\left(\bm{x}, a_2(\cdot) + b_1^{(2)}\gamma_1^{(2)}(\cdot) \right)dQ^*(b_1^{(2)}|\bm{x}, a_2(\cdot))\right\vert \\
        &= \left\vert\int_{\mathcal{B}^{(2)}_1}m\left(\bm{x}, a_1(\cdot) + b_1^{(2)}\gamma_1^{(2)}(\cdot) \right)\{dQ^*(b_1^{(2)}|\bm{x}, a_1(\cdot)) - dQ^*(b_1^{(2)}|\bm{x}, a_2(\cdot))\} \right. \\
        &\quad - \left.\int_{\mathcal{B}^{(2)}_1}\{m\left(\bm{x}, a_2(\cdot) + b_1^{(2)}\gamma_1^{(2)}(\cdot) \right) - m\left(\bm{x}, a_1(\cdot) + b_1^{(2)}\gamma_1^{(2)}(\cdot) \right)\}dQ^*(b_1^{(2)}|\bm{x}, a_2(\cdot))\right\vert \\
        &\leq \underbrace{\left\vert\int_{\mathcal{B}^{(2)}_1}m\left(\bm{x}, a_1(\cdot) + b_1^{(2)}\gamma_1^{(2)}(\cdot) \right)\{dQ^*(b_1^{(2)}|\bm{x}, a_1(\cdot)) - dQ^*(b_1^{(2)}|\bm{x}, a_2(\cdot))\} \right\vert}_{\Circled{1}} \\
        &\quad+ \underbrace{\left\vert\int_{\mathcal{B}^{(2)}_1}\{m\left(\bm{x}, a_2(\cdot) + b_1^{(2)}\gamma_1^{(2)}(\cdot) \right) - m\left(\bm{x}, a_1(\cdot) + b_1^{(2)}\gamma_1^{(2)}(\cdot) \right)\}dQ^*(b_1^{(2)}|\bm{x}, a_2(\cdot))\right\vert}_{\Circled{2}}
    \end{align*}
    Next by conditions \textbf{C2} and \textbf{C4} we have that,
    \begin{align*}
        \Circled{1} &= \left\vert E_{Q^*}\left[m\left(\bm{x}, a_1(\cdot) + B_1^{(2)}\gamma_1^{(2)}(\cdot) \right)|\bm{x}, a_1(\cdot)\right] - E_{Q^*}\left[m\left(\bm{x}, a_1(\cdot) + B_1^{(2)}\gamma_1^{(2)}(\cdot) \right)|\bm{x}, a_2(\cdot)\right] \right\vert \\
        &\leq C_{Q^*}||a_1(\cdot) - a_2(\cdot)||_2 \\
        \Circled{2} &\leq \int_{\mathcal{B}^{(2)}_1}\left\vert m\left(\bm{x}, a_2(\cdot) + b_1^{(2)}\gamma_1^{(2)}(\cdot) \right) - m\left(\bm{x}, a_1(\cdot) + b_1^{(2)}\gamma_1^{(2)}(\cdot) \right)\right\vert dQ^*(b_1^{(2)}|\bm{x}, a_2(\cdot)) \\
        &\leq \int_{\mathcal{B}^{(2)}_1} C||a_1(\cdot) - a_2(\cdot)||_2dQ^*(b_1^{(2)}|\bm{x}, a_2(\cdot)) \\
        &= C||a_1(\cdot) - a_2(\cdot)||_2,
    \end{align*}
    
    Thus,
    \begin{align*}
        |&\tilde{m}^{Q^*}(\bm{x}, a_1(\cdot)) - \tilde{m}^{Q^*}(\bm{x}, a_2(\cdot))| \leq (C + C_{Q^*})\left\Vert a_1(\cdot) - a_2(\cdot)\right\Vert_2
    \end{align*}
    Therefore, $\tilde{m}^{Q^*}$ satisfies condition \textbf{C2} and $\lim_{\bm{J_1, J_2} \to \infty} \mu^{Q^*}_{J_1, J_2}$ exists through a similar argument to the proof of Lemma 2. 
    Thus, we have that
    \begin{align*}
        E[Y\{A^{Q^*}(\cdot)\}] &= \lim_{J_1, J_2 \to \infty} \mu^{Q^*}_{J_1, J_2} \\
        &= \lim_{\bm{J} \to \infty} E[\tilde{m}^{Q^*}(\bm{X}, A_{\bm{J}}^{-}(\cdot))]\\
        &= \lim_{\bm{J} \to \infty} E\left[\int_{\mathcal{B}_1^{(2)}} m\left(\bm{X},  b^{(2)}_1\gamma_j^{(2)} + A_{\bm{J}}^{-}(\cdot)\right)dQ^*(b^{(2)}_1|\bm{X}, A_{\bm{J}}^{-[0, t_2]}(\cdot)) \right] 
    \end{align*}
    as desired. \\ 
\end{proof}

\subsection{Theorem \ref{thm:weighting}}

\begin{supptheorem}
    \textnormal{(Weighting identification)} Given the conditions and assumptions of Theorem \ref{thm:outcome_reg}, condition \textbf{C$\tilde{\bm{4}}$}, and assumption \textbf{A$\tilde{\bm{2}}$},
    \begin{align*}
        \mu^{\tilde{Q}} &= \lim_{\bm{J} \to \infty} E\left[ \frac{\tilde{q}(B_1^{(2)}| \bm{X}, A_{\bm{J}}^{-\tilde{Q}}(\cdot))}{f(B_1^{(2)}| \bm{X}, A_{\bm{J}}^{-\tilde{Q}}(\cdot))}Y\right],\:\:\: \text{for}\:\: \tilde{q} = q, q^*.
    \end{align*}
    \label{thm:weighting}
\end{supptheorem}
\begin{proof}
    We will start by proving the theorem for $\tilde{q} = q$. 
    We can decompose $E\left[ \frac{q(B_1^{(2)}| \bm{X}, A_{\bm{J}}^{-}(\cdot))}{f(B_1^{(2)}| \bm{X}, A_{\bm{J}}^{-}(\cdot))}Y\right] = \Circled{1} + \Circled{2} + \Circled{3}$ where
    \begin{align*}
        \Circled{1} &= E\left[\frac{q(B_1^{(2)}| \bm{X}, A_{\bm{J}}^{-}(\cdot))}{f(B_1^{(2)}| \bm{X}, A_{\bm{J}}^{-}(\cdot))}\{Y - \epsilon - m(\bm{X}, A_{\bm{J}}(\cdot))\}\right] \\
        \Circled{2} &= E\left[\frac{q(B_1^{(2)}| \bm{X}, A_{\bm{J}}^{-}(\cdot))}{f(B_1^{(2)}| \bm{X}, A_{\bm{J}}^{-}(\cdot))}\epsilon\right] \\
         \Circled{3} &= E\left[\frac{q(B_1^{(2)}| \bm{X}, A_{\bm{J}}^{-}(\cdot))}{f(B_1^{(2)}| \bm{X}, A_{\bm{J}}^{-}(\cdot))}m(\bm{X}, A_{\bm{J}}(\cdot))\right]
    \end{align*}
    It remains to show that $\lim_{\bm{J} \to \infty} \Circled{1} = \lim_{\bm{J} \to \infty} \Circled{2} = 0$ and $\lim_{\bm{J} \to \infty} \Circled{3} = E[Y\{A^Q(\cdot)\}]$.
    First we will show that $\lim_{\bm{J} \to \infty} \Circled{1} = 0$. By the definition of covariance and the Cauchy-Schwartz inequality we have that $E[XY] \leq E[X]E[Y] + \sqrt{\Var(X)\Var(Y)}$ such that 
    \begin{align*}
        \Circled{1} &\leq \underbrace{E\left[\frac{q(B_1^{(2)}| \bm{X}, A_{\bm{J}}^{-}(\cdot))}{f(B_1^{(2)}| \bm{X}, A_{\bm{J}}^{-}(\cdot))}\right]E[Y - \epsilon - m(\bm{X}, A_{\bm{J}}(\cdot))]}_{\Circled{1A}} \\
        &\quad + \underbrace{\sqrt{\Var\left[\frac{q(B_1^{(2)}| \bm{X}, A_{\bm{J}}^{-}(\cdot))}{f(B_1^{(2)}| \bm{X}, A_{\bm{J}}^{-}(\cdot))}\right]\Var[Y - \epsilon - m(\bm{X}, A_{\bm{J}}(\cdot))]}}_{\Circled{1B}}
    \end{align*}
    Then,
    \begin{align*}
        \lim_{\bm{J} \to \infty}\Circled{1A} &\leq C_{Q1}\lim_{\bm{J}\to\infty} E[Y - \epsilon - m(\bm{X}, A_{\bm{J}}(\cdot))] \\
        &= C_{Q_1}\lim_{\bm{J}\to\infty} E[Y - m(\bm{X}, A_{\bm{J}}(\cdot))]  \\
        &= 0
    \end{align*}
    where the first equality follows from condition \textbf{C5}, the second from the fact that $E[\epsilon] = 0$, and the third from Lemma 3.
    Next,
    \begin{align*}
        \lim_{\bm{J} \to \infty}\Circled{1B} &= \lim_{\bm{J} \to \infty}\sqrt{C_{Q2}\Var[Y - \epsilon - m(\bm{X}, A_{\bm{J}}(\cdot))]} \\
        &= \lim_{\bm{J} \to \infty}\sqrt{C_{Q2}E[(Y - \epsilon - m(\bm{X}, A_{\bm{J}}(\cdot)))^2]} \\
        &= \lim_{\bm{J} \to \infty}\sqrt{C_{Q2}E[(m(\bm{X}, A(\cdot)) - m(\bm{X}, A_{\bm{J}}(\cdot)))^2]} \\
        &\leq \lim_{\bm{J} \to \infty}\sqrt{C_{Q2}E\left[C^2\sum_{k=1}^3\sum_{j=J_k + 1}^\infty\theta_j^{(k)^2}A_j^{(k)^2}\right]} \\
        &= \lim_{\bm{J} \to \infty}\sqrt{C_{Q2}C^2\sum_{k=1}^3\sum_{j=J_k + 1}^\infty \theta_j^{(k)^2}}\\
        &= 0
    \end{align*}
    where first equality follows from condition \textbf{C4}, the second equality follows from the fact that $\lim_{\bm{J} \to \infty}E[Y - \epsilon - \mu(\bm{X}, A_{\bm{J}}(\cdot))] = 0$, the third follows from conditioning on $\bm{X}, A(\cdot)$, and the fourth and fifth inequality follows from condition \textbf{C2} and the proof of Lemma 2. Then, we have that $\lim_{\bm{J} \to \infty} \Circled{1} = 0$ as desired.

    Next, we must show that $\lim_{\bm{J} \to \infty} \Circled{2} = 0$. Note that for every $\bm{J}$ we have that,
    \begin{align*}
        \Circled{2} &= E\left[\frac{q(B_1^{(2)}| \bm{X}, A_{\bm{J}}^{-}(\cdot))}{f(B_1^{(2)}| \bm{X}, A_{\bm{J}}^{-}(\cdot))}\epsilon\right] \\
        &= E\left[E\left[\frac{q(B_1^{(2)}| \bm{X}, A_{\bm{J}}^{-}(\cdot))}{f(B_1^{(2)}| \bm{X}, A_{\bm{J}}^{-}(\cdot))}\epsilon |\bm{X}, A_{\bm{J}}^{-}(\cdot) \right]\right] \\
        &= E\left[\frac{q(B_1^{(2)}| \bm{X}, A_{\bm{J}}^{-}(\cdot))}{f(B_1^{(2)}| \bm{X}, A_{\bm{J}}^{-}(\cdot))}E[\epsilon |\bm{X}, A_{\bm{J}}^{-}(\cdot)]\right] \\
        &= 0
    \end{align*}
    where the last equality follows from the fact that 
    $E[\epsilon |\bm{X}, A_{\bm{J}}^{-}(\cdot)] = E[E[\epsilon|\bm{X},A(\cdot)]|\bm{X}, A_{\bm{J}}^{-}(\cdot)] = E[0] = 0$. Then since $\Circled{2} = 0$ for all $\bm{J}$, $\lim_{\bm{J} \to \infty} \Circled{2} = 0$ as desired.
    
    Lastly, we have that,
    \begin{align*}
    \Circled{3} &=E\left[ \int_{\mathcal{B}_1^{(2)}} \frac{q(B_1^{(2)}| \bm{X}, A_{\bm{J}}^{-}(\cdot))}{f(B_1^{(2)}| \bm{X}, A_{\bm{J}}^{-}(\cdot))} m(\bm{X}, b_1^{(2)}\gamma_1^{(2)}(\cdot) + A_{\bm{J}}^{-}(\cdot))f(b_1^{(2)}|\bm{X}, A_{\bm{J}}^{-}(\cdot))db_1^{(2)}\right] \\
    &= E\left[ \int_{\mathcal{B}_1^{(2)}} m(\bm{X}, b_1^{(2)}\gamma_1^{(2)}(\cdot) + A_{\bm{J}}^{-}(\cdot))q(b_1^{(2)}|\bm{X}, A_{\bm{J}}^{-}(\cdot))db_1^{(2)}\right]
    \end{align*}
    such that $\lim_{\bm{J} \to \infty} \Circled{3}$ is equal to the identification in Theorem \ref{thm:outcome_reg} such that $\lim_{\bm{J} \to \infty} \Circled{3} = E[Y\{A^Q(\cdot)\}]$ as desired.

    We will now prove the theorem for $\tilde{q} = q^*$.
    We can similarly decompose $E\left[\frac{q^*(B_1^{(2)}|\bm{X}, A_{\bm{J}}^{-[0, t_2]}(\cdot))}{f(B_1^{(2)}| \bm{X}, A_{\bm{J}}^{-[0, t_2]}(\cdot))}Y\right] = \Circled{1*} + \Circled{2*} + \Circled{3*}$, where
    \begin{align*}
       \Circled{1*} &= E\left[\frac{q^*(B_1^{(2)}| \bm{X},  A_{\bm{J}}^{-[0, t_2]}(\cdot))}{f(B_1^{(2)}| \bm{X},  A_{\bm{J}}^{-[0, t_2]}(\cdot))}\{Y - \epsilon - m(\bm{X}, A_{\bm{J}}(\cdot))\}\right] \\
        \Circled{2*} &= E\left[\frac{q^*(B_1^{(2)}| \bm{X},  A_{\bm{J}}^{-[0, t_2]}(\cdot))}{f(B_1^{(2)}| \bm{X},  A_{\bm{J}}^{-[0, t_2]}(\cdot))}\epsilon\right] \\
        \Circled{3*} &= E\left[\frac{q^*(B_1^{(2)}| \bm{X},  A_{\bm{J}}^{-[0, t_2]}(\cdot))}{f(B_1^{(2)}| \bm{X},  A_{\bm{J}}^{-[0, t_2]}(\cdot))}m(\bm{X}, A_{\bm{J}}(\cdot))\right] 
    \end{align*}
    By similar arguments to the proof of this theorem for $\tilde{q} = q^*$ we have that $\lim_{\bm{J} \to \infty} \Circled{1*} = \lim_{\bm{J} \to \infty} \Circled{2*} = 0$ and it remains to show that $\lim_{\bm{J} \to \infty} \Circled{3} = E[Y\{A^{Q^*}(\cdot)\}]$. Now,
    \begin{align*}
     \Circled{3*} &=E\left[ \int_{\mathcal{B}_1^{(2)}} \frac{q^*(B_1^{(2)}| \bm{X}, A_{\bm{J}}^{-[0,t_2]}(\cdot))}{f(B_1^{(2)}| \bm{X}, A_{\bm{J}}^{-[0,t_2]}(\cdot))} m(\bm{X}, b_1^{(2)}\gamma_1^{(2)}(\cdot) + A_{\bm{J}}^{-}(\cdot))f(b_1^{(2)}|\bm{X}, A_{\bm{J}}^{-[0,t_2]}(\cdot))db_1^{(2)}\right] \\
    &= E\left[ \int_{\mathcal{B}_1^{(2)}} m(\bm{X}, b_1^{(2)}\gamma_1^{(2)}(\cdot) + A_{\bm{J}}^{-}(\cdot))q(b_1^{(2)}|\bm{X}, A_{\bm{J}}^{-[0,t_2]}(\cdot))db_1^{(2)}\right]
    \end{align*}
    such that $\lim_{\bm{J} \to \infty} \Circled{3*}$ is equal to the identification in Theorem \ref{thm:outcome_reg} and $\lim_{\bm{J} \to \infty} \Circled{3*} = E[Y\{A^{Q^*}(\cdot)\}]$ as desired.
\end{proof}

\subsection{Theorem \ref{thm:dr_identification}}
\begin{theorem}
\textnormal{(Augmented identification)} Given conditions \textbf{C1}-\textbf{C2}, \textbf{C$\tilde{\bm{3}}$}-\textbf{C$\tilde{\bm{4}}$} and assumptions \textbf{A1}, \textbf{A$\tilde{\bm{2}}$}, and assumption \textbf{A3} (for $\tilde{q} = q$)  or \textbf{A$\bm{3}^*$} (for $\tilde{q} = q^*$) 
\begin{align*}
    \mu^{\tilde{Q}}  &= \lim_{\bm{J} \to \infty} E\left(\frac{\tilde{q}(B_1^{(2)}|  \bm{X}, A_{\bm{J}}^{-\tilde{Q}}(\cdot))}{f(B_1^{(2)}|  \bm{X}, A_{\bm{J}}^{-\tilde{Q}}(\cdot))}\left[Y - E_{\tilde{Q}}\{m(\bm{X}, A_{\bm{J}}(\cdot))|\bm{X},  A_{\bm{J}}^{-}(\cdot)\}\right] \right.\\
    &\quad+ \left. E_{\tilde{Q}}\{m(\bm{X}, A_{\bm{J}}(\cdot))|\bm{X},  A_{\bm{J}}^{-}(\cdot)\}\right),\:\:\: \text{for}\:\: \tilde{q} = q, q^*.
\end{align*}
\end{theorem}

\begin{proof}
It follows directly from Theorems \ref{thm:outcome_reg} and \ref{thm:weighting} that
\begin{align*}
    \mu^{\tilde{Q}}  &= \lim_{\bm{J} \to \infty} E\left[\frac{\tilde{q}(B_1^{(2)}|  \bm{X}, A_{\bm{J}}^{-\tilde{Q}}(\cdot))}{f(B_1^{(2)}|  \bm{X}, A_{\bm{J}}^{-\tilde{Q}}(\cdot))}\left\{Y - m(\bm{X}, A_{\bm{J}})\right\} + E_{\tilde{Q}}[m(\bm{X}, A_{\bm{J}}(\cdot))|\bm{X},  A_{\bm{J}}^{-}(\cdot)]\right]
\end{align*}
for $\tilde{q} = q, q^*$. Thus, it remains to show that $E\left[\frac{\tilde{q}(B_1^{(2)}|  \bm{X}, A_{\bm{J}}^{-\tilde{Q}}(\cdot))}{f(B_1^{(2)}|  \bm{X}, A_{\bm{J}}^{-\tilde{Q}}(\cdot))}\left\{m(\bm{X}, A_{\bm{J}}) - E_{\tilde{Q}}[m(\bm{X}, A_{\bm{J}}(\cdot))|\bm{X},  A_{\bm{J}}^{-}(\cdot)]\right\}\right] = 0$ Now,
\begin{align*}
  E&\left[\frac{\tilde{q}(B_1^{(2)}|  \bm{X}, A_{\bm{J}}^{-\tilde{Q}}(\cdot))}{f(B_1^{(2)}|  \bm{X}, A_{\bm{J}}^{-\tilde{Q}}(\cdot))}\left\{m(\bm{X}, A_{\bm{J}}) - E_{\tilde{Q}}[m(\bm{X}, A_{\bm{J}}(\cdot))|\bm{X},  A_{\bm{J}}^{-}(\cdot)]\right\}\right] \\
  &=  E\left[E\left[\frac{\tilde{q}(B_1^{(2)}|  \bm{X}, A_{\bm{J}}^{-\tilde{Q}}(\cdot))}{f(B_1^{(2)}|  \bm{X}, A_{\bm{J}}^{-\tilde{Q}}(\cdot))}\left\{m(\bm{X}, A_{\bm{J}}) - E_{\tilde{Q}}[m(\bm{X}, A_{\bm{J}}(\cdot))|\bm{X},  A_{\bm{J}}^{-}(\cdot)]\right\}\middle\vert\bm{X}, A^{-}_{\bm{J}}(\cdot) \right]\right] \\
  &= E\left[\int_{\mathcal{B}^{(1)_2}}m(\bm{X}, A_{\bm{J}}(\cdot)\tilde{q}(b_1^{(2)}|  \bm{X}, A_{\bm{J}}^{-\tilde{Q}}(\cdot))db_1^{(2)} - E_{\tilde{Q}}[m(\bm{X}, A_{\bm{J}}(\cdot))|\bm{X},  A_{\bm{J}}^{-}(\cdot)]\right] \\
  &= 0
\end{align*}
as desired.
\end{proof}

\subsection{Theorem \ref{thm:estimator_asymp}}
\begin{theorem}
   Assume the conditions and assumptions of Theorem \ref{thm:dr_identification}, assumption \textbf{A7}, and conditions  \textbf{C7}-\textbf{C8} hold. Given that assumptions \textbf{A$\tilde{\bm{4}}$}-\textbf{A$\tilde{\bm{6}}$} and conditions \textbf{C$\tilde{\bm{5}}$}-\textbf{C$\bm{6}$} hold for all $\bm{J}$, $
   \frac{\sqrt{n}}{\sigma^{\tilde{Q}(\delta)}_{\bm{J}(n)}}( \hat{\mu}^{\tilde{Q}(\delta)}_{\bm{J}(n)} - \mu^{\tilde{Q}(\delta)}) \to N(0,1)$,
    where $\sigma^{\tilde{Q}(\delta)^2}_{\bm{J}(n)} = E\left[w^{\tilde{Q}(\delta)}(\bm{X}, A_{\bm{J}(n)}^{-\tilde{Q}(\delta)}(\cdot))^2\left\{v(\bm{X}, A_{\bm{J}(n)}(\cdot))\right.\right. $ $\left.\left.+  d^{\tilde{Q}(\delta)}(\bm{X}, A_{\bm{J}(n)}(\cdot))^2 \right\}\right] $ $+ \Var\left[E_{\tilde{Q}(\delta)}\{m(\bm{X}, A_{\bm{J}(n)}(\cdot))|\bm{X}, A_{\bm{J}(n)}^{-}(\cdot)\}\right]$ for $\tilde{q} = q, q^*$.
\end{theorem}

\begin{proof}

To simplify notation let, 
\begin{align*}
    \hat{\mu}_{\bm{J}(n), i} &= \hat{\mu}^{\tilde{Q}(\delta)}_{\bm{J}(n), i} \\
    &= \frac{\hat{\tilde{q}}_\delta(B_{1,i}^{(2)}|\bm{X}, A_{\bm{J}(n),i}^{-\tilde{Q}(\delta)}(\cdot))}{\hat{f}(B_{1,i}^{(2)}|\bm{X}, A_{\bm{J}(n),i}^{-\tilde{Q}(\delta)}(\cdot))}\left\{Y_i - \int_{\mathcal{B}_{1}^{(2)}}\hat{m}\left(\bm{X}_i, b_1^{(2)}\gamma_j^{(2)}(\cdot) + A_{\bm{J}(n),i}^{-}(\cdot)\right)\hat{\tilde{q}}_\delta\left(b_1^{(2)}|\bm{X},A_{\bm{J}(n),i}^{-\tilde{Q}(\delta)}(\cdot)\right)db_1^{(2)}\right\} \\
    &\quad+ \int_{\mathcal{B}_{1}^{(2)}}\hat{m}\left(\bm{X}_i, b_1^{(2)}\gamma_j^{(2)}(\cdot) + A_{\bm{J}(n),i}^{-}(\cdot)\right)\hat{\tilde{q}}_\delta\left(b_1^{(2)}|\bm{X},A_{\bm{J}(n),i}^{-\tilde{Q}(\delta)}(\cdot)\right)db_1^{(2)} 
\end{align*}
and let $\mu_{\bm{J}(n), i}$ be the corresponding quantity with the true values of $\tilde{q}_\delta, f,$ and $m$. Then, we have the following decomposition:
\begin{align}
    \hat{\mu}^{\tilde{Q}(\delta)}_{\bm{J}(n)} - \mu^{\tilde{Q}(\delta)} &=\frac1n\sum_{i=1}^n  \hat{\mu}_{\bm{J}(n), i} - \mu^{\tilde{Q}(\delta)} \\
    &= \frac1n\sum_{i=1}^n \mu_{\bm{J}(n), i} - E[\mu_{\bm{J}(n), i}]  \label{eq:est_term1} \\
    &\quad + \frac1n\sum_{i=1}^n  E[\mu_{\bm{J}(n), i}]- \mu^{\tilde{Q}(\delta)}  \label{eq:est_term2} \\
    &\quad + \frac1n\sum_{i=1}^n  \hat{\mu}^{\tilde{Q}(\delta)}_{\bm{J}(n)} -   \mu_{\bm{J}(n), i} 
    \label{eq:est_term3}
\end{align}

Term \eqref{eq:est_term3} is equal to the empirical process and bias remainder terms as defined in \cite{schindl_incremental_2026}. Thus, given the assumptions of Theorem \ref{thm:estimator_asymp}, by Theorem 4 of \cite{schindl_incremental_2026} term \eqref{eq:est_term3} is $o_p(\sqrt{\delta/n})$ as desired. Thus it remains to show that term \eqref{eq:est_term1} converges to a normal distribution and that term \eqref{eq:est_term2} is $o(\sqrt{\delta/n})$.

Consider the triangular array $\{\eta_{\bm{J}(n), i} = \mu_{\bm{J}(n), i} - E[\mu_{\bm{J}(n), i}], n \in \mathbb{Z}^+, i = 1, \ldots n\}$ where for every $n$, $\{\eta_{\bm{J}(n), i}\}_{i=1}^n$ are mutually independent. It is clear that $E[\eta_{\bm{J}(n), i}] = 0$ and that term \eqref{eq:est_term1} is equivalent to $\sum_{i=1}^n \eta_{\bm{J}(n), i}$. Next, let  $\sigma_{\bm{J}(n)}^{\tilde{Q}(\delta)^2} = \frac1n\sum_{i=1}^n\Var(\eta_{\bm{J}(n), i})$. Then given that the Lindeberg condition holds, by the Lindeberg-Feller Central Limit Theorem we have that,
\begin{align*}
    \frac{\sqrt{n}}{\sigma_{\bm{J}(n)}^{\tilde{Q}(\delta)}}\sum_{i=1}^n \eta_{\bm{J}(n), i} \to N(0,1)
\end{align*}
It remains to show that the Lindeberg condition,
$\lim_{n\to\infty} \frac{1}{\sigma_{\bm{J}(n)}^{\tilde{Q}(\delta)}}\sum_{i=1}^n E[(\eta_{\bm{J}(n), i})^2I\{|\eta_{\bm{J}(n), i}| > \epsilon\sigma_{\bm{J}(n)}^{\tilde{Q}(\delta)}\}] = 0$, holds for any $\epsilon > 0$. To show that this condition holds, we will use similar techniques to the proof of Theorem 4 in \cite{schindl_incremental_2026}. By Lemma 1 of \cite{schindl_incremental_2026}, we have that,
\begin{align*}
    \sigma_{\bm{J}(n)}^{\tilde{Q}(\delta)} \geq \sqrt{n\delta}\frac{\pi^{1/2}_{\text{min}}\sigma_{\min}}{\pi_{\max}\sqrt{2K}},
\end{align*}
such that,
\begin{align*}
   E\left[(\eta_{\bm{J}(n), i} )^2I\{|\eta_{\bm{J}(n), i} | > \epsilon\sigma_{\bm{J}(n)}^{\tilde{Q}(\delta)} \}\right] \leq E\left[(\eta_{\bm{J}(n), i} )^2I\left\{|\eta_{\bm{J}(n), i} | > \epsilon\sqrt{n\delta}\frac{\pi^{1/2}_{\text{min}}\sigma_{\min}}{\pi_{\max}\sqrt{2K}}\right\}\right].
\end{align*}
Given the assumption that $|Y| \leq B$ with probability 1,
\begin{align*}
   |\eta_{\bm{J}(n), i}| &\leq \underbrace{\left\vert \frac{\tilde{q}(B_{1,i}^{(2)}|\bm{X}_i, A_{\bm{J}(n), i}^{-\tilde{Q}(\delta)}(\cdot))}{f(B_{1,i}^{(2)}| \bm{X}_i,A_{\bm{J}(n), i}^{-\tilde{Q}(\delta)}(\cdot))}\left\{Y_i - \int_{\mathcal{B}_{1}^{(2)}}m(\bm{X}_i, b_1^{(2)}\gamma_j^{(2)}(\cdot) + A_{\bm{J}(n), i}^{-}(\cdot))\tilde{q}(b_1^{(2)}|\bm{X}_i,A_{\bm{J}(n), i}^{-\tilde{Q}(\delta)}(\cdot))db_1^{(2)}\right\} \right\vert}_{\Circled{1}}\\ 
   &\quad + \underbrace{\left\vert\int_{\mathcal{B}_{1}^{(2)}}m(\bm{X}_i, b_1^{(2)}\gamma_j^{(2)}(\cdot) + A_{\bm{J}(n), i}^{-}(\cdot))\tilde{q}(b_1^{(2)}|\bm{X}_i,A_{\bm{J}(n), i}^{-\tilde{Q}(\delta)}(\cdot))db_1^{(2)}\right\vert}_{\Circled{2}}\\
   &\quad + |E[\mu_{\bm{J}(n),i}]|.\:\: \text{Then,}\\
   \Circled{2} &\leq B\int_{\mathcal{B}_{1}^{(2)}}\tilde{q}(b_1^{(2)}|\bm{X}_i,A_{\bm{J}(n), i}^{-\tilde{Q}(\delta)}(\cdot))db_1^{(2)} = B. \:\: \text{We also have that,} \\
    \Circled{1} &\leq 2B\left\vert\frac{\tilde{q}(B_{1,i}^{(2)}| \bm{X}_i,A_{\bm{J}(n), i}^{-\tilde{Q}(\delta)}(\cdot))}{f(B_{1,i}^{(2)}|\bm{X}_i, A_{\bm{J}(n), i}^{-\tilde{Q}(\delta)}(\cdot))}\right\vert \leq \frac{2B}{\pi_{\min}}\left\vert\frac{\exp(\delta s)}{\int_{s}\exp(\delta s)}\right\vert \leq \frac{2B\delta\exp(\delta)}{\pi_{\min}(\exp(\delta)-1)}. \:\: \text{Therefore,} \\
    |\eta_{\bm{J}(n), i}| &\leq \frac{4B\delta\exp(\delta)}{\pi_{\min}(\exp(\delta)-1)} + 2B.
\end{align*}
Then, we have that
\begin{align*}
   E\left[(\eta_{\bm{J}(n), i} )^2I\{|\eta_{\bm{J}(n), i} | > \epsilon\sigma_{\bm{J}(n)}^{\tilde{Q}(\delta)} \}\right] \leq E\left[(\eta_{\bm{J}(n), i} )^2I\left\{\frac{4B\exp(\delta)}{\pi_{\min}(\exp(\delta)-1)} + 2B/\delta > \epsilon\sqrt{n/\delta}\frac{\pi^{1/2}_{\text{min}}\sigma_{\min}}{\pi_{\max}\sqrt{2K}}\right\}\right] 
\end{align*}
where
\begin{align*}
   I\left\{\frac{4B\exp(\delta)}{\pi_{\min}(\exp(\delta)-1)} + 2B/\delta > \epsilon\sqrt{n/\delta}\frac{\pi^{1/2}_{\text{min}}\sigma_{\min}}{\pi_{\max}\sqrt{2K}}\right\} &\to 0 
\end{align*}
because $\sqrt{n/\delta} \to \infty$. We finish the proof by using the dominated convergence theorem. For all $n$, we have that 
\begin{align*}
    \frac{1}{\sigma_{\bm{J}(n)}^{\tilde{Q}(\delta)}} \sum_{i=1}^n (\eta_{\bm{J}(n), i} )^2I\{|\eta_{\bm{J}(n), i} | > \epsilon\sigma_{\bm{J}(n)}^{Q(\delta)} \} &\leq \frac{1}{\sigma_{\bm{J}(n)}^{\tilde{Q}(\delta)}} \sum_{i=1}^n (\eta_{\bm{J}(n), i} )^2, \text{   where} \\
    E\left[ \frac{1}{\sigma_{\bm{J}(n)}^{\tilde{Q}(\delta)}} \sum_{i=1}^n (\eta_{\bm{J}(n), i} )^2\right] &= \frac{1}{\sum_{i=1}^n E[(\eta_{\bm{J}(n), i} )^2]} \sum_{i=1}^n E[(\eta_{\bm{J}(n), i} )^2] < \infty
\end{align*}
such that by the dominated convergence theorem
\begin{align*}
   \lim_{n\to\infty} \frac{1}{\sigma_{\bm{J}(n)}^{\tilde{Q}(\delta)}}\sum_{i=1}^n E[(\eta_{\bm{J}(n), i})^2I\{|\eta_{\bm{J}(n), i}]| > \epsilon\sigma_{\bm{J}(n)}^{\tilde{Q}(\delta)} \}] = 0
\end{align*}
and thus the Lindeberg condition holds. \\

Next, we must show that term \eqref{eq:est_term2} is $o(\sqrt{n/\delta})$. Now, note that since $E[\mu_{\bm{J}(n), 1}] = \cdots = [\mu_{\bm{J}(n), n}]$ for all $n$,
\begin{align*}
    \frac1n\sum_{i=1}^n  E[\mu_{\bm{J}(n), i}] - \mu^{\tilde{Q}(\delta)} &= E[\mu_{\bm{J}(n), i}] - \mu^{\tilde{Q}(\delta)}.
\end{align*}
By Theorems \ref{thm:outcome_reg} and \ref{thm:dr_identification}, we have the following simplification of $E[\mu_{\bm{J}(n), i}]$:
\begin{align*}
    E[\mu_{\bm{J}(n), i}] &= E\left[\frac{\tilde{q}(B_{1,i}^{(2)}|\bm{X}_i, A_{\bm{J}(n), i}^{-\tilde{Q}(\delta)}(\cdot))}{f(B_{1,i}^{(2)}| \bm{X}_i,A_{\bm{J}(n), i}^{-\tilde{Q}(\delta)}(\cdot))}\left\{Y_i - \int_{\mathcal{B}_{1}^{(2)}}m(\bm{X}_i, b_1^{(2)}\gamma_j^{(2)}(\cdot) + A_{\bm{J}(n), i}^{-}(\cdot))\tilde{q}(b_1^{(2)}|\bm{X}_i,A_{\bm{J}(n), i}^{-\tilde{Q}(\delta)}(\cdot))db_1^{(2)}\right\}\right]\\
    &\quad + E\left[\int_{\mathcal{B}_{1}^{(2)}}m(\bm{X}_i, b_1^{(2)}\gamma_j^{(2)}(\cdot) + A_{\bm{J}(n), i}^{-}(\cdot))\tilde{q}(b_1^{(2)}|\bm{X}_i,A_{\bm{J}(n), i}^{-\tilde{Q}(\delta)}(\cdot))db_1^{(2)}\right] \\
    &= E\left[\int_{\mathcal{B}_{1}^{(2)}}m(\bm{X}_i, b_1^{(2)}\gamma_j^{(2)}(\cdot) + A_{i}^{-}(\cdot))\tilde{q}(b_1^{(2)}|\bm{X}_i,A_{\bm{J}(n), i}^{-\tilde{Q}(\delta)}(\cdot))db_1^{(2)}\right]
\end{align*}
Then we have that,
\begin{align*}
  E[\mu_{\bm{J}(n), i}] - \mu^{\tilde{Q}(\delta)} &=  E\left[\int_{\mathcal{B}_{1}^{(2)}}m(\bm{X}_i, b_1^{(2)}\gamma_j^{(2)}(\cdot) + A_{i}^{-}(\cdot))\tilde{q}(b_1^{(2)}|\bm{X}_i,A_{\bm{J}(n), i}^{-\tilde{Q}(\delta)}(\cdot))db_1^{(2)}\right] \\
   &\quad -  E\left[\int_{\mathcal{B}_{1}^{(2)}}m(\bm{X}_i, b_1^{(2)}\gamma_j^{(2)}(\cdot) + A_{i}^{-}(\cdot))\tilde{q}(b_1^{(2)}|\bm{X}_i,A_{i}^{-\tilde{Q}(\delta)}(\cdot))db_1^{(2)}\right] \\
   &= E[E_{\tilde{Q}}[m(\bm{X}_i, B_{1, i}^{(2)}\gamma_j^{(2)}(\cdot) + A_{i}^{-}(\cdot))|\bm{X}_i, A_{\bm{J}(n),i}^{-}(\cdot)]] \\
   &- E[E_{\tilde{Q}}[m(\bm{X}_i, B_{1, i}^{(2)}\gamma_j^{(2)}(\cdot) + A_{i}^{-}(\cdot))|\bm{X}_i, A_{i}^{-}(\cdot)]]  \\
   &\leq E[\vert E_{\tilde{Q}}[m(\bm{X}_i, B_{1, i}^{(2)}\gamma_j^{(2)}(\cdot) + A_{i}^{-}(\cdot))|\bm{X}_i, A_{\bm{J}(n)i}^{-}(\cdot)]] \\
   &- E_{\tilde{Q}}[m(\bm{X}_i, B_{1, i}^{(2)}\gamma_j^{(2)}(\cdot) + A_{i}^{-}(\cdot))|\bm{X}_i, A_{i}^{-}(\cdot)]\vert]  \\
   &\leq C_{\tilde{Q}}E\left[\left\Vert\sum_{k=1}^3\sum_{j=J_k(n)}^\infty \theta^{(k)^{1/2}}A^{(k)}_j\psi^{(k)}_j(\cdot)\right\Vert_2\right] \\
   &\leq C_{\tilde{Q}}\sqrt{\sum_{k=1}^3\sum_{j=J_k(n)}^\infty \theta^{(k)}_j} \\
   &= o(\sqrt{\delta/n})
\end{align*}
by the assumption that $\Delta_{\bm{J}(n)} = o(\delta/n)$. 
Taking this together we have that,
\begin{align*}
    \frac{\sqrt{n}}{\sigma_{\bm{J}(n)}^{\tilde{Q}(\delta)} }\left(\hat{\mu}^{\tilde{Q}(\delta)}_{\bm{J}(n)}- \mu^{Q(\delta)}\right) &= \frac{\sqrt{n}}{\sigma_{\bm{J}(n)}^{\tilde{Q}(\delta)} }\left(\frac1n\sum_{i=1}^n \mu_{\bm{J}(n),i} - E[\mu_{\bm{J}(n),i}]\right) \\ 
    &\quad + \frac{\sqrt{n}}{\sigma_{\bm{J}(n)}^{\tilde{Q}(\delta)} }\left(\frac1n\sum_{i=1}^n  E[\mu_{\bm{J}(n), i}] - \mu^{\tilde{Q}(\delta)}  + \frac1n\sum_{i=1}^n  \hat{\mu}_{\bm{J}(n),i} -  \mu_{\bm{J}(n),i}\right) \\
    &= \frac{\sqrt{n}}{\sigma_{\bm{J}(n)}^{\tilde{Q}(\delta)} }\left(\frac1n\sum_{i=1}^n \mu_{\bm{J}(n),i} - E[\mu_{\bm{J}(n),i}]\right) + o_p(\sqrt{\delta/n}) \\
    &\to N(0,1)
\end{align*}
where 
\begin{align*}
        \sigma^{\tilde{Q}(\delta)^2}_{\bm{J}(n)} &= E\left[\left(\frac{\tilde{q}(B_{1}^{(2)}|\bm{X}, A_{\bm{J}(n)}^{-\tilde{Q}(\delta)}(\cdot))}{f(B_{1}^{(2)}| \bm{X},A_{\bm{J}(n)}^{-\tilde{Q}(\delta)}(\cdot))}\right)^2\left\{\Var[Y|B_{1}^{(2)}, A_{\bm{J}(n)}(\cdot)] \right.\right.\\
        &\quad + \left.\left.\left[m(\bm{X}, A_{\bm{J}(n)}(\cdot)) - E_{\tilde{Q}}[m(\bm{X}, A_{\bm{J}(n)}(\cdot))|\bm{X}_i, A_{\bm{J}(n)}^{-}(\cdot)]\right]^2 \right\}\right] \\
        &\quad + \Var\left(E_{\tilde{Q}}[m(\bm{X}, A_{\bm{J}(n)}(\cdot))|\bm{X}, A_{\bm{J}(n)}^{-}(\cdot)]\right)
    \end{align*}
by Theorem 1 of \cite{schindl_incremental_2026}.
\end{proof}

\subsection{Corollary \ref{thm:coro_estimator_asymp}}
\begin{suppcoro}
    Assume the conditions and assumptions of Theorem \ref{thm:dr_identification} and conditions \textbf{C7}-\textbf{C8} hold. Given that assumptions \textbf{A$\tilde{\bm{4}}$}-\textbf{A$\tilde{\bm{6}}$} and conditions \textbf{C$\tilde{\bm{5}}$}-\textbf{C$\bm{6}$} hold at the limit as $\bm{J} \to \infty$,
    \begin{align*}
        \frac{\sqrt{n}}{\sigma^{\tilde{Q}(\delta)}}( \hat{\mu}^{\tilde{Q}(\delta)} - \mu^{\tilde{Q}(\delta)}) \to N(0,1),\:\:\: \text{for}\:\: \tilde{q} =q, q^*.
    \end{align*}
    \label{thm:coro_estimator_asymp}
\end{suppcoro}

\subsection{Theorem \ref{thm:contrast_asymp}}
\begin{theorem}
   Assume the conditions and assumptions of Theorem \ref{thm:estimator_asymp} hold. Then, $\frac{\sqrt{n}}{\sigma^{\tau^{\tilde{Q}(\delta)}}_{\bm{J}(n)}}( \hat{\tau}^{\tilde{Q}(\delta)}_{\bm{J}(n)} - \tau^{\tilde{Q}(\delta)}) \to N(0,1)$,
    where $\sigma^{\tau^{\tilde{Q}(\delta)^2}}_{\bm{J}(n)} = E\left[\{w^{\tilde{Q}(\delta)}(\bm{X}, A_{\bm{J}(n)}^{-\tilde{Q}(\delta)}(\cdot)) - 1\}^2v(\bm{X}, A_{\bm{J}(n)}(\cdot))\right.$ \\
    $\left. +  w^{\tilde{Q}(\delta)}(\bm{X}, A_{\bm{J}(n)}^{-\tilde{Q}(\delta)}(\cdot))^2d^{\tilde{Q}(\delta)}(\bm{X}, A_{\bm{J}(n)}(\cdot))^2 \right.\left.-2\Var_{\tilde{Q}(\delta)}\{m(\bm{X}, A_{\bm{J}(n)}(\cdot))| \bm{X}, A_{\bm{J}(n)}^{-}(\cdot)\}\right] $ \\ $+ \Var\{d^{\tilde{Q}(\delta)}(\bm{X}, A_{\bm{J}(n)}(\cdot))\}$ for $\tilde{q} = q, q^*$.
\end{theorem}

\begin{proof}
    Using the same notation as the proof of Theorem \ref{thm:estimator_asymp},
    \begin{align*}
        \hat{\tau}^{\tilde{Q}(\delta)}_{\bm{J}(n)} - \tau^{\tilde{Q}(\delta)} &= \hat{\mu}^{\tilde{Q}(\delta)}_{\bm{J}(n)} - \frac1n\sum_{i=1}^n Y_i - (\mu^{\tilde{Q}(\delta)} - \mu) \\
        &= \frac1n\sum_{i=1}^n \mu_{\bm{J}(n), i} - Y_i - (E[\mu_{\bm{J}(n), i}] - \mu)  \\
        &\quad + \frac1n\sum_{i=1}^n  E[\mu_{\bm{J}(n), i}]- \mu^{\tilde{Q}(\delta)}  \\
        &\quad + \frac1n\sum_{i=1}^n  \hat{\mu}^{\tilde{Q}(\delta)}_{\bm{J}(n)} -   \mu_{\bm{J}(n), i} 
    \end{align*}
    where from the proof of Theorem \ref{thm:estimator_asymp} we know these last two terms of $o_p(\delta/n)$. Thus it remains to determine the distribution $\frac1n\sum_{i=1}^n \mu_{\bm{J}(n), i} - Y_i - (E[\mu_{\bm{J}(n), i}] - \mu)$.

Consider the triangular array $\{\upsilon_{\bm{J}(n), i} = \mu_{\bm{J}(n), i} - Y_i - (E[\mu_{\bm{J}(n), i}] - \mu), n \in \mathbb{Z}^+, i = 1, \ldots n\}$ where for every $n$, $\{\upsilon_{\bm{J}(n), i}\}_{i=1}^n$ are mutually independent. It is clear that $E[\upsilon_{\bm{J}(n), i}] = 0$ and that term $\frac1n\sum_{i=1}^n \mu_{\bm{J}(n), i} - Y_i - (E[\mu_{\bm{J}(n), i}] - \mu)$ is equivalent to $\sum_{i=1}^n \upsilon_{\bm{J}(n), i}$. Next, let  $\sigma_{\bm{J}(n)}^{\tau^{\tilde{Q}(\delta)^2}} = \frac1n\sum_{i=1}^n\Var(\upsilon_{\bm{J}(n), i})$. Then given that the Lindeberg condition holds, by the Lindeberg-Feller Central Limit Theorem we have that,
\begin{align*}
    \frac{\sqrt{n}}{\sigma_{\bm{J}(n)}^{\tau^{\tilde{Q}(\delta)}}}\sum_{i=1}^n \upsilon_{\bm{J}(n), i} \to N(0,1)
\end{align*}
It remains to show that the Lindeberg condition,
$\lim_{n\to\infty} \frac{1}{\sigma_{\bm{J}(n)}^{\tau^{\tilde{Q}(\delta)}}}\sum_{i=1}^n E[(\upsilon_{\bm{J}(n), i})^2I\{|\upsilon_{\bm{J}(n), i}| > \epsilon\sigma_{\bm{J}(n)}^{\tau^{\tilde{Q}(\delta)}}\}] = 0$ holds for any $\epsilon > 0$. To show that this condition holds, we will use similar techniques to the proof of Theorem 4 in \cite{schindl_incremental_2026} and our proof of Theorem \ref{thm:estimator_asymp}. By the Cauchy-Schwartz inequality, Lemma 1 of \cite{schindl_incremental_2026}, and the assumption that $|Y| \leq B$ with probability 1, we have that,
\begin{align*}
    \sigma_{\bm{J}(n)}^{\tau^{\tilde{Q}(\delta)^2}} &= \Var(\eta_{\bm{J}(n), i} + (Y_i - \mu)) \\
    &\geq (\sigma_{\bm{J}(n)}^{{\tilde{Q}(\delta)}} - \sqrt{\Var(Y_i)})^2 \\
    \sigma_{\bm{J}(n)}^{\tau^{\tilde{Q}(\delta)}} &\geq \sigma_{\bm{J}(n)}^{{\tilde{Q}(\delta)}} - \sqrt{\Var(Y_i)} \\
    &\geq \sqrt{n\delta}\frac{\pi^{1/2}_{\text{min}}\sigma_{\min}}{\pi_{\max}\sqrt{2K}} - |B|/2
\end{align*}
since $\sigma_{\bm{J}(n)}^{{\tilde{Q}(\delta)}} \geq \sqrt{n\delta}\frac{\pi^{1/2}_{\text{min}}\sigma_{\min}}{\pi_{\max}\sqrt{2K}} $ and $\Var(Y_i) \leq B^2/4$.
Next, given the assumption that $|Y| \leq B$ with probability 1,
\begin{align*}
   |\upsilon_{\bm{J}(n), i}| &\leq \underbrace{\left\vert \frac{\tilde{q}(B_{1,i}^{(2)}| \bm{X}_i,A_{\bm{J}(n), i}^{-\tilde{Q}(\delta)}(\cdot))}{f(B_{1,i}^{(2)}| \bm{X}_i, A_{\bm{J}(n), i}^{-\tilde{Q}(\delta)}(\cdot))}\left\{Y_i - \int_{\mathcal{B}_{1}^{(2)}}m(\bm{X}_i, b_1^{(2)}\gamma_j^{(2)}(\cdot) + A_{\bm{J}(n), i}^{-}(\cdot))\tilde{q}(b_1^{(2)}|\bm{X}_i,A_{\bm{J}(n), i}^{-\tilde{Q}(\delta)}(\cdot))db_1^{(2)}\right\} \right\vert}_{\Circled{1}}\\ 
   &\quad + \underbrace{\left\vert\int_{\mathcal{B}_{1}^{(2)}}m(\bm{X}_i, b_1^{(2)}\gamma_j^{(2)}(\cdot) + A_{\bm{J}(n), i}^{-}(\cdot))\tilde{q}(b_1^{(2)}|\bm{X}_i,A_{\bm{J}(n), i}^{-\tilde{Q}(\delta)}(\cdot))db_1^{(2)}\right\vert}_{\Circled{2}}\\
   &\quad + |E[\mu_{\bm{J}(n),i}]| \\
   &\quad + |Y| +  |\mu|.\:\: \text{Then,} \\
   \Circled{2} &\leq B\int_{\mathcal{B}_{1}^{(2)}}\tilde{q}(b_1^{(2)}|\bm{X}_i,A_{\bm{J}(n), i}^{-\tilde{Q}(\delta)}(\cdot))db_1^{(2)} = B.\:\: \text{We also have that,} \\
    \Circled{1} &\leq 2B\left\vert\frac{\tilde{q}(B_{1,i}^{(2)}| \bm{X}_i,A_{\bm{J}(n), i}^{-\tilde{Q}(\delta)}(\cdot))}{f(B_{1,i}^{(2)}|\bm{X}_i, A_{\bm{J}(n), i}^{-\tilde{Q}(\delta)}(\cdot))}\right\vert \leq \frac{2B}{\pi_{\min}}\left\vert\frac{\exp(\delta s)}{\int_{s}\exp(\delta s)}\right\vert \leq \frac{2B\delta\exp(\delta)}{\pi_{\min}(\exp(\delta)-1)} .\:\: \text{Thus,}\\
    |\upsilon_{\bm{J}(n), i}| &\leq \frac{4B\delta\exp(\delta)}{\pi_{\min}(\exp(\delta)-1)} + 4B.
\end{align*}
Then, we have that
\begin{align*}
   E\left[(\upsilon_{\bm{J}(n), i} )^2I\{|\upsilon_{\bm{J}(n), i} | > \epsilon\sigma_{\bm{J}(n)}^{\tau^{\tilde{Q}(\delta)}} \}\right] &\leq E\left[(\upsilon_{\bm{J}(n), i} )^2I\left\{\frac{4B\exp(\delta)}{\pi_{\min}(\exp(\delta)-1)} + 4B/\delta \right.\right.\\
   &\quad > \left.\left. \epsilon\sqrt{n/\delta}\frac{\pi^{1/2}_{\text{min}}\sigma_{\min}}{\pi_{\max}\sqrt{2K}} - |B|/(2\delta)\right\}\right] 
\end{align*}
where
\begin{align*}
   I\left\{\frac{4B\exp(\delta)}{\pi_{\min}(\exp(\delta)-1)} + 4B/\delta > \epsilon\sqrt{n/\delta}\frac{\pi^{1/2}_{\text{min}}\sigma_{\min}}{\pi_{\max}\sqrt{2K}} - |B|/(2\delta)\right\} &\to 0 
\end{align*}
given that $\sqrt{n/\delta} \to \infty$. We finish the proof by using the dominated convergence theorem. For all $n$, we have that 
\begin{align*}
    \frac{1}{\sigma_{\bm{J}(n)}^{\tau^{\tilde{Q}(\delta)}}} \sum_{i=1}^n (\upsilon_{\bm{J}(n), i} )^2I\{|\upsilon_{\bm{J}(n), i} | > \epsilon\sigma_{\bm{J}(n)}^{\tau^{\tilde{Q}(\delta)}} \} &\leq \frac{1}{\sigma_{\bm{J}(n)}^{\tau^{\tilde{Q}(\delta)}}} \sum_{i=1}^n (\upsilon_{\bm{J}(n), i} )^2, \text{   where} \\
    E\left[ \frac{1}{\sigma_{\bm{J}(n)}^{\tau^{\tilde{Q}(\delta)}}} \sum_{i=1}^n (\upsilon_{\bm{J}(n), i} )^2\right] &= \frac{1}{\sum_{i=1}^n E[(\upsilon_{\bm{J}(n), i} )^2]} \sum_{i=1}^n E[(\upsilon_{\bm{J}(n), i} )^2] < \infty
\end{align*}
such that by the dominated convergence theorem
\begin{align*}
   \lim_{n\to\infty} \frac{1}{\sigma_{\bm{J}(n)}^{\tau^{\tilde{Q}(\delta)}}}\sum_{i=1}^n E[(\upsilon_{\bm{J}(n), i})^2I\{|\upsilon_{\bm{J}(n), i}]| > \epsilon\sigma_{\bm{J}(n)}^{\tau{\tilde{Q}(\delta)} }\}] = 0.
\end{align*}
and the Lindeberg condition holds. 
It remains to determine $\sigma_{\bm{J}(n)}^{\tau^{\tilde{Q}(\delta)^2}}$. Note that since $\sigma_{\bm{J}(n)}^{\tau^{\tilde{Q}(\delta)^2}} = \frac1n\sum_{i=1}^n\Var(\upsilon_{\bm{J}(n), i}) = \Var(\upsilon_{\bm{J}(n), i})$ it suffices to determine $\Var(\upsilon_{\bm{J}(n), i})$. For clarity in notation, we drop the observation index. Then, note that
\begin{align*}
   \upsilon_{\bm{J}(n)} &= \frac{\tilde{q}_\delta(B_{1}^{(2)}|\bm{X}, A_{\bm{J}(n)}^{-\tilde{Q}(\delta)}(\cdot))}{f(B_{1}^{(2)}|\bm{X}, A_{\bm{J}(n)}^{-\tilde{Q}(\delta)}(\cdot))}\left\{Y- E_{\tilde{Q}}[m(\bm{X}, A_{\bm{J}(n)}(\cdot))|\bm{X}, A_{\bm{J}(n)}^{-}(\cdot)]\right\} + E_{\tilde{Q}}[m(\bm{X}, A_{\bm{J}(n)}(\cdot))|\bm{X}, A_{\bm{J}(n)}^{-}(\cdot)] \\
   &\quad - Y - (E[\mu_{\bm{J}(n), i}] - \mu) \\
   &= \frac{\tilde{q}_\delta(B_{1}^{(2)}| \bm{X},A_{\bm{J}(n)}^{-\tilde{Q}(\delta)}(\cdot))}{f(B_{1}^{(2)}|\bm{X}, A_{\bm{J}(n)}^{-\tilde{Q}(\delta)}(\cdot))}\left\{Y- m(\bm{X}, A_{\bm{J}(n)}(\cdot))\right\} \\
    &\quad + \frac{\tilde{q}_\delta(B_{1}^{(2)}| \bm{X},A_{\bm{J}(n)}^{-\tilde{Q}(\delta)}(\cdot))}{f(B_{1}^{(2)}|\bm{X}, A_{\bm{J}(n)}^{-\tilde{Q}(\delta)}(\cdot))}\left\{m(\bm{X}, A_{\bm{J}(n)}(\cdot))- E_{\tilde{Q}}[m(\bm{X}, A_{\bm{J}(n)}(\cdot))|\bm{X}, A_{\bm{J}(n)}^{-}(\cdot)]\right\} \\
    &\quad + E_{\tilde{Q}}[m(\bm{X}, A_{\bm{J}(n)}(\cdot))|\bm{X}, A_{\bm{J}(n)}^{-}(\cdot)] - E[\mu_{\bm{J}(n), i}] \\
    &\quad - (Y - m(\bm{X}, A_{\bm{J}(n)}(\cdot)))  \\
    &\quad - (m(\bm{X}, A_{\bm{J}(n)}(\cdot)) - \mu) \\
    &= D_{Y-} + D_{\tilde{q}, \mu} + D_{\mu^{\tilde{Q}}} - D_{\mu} \\
    D_{Y-} &= \left(\frac{\tilde{q}_\delta(B_{1}^{(2)}|\bm{X}, A_{\bm{J}(n)}^{-\tilde{Q}(\delta)}(\cdot))}{f(B_{1}^{(2)}|\bm{X}, A_{\bm{J}(n)}^{-\tilde{Q}(\delta)}(\cdot))}- 1\right)\left\{Y- m(\bm{X}, A_{\bm{J}(n)}(\cdot))\right\} \\
    D_{\tilde{q}, \mu} &= \frac{\tilde{q}_\delta(B_{1}^{(2)}|\bm{X}, A_{\bm{J}(n)}^{-\tilde{Q}(\delta)}(\cdot))}{f(B_{1}^{(2)}| \bm{X},A_{\bm{J}(n)}^{-\tilde{Q}(\delta)}(\cdot))}\left\{m(\bm{X}, A_{\bm{J}(n)}(\cdot))- E_{\tilde{Q}}[m(\bm{X}, A_{\bm{J}(n)}(\cdot))|\bm{X}, A_{\bm{J}(n)}^{-}(\cdot)]\right\} \\
    D_{\mu^{\tilde{Q}}} &= E_{\tilde{Q}}[m(\bm{X}, A_{\bm{J}(n)}(\cdot))|\bm{X}, A_{\bm{J}(n)}^{-}(\cdot)] - E[\mu_{\bm{J}(n), i}] \\
    D_{\mu} &= m(\bm{X}, A_{\bm{J}(n)}(\cdot)) - \mu
\end{align*}
Since $E[\upsilon_{\bm{J}(n)}] = 0$, we have that  $\Var(\upsilon_{\bm{J}(n)}) = E[(D_{Y-} + D_{\tilde{q}, \mu} + D_{\mu^{\tilde{Q}}} - D_{\mu})^2]$. From the proof of Theorem 1 in \cite{schindl_incremental_2026}, we know that $E[D_{Y-}D_{\tilde{q}, \mu}] = E[D_{Y-}D_{\mu^{\tilde{Q}}}] = E[D_{Y-}D_{\mu}] = E[D_{\tilde{q}, \mu} D_{\mu^{\tilde{Q}}}] = 0$, $E[D_{Y-}^2] = E\left[ \left(\frac{\tilde{q}_\delta(B_{1}^{(2)}|\bm{X}, A_{\bm{J}(n)}^{-\tilde{Q}(\delta)}(\cdot))}{f(B_{1}^{(2)}|\bm{X}, A_{\bm{J}(n)}^{-\tilde{Q}(\delta)}(\cdot))}- 1\right)^2 \Var(Y|\bm{X}, A_{\bm{J}(n)}(\cdot))\right]$, and $E[D_{\mu^{\tilde{Q}}}^2] = \Var(E_{\tilde{Q}}[m(\bm{X}, A_{\bm{J}(n)}(\cdot))|\bm{X}, A_{\bm{J}(n)}^{-}(\cdot)])$. It remains to determine $E[D_{\tilde{q}, \mu}D_{\mu}]$ and $E[D_{\mu^{\tilde{Q}}}D_{\mu}]$. Then,
\begin{align*}
   E[D_{\tilde{q}, \mu}D_{\mu}] &=  E\left[\frac{\tilde{q}_\delta(B_{1}^{(2)}| \bm{X},A_{\bm{J}(n)}^{-\tilde{Q}(\delta)}(\cdot))}{f(B_{1}^{(2)}| \bm{X},A_{\bm{J}(n)}^{-\tilde{Q}(\delta)}(\cdot))}\left\{m(\bm{X}, A_{\bm{J}(n)}(\cdot))- E_{\tilde{Q}}[m(\bm{X}, A_{\bm{J}(n)}(\cdot))|\bm{X}, A_{\bm{J}(n)}^{-}(\cdot)]\right\} \right. \\
   &\quad \left.\times \{m(\bm{X}, A_{\bm{J}(n)}(\cdot)) - \mu \}\right] \\
   &= E\left[\frac{\tilde{q}_\delta(B_{1}^{(2)}| \bm{X},A_{\bm{J}(n)}^{-\tilde{Q}(\delta)}(\cdot))}{f(B_{1}^{(2)}| \bm{X},A_{\bm{J}(n)}^{-\tilde{Q}(\delta)}(\cdot))}m(\bm{X}, A_{\bm{J}(n)}(\cdot))^2 \right] \\
   &- E\left[\frac{\tilde{q}_\delta(B_{1}^{(2)}|\bm{X}, A_{\bm{J}(n)}^{-\tilde{Q}(\delta)}(\cdot))}{f(B_{1}^{(2)}|\bm{X}, A_{\bm{J}(n)}^{-\tilde{Q}(\delta)}(\cdot))}m(\bm{X}, A_{\bm{J}(n)}(\cdot))E_{\tilde{Q}}[m(\bm{X}, A_{\bm{J}(n)}(\cdot))|\bm{X}, A_{\bm{J}(n)}^{-}(\cdot)] \right] \\
   &= E\left[E\left[\frac{\tilde{q}_\delta(B_{1}^{(2)}| \bm{X},A_{\bm{J}(n)}^{-\tilde{Q}(\delta)}(\cdot))}{f(B_{1}^{(2)}| \bm{X},A_{\bm{J}(n)}^{-\tilde{Q}(\delta)}(\cdot))}m(\bm{X}, A_{\bm{J}(n)}(\cdot))^2 \middle\vert \bm{X}, A_{\bm{J}(n)}^{-}(\cdot)\right]\right] \\
   &- E\left[E\left[\frac{\tilde{q}_\delta(B_{1}^{(2)}| \bm{X},A_{\bm{J}(n)}^{-\tilde{Q}(\delta)}(\cdot))}{f(B_{1}^{(2)}| \bm{X},A_{\bm{J}(n)}^{-\tilde{Q}(\delta)}(\cdot))}m(\bm{X}, A_{\bm{J}(n)}(\cdot))E_{\tilde{Q}}[m(\bm{X}, A_{\bm{J}(n)}(\cdot))|\bm{X}, A_{\bm{J}(n)}^{-}(\cdot)] \middle\vert \bm{X}, A_{\bm{J}(n)}^{-}(\cdot)\right]\right] \\
   &= E\left[E_{\tilde{Q}}[m(\bm{X}, A_{\bm{J}(n)}(\cdot))^2|\bm{X}, A_{\bm{J}(n)}^{-}(\cdot)] - E_{\tilde{Q}}[m(\bm{X}, A_{\bm{J}(n)}(\cdot))|\bm{X}, A_{\bm{J}(n)}^{-}(\cdot)]^2\right] \\
   &= E\left[\Var_{\tilde{Q}}(m(\bm{X}, A_{\bm{J}(n)}(\cdot))| \bm{X}, A_{\bm{J}(n)}^{-}(\cdot))\right]. 
\end{align*}
Next,
\begin{align*}
   E[D_{\mu^{\tilde{Q}}}D_{\mu}] &= E\left[\left(E_{\tilde{Q}}[m(\bm{X}, A_{\bm{J}(n)}(\cdot))|\bm{X}, A_{\bm{J}(n)}^{-}(\cdot)] - E[\mu_{\bm{J}(n), i}]\right)(m(\bm{X}, A_{\bm{J}(n)}(\cdot)) - \mu)\right]  \\
   &= \Cov\left(E_{\tilde{Q}}[m(\bm{X}, A_{\bm{J}(n)}(\cdot))|\bm{X}, A_{\bm{J}(n)}^{-}(\cdot)], m(\bm{X}, A_{\bm{J}(n)}(\cdot))\right).
\end{align*}
Then, note that
\begin{align*}
    E[D_{\mu^{\tilde{Q}}}^2 + D_{\mu}^2 - 2D_{\mu^{\tilde{Q}}}D_{\mu} - 2D_{\tilde{q}, \mu}D_{\mu}] &= \Var\left(E_{\tilde{Q}}[m(\bm{X}, A_{\bm{J}(n)}(\cdot))|\bm{X}, A_{\bm{J}(n)}^{-}(\cdot)]\right)  + \Var(m(\bm{X}, A_{\bm{J}(n)}(\cdot))) \\
    &\quad -2\Cov\left(E_{\tilde{Q}}[m(\bm{X}, A_{\bm{J}(n)}(\cdot))|\bm{X}, A_{\bm{J}(n)}^{-}(\cdot)], m(\bm{X}, A_{\bm{J}(n)}(\cdot))\right) \\
    &\quad -2E\left[\Var_{\tilde{Q}}(m(\bm{X}, A_{\bm{J}(n)}(\cdot))| \bm{X}, A_{\bm{J}(n)}^{-}(\cdot))\right] \\
    &= \Var\left(m(\bm{X}, A_{\bm{J}(n)}(\cdot)) - E_{\tilde{Q}}[m(\bm{X}, A_{\bm{J}(n)}(\cdot))|\bm{X}, A_{\bm{J}(n)}^{-}(\cdot)]\right) \\
    &\quad - 2E\left[\Var_{\tilde{Q}}(m(\bm{X}, A_{\bm{J}(n)}(\cdot))| \bm{X}, A_{\bm{J}(n)}^{-}(\cdot))\right].
\end{align*}
Thus, we have that
\begin{align*}
        \sigma^{\tau^{\tilde{Q}(\delta)^2}}_{\bm{J}(n)} &= E\left[\left(\frac{\tilde{q}(B_{1}^{(2)}|\bm{X}, A_{\bm{J}(n)}^{-\tilde{Q}(\delta)}(\cdot))}{f(B_{1}^{(2)}| \bm{X},A_{\bm{J}(n)}^{-\tilde{Q}(\delta)}(\cdot))} - 1\right)^2\Var[Y|\bm{X}, A_{\bm{J}(n)}(\cdot)] \right.\\
        &\quad + \left.\left(\frac{\tilde{q}(B_{1}^{(2)}| \bm{X},A_{\bm{J}(n)}^{-\tilde{Q}(\delta)}(\cdot))}{f(B_{1}^{(2)}|\bm{X}, A_{\bm{J}(n)}^{-\tilde{Q}(\delta)}(\cdot))}\right)^2\left[m(\bm{X}, A_{\bm{J}(n)}(\cdot)) - E_{\tilde{Q}}[m(\bm{X}, A_{\bm{J}(n)}(\cdot))|\bm{X}_i, A_{\bm{J}(n)}^{-}(\cdot)]\right]^2\right] \\
        &\quad + \Var\left(m(\bm{X}, A_{\bm{J}(n)}(\cdot)) - E_{\tilde{Q}}[m(\bm{X}, A_{\bm{J}(n)}(\cdot))|\bm{X}, A_{\bm{J}(n)}^{-}(\cdot)]\right) \\
        &\quad - 2E\left[\Var_{\tilde{Q}}(m(\bm{X}, A_{\bm{J}(n)}(\cdot))| \bm{X}, A_{\bm{J}(n)}^{-}(\cdot))\right].
    \end{align*}

\end{proof}

\subsection{Corollary \ref{thm:coro_contrast_asymp}}
\begin{suppcoro}
   Assume the conditions and assumptions of Corollary \ref{thm:coro_estimator_asymp} hold. Then,
    \begin{align*}
        \frac{\sqrt{n}}{\sigma^{\tau^{\tilde{Q}(\delta)}}}( \hat{\tau}^{\tilde{Q}(\delta)} - \tau^{\tilde{Q}(\delta)}) &\to N(0,1), \:\:\: \text{for}\:\: \tilde{q} =q, q^*.
    \end{align*}
    \label{thm:coro_contrast_asymp}
\end{suppcoro}

\section{Cross-fit estimator form}
\label{sec:supp_cross_fit}
Consider partitioning the covariate set into $K$ where the $k$ fold has $n_k$ observations for $k \in \{1, \ldots, K\}$ . Let $\mathcal{I}_k$ be the observations in fold $k$ and let $\mathcal{I}_{-k}$ be the observations that are not in fold $k$. The estimate in fold $k$ is:
\begin{align*}
    \hat{\mu}^{\tilde{Q}(\delta)}_{\bm{J}, k}&= \frac1{n_k}\sum_{i \in \mathcal{I}_k} \frac{\hat{\tilde{q}}_{\delta, -k}(B_{1,i}^{(2)}|\bm{X}_i, A_{\bm{J},i}^{-\tilde{Q}}(\cdot))}{\hat{f}_{-k}(B_{1,i}^{(2)}|\bm{X}_i, A_{\bm{J},i}^{-\tilde{Q}}(\cdot))}\left\{Y_i \right.\\
    &\quad \left.- \int_{\mathcal{B}_{1}^{(2)}}\hat{m}_{-k}\left(\bm{X}_i, b_1^{(2)}\gamma_1^{(2)}(\cdot) + A_{\bm{J},i}^{-}(\cdot)\right)\hat{\tilde{q}}_{\delta, -k}\left(b_1^{(2)}|\bm{X}_i, A_{\bm{J},i}^{-\tilde{Q}}(\cdot))\right)db_1^{(2)}\right\} \\
    &\quad + \int_{\mathcal{B}_{1}^{(2)}}\hat{m}_{-k}\left(\bm{X}_i, b_1^{(2)}\gamma_1^{(2)}(\cdot) + A_{\bm{J},i}^{-}(\cdot)\right)\hat{\tilde{q}}_{\delta, -k}\left(b_1^{(2)}|\bm{X}_i, A_{\bm{J},i}^{-\tilde{Q}}(\cdot))\right)db_1^{(2)}, 
\end{align*}
where $\hat{m}_{-k}, \hat{\tilde{q}}_{\delta, -k},$ and $\hat{f}_{-k}$ are estimated on $\mathcal{I}_{-k}$. Then, the cross-fit estimator is an average across the folds:
\begin{align*}
    \hat{\mu}^{\tilde{Q}(\delta)}_{\bm{J}} &= \frac1K \sum_{k=1}^K \hat{\mu}^{\tilde{Q}(\delta)}_{\bm{J}, k}.
\end{align*}

\section{Additional simulation methods}
\label{sec:supp_sim_methods}
We drop the observation index for clarity of presentation. We simulate $V_1$ - $V_{10}$ from a truncated normal distribution with $E[V_i] = 0$, $\Var(V_i) = 1$ and $\text{Cov}(V_i, V_j) = 0$ for all $i\neq j = 1, \ldots, 10$. We then take $X_1$ - $X_{5} =$ $V_1$ - $V_{5}$ and $X_i = I[V_i < 0]$ for $i = 6, \ldots 10$ to obtain a mix of categorical and continuous covariates. We simulate $J_1 = 2,5,8$, $J_2 = 2$ and $J_3 = 2,5,8$ true basis functions over the three time intervals $[0, 149]$, $[t_1 = 150, t_2 = 200]$, and $[201, T= 300]$. Each true basis function, $\Psi^{(k)}_j$, was selected as a polynomial of order $j+1$ for $j = 1, \ldots, J_k$ and $k=1,2,3$. We consider $\gamma^{(2)}_1 = \Psi^{(2)}_1$. Let $l(\cdot)$ be a linear combination of $(\cdot)$. For scenarios evaluating $q$, the true basis coefficients, $\{\mathcal{A}^{(k)}\}_j^{J_k}$, are simulated sequentially in the following manner: 1) $\mathcal{A}^{(1)}_j \overset{\text{iid}}{\sim} TN(l^{(1)}_j\left(\{X_i\}_{i=1}^{10}\right),1, -1, 1)$ for $j=1, \ldots, J_1$;
2) $\mathcal{A}^{(3)}_j \overset{\text{iid}}{\sim} TN(l^{(3)}_j\left(\{\mathcal{A}^{(1)}_j\}_{j=1}^{J_1}\right),1, -1, 1)$ for $j=1, \ldots, J_3$; 3) $\mathcal{A}^{(2)}_2 \overset{\text{iid}}{\sim} TN(l^{(2)}_2\left(\{\mathcal{A}^{(3)}_j\}_{j=1}^{J_3}\right),1, -1, 1)$; and 4) $\mathcal{A}^{(2)}_1 \overset{\text{iid}}{\sim} TN(l^{(2)}_1\left(\{\mathcal{A}^{(3)}_j\}_{j=1}^{J_3}\right), \sigma, L, U)$.
We explore $\sigma = 10, 12$ to ensure sufficient residual variance for modeling the conditional density of $B_1^{(2)} = \mathcal{A}^{(2)}_1$. The truncation bounds $L$ and $U$ are selected to based on the minimum and maximum simulated values of $\mathcal{A}^{(2)}_1$. For scenarios evaluating  $q^*$, $\mathcal{A}^{(2)}_1$ is simulated differently to explore estimator performance when assumption \textbf{A$\bm{3}^*$} does and does not hold. For these scenarios, $\mathcal{A}^{(2)}_1 \overset{\text{iid}}{\sim} TN(l^{(2)}_1\left(\{\mathcal{A}^{(1)}_j\}_{j=1}^{J_1}\right) + cA^{(3)}_1, \sigma, L, U)$  for $c = 0,3$ and $SD(A^{(3)}_1) = 1, 3$. 
For all scenarios, nuisance function estimation tends to be more challenging as $J_1$ and $J_3$ increase.
We simulate the observed data as $A(t) + \epsilon_t$ where $A(t)$ is the true function and $\epsilon_t \sim N(0, 0.25)$  for $t = 1, 2, \ldots, 300$. We simulate the outcome as $Y \sim N(\mu(\bm{X}, A(\cdot)), 0.5)$ where $\mu(\bm{X}, A(\cdot)) = l\left(\bm{X}, \{\mathcal{A}^{(1)}_j\}_{j=1}^{J_1}, \{\mathcal{A}^{(3)}_j\}_{j=1}^{J_3}\right) + l_1\left(\{\mathcal{A}^{(2)}_j\}_{j=1}^{2}\right) + I\left[l\left(\bm{X}, \{\mathcal{A}^{(1)}_j\}_{j=1}^{J_1}, \{\mathcal{A}^{(3)}_j\}_{j=1}^{J_3} \right)\geq 0\right]l_2\left(\{\mathcal{A}^{(2)}_j\}_{j=1}^{2}\right)$. We include an interaction between the function at different time periods in $\mu(\bm{X}, A(\cdot))$ to add temporal complexity between the relationship between $Y$ and $A(\cdot)$.

\section{Additional NHANES physical activity data applications methods}
\label{sec:supp_real_data}

We remove any days that have over 14 hours of estimated accelerometer non-wear time or over 14 hours with data quality concerns \citep{leroux_organizing_2019}. We then remove any individuals who did not complete the full 7 days of quality data collection. After aggregating the full 7 days of data to a single 24-hour trajectory, we remove any observations with extreme MIMS values (MIMS $> 60$, less than 1\% of observations) or with missing data for any covariates. We implement our proposed methods to 7AM-10AM, 1PM-4PM, and 5PM-8PM. We choose $\gamma_1^{(2)} = \gamma_{1,3}(t)$ for 7AM-10AM as this function peaks earlier in the time periods which corresponds to the times where the SoFR coefficient is not significant (supplementary Figure \ref{fig:nhanes_sofr}). We choose $\gamma_1^{(2)} = \gamma_{1,1}(t)$ for 1PM-4PM as the SoFR coefficient is relatively constant across this time period. We choose $\gamma_1^{(2)} = \gamma_{1.5,1}(t)$ for 5PM-8PM as this function peaks later in the time period, which corresponds to the times where the SoFR coefficient is closer to zero.  To relate choices of $\delta$ to changes in physical activity, we compute the expected increase in MIMS over the corresponding time period by individual for each $\delta$. To increase the interpretability of the results, we transform the MIMS units to step counts by an individual-specific scaling factor (i.e., the ratio of the total step count over each time period to the total MIMS over each time period). We obtain minute-level step count data computed from the raw accelerometer data from
\cite{PhysioNet-minute-level-step-count-nhanes-1.0.0}. We choose to use step counts computed with the Oak algorithm \citep{straczkiewicz_one-size-fits-most_2023}, as it is designed to be a ``one-size-fits-most'' step counting algorithm. For further details on the step count computation, see \cite{koffman_comparing_2025}. We explore $\delta = 0.0075-0.05$ for all time periods as these $\delta$ values produced expected step count increases that could realistically be achieved. We implement the same estimation procedure as in Section \ref{sec:simulation4} and compute $\hat{\tau}_{\bm{J}(n)}^{Q(\delta)}$ and corresponding 95\% confidence bounds.

\begin{figure}
    \centering
    \includegraphics[width=0.8\linewidth]{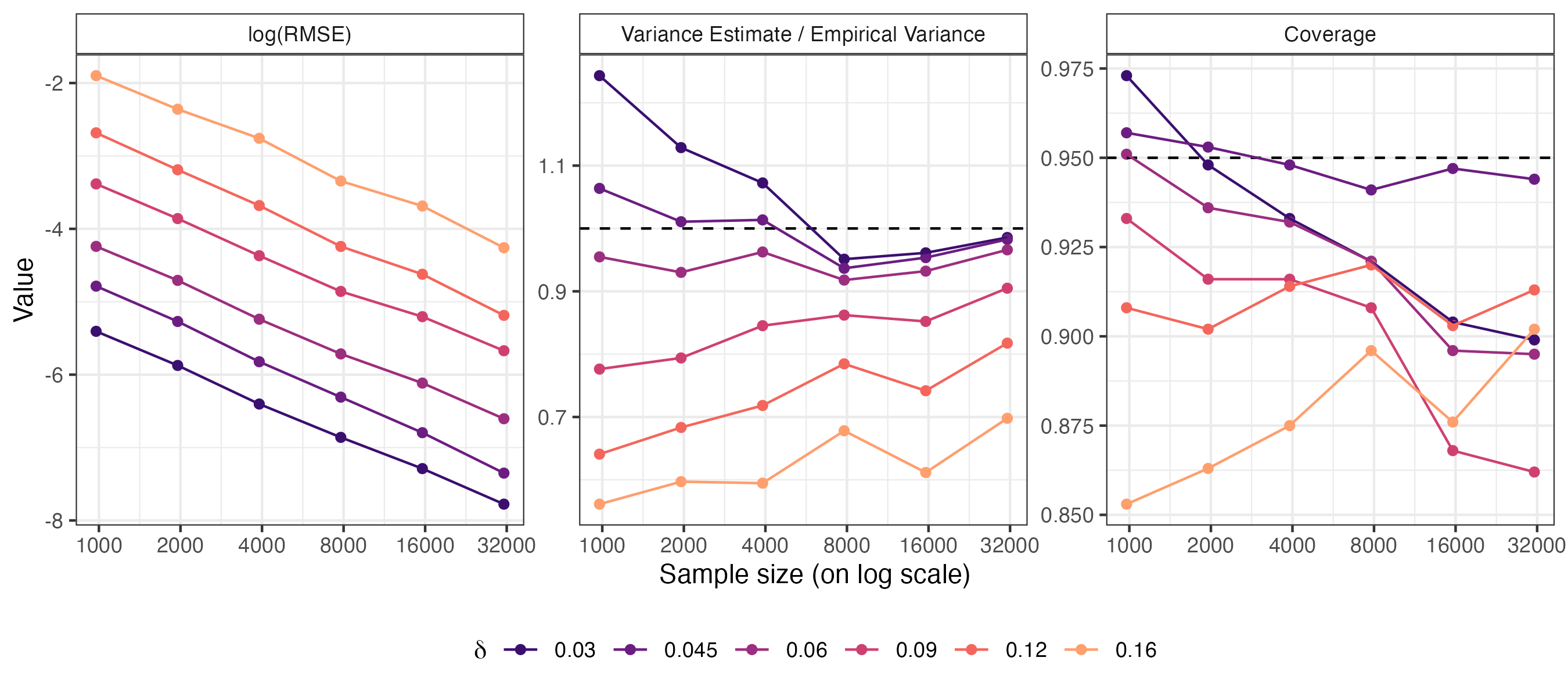}
    \caption{\footnotesize Estimator log(RMSE), the ratio of the average variance estimate and empirical estimator variance, and coverage for $\hat{\tau}^{Q(\delta)}_{\bm{J}(n)}$ by sample size and $\delta$. Dashed lines indicate the desired or nominal values for the right two panels.}
    \label{fig:sim_rate_tau}
\end{figure}

\begin{figure}
    \centering
    \includegraphics[width=0.8\linewidth]{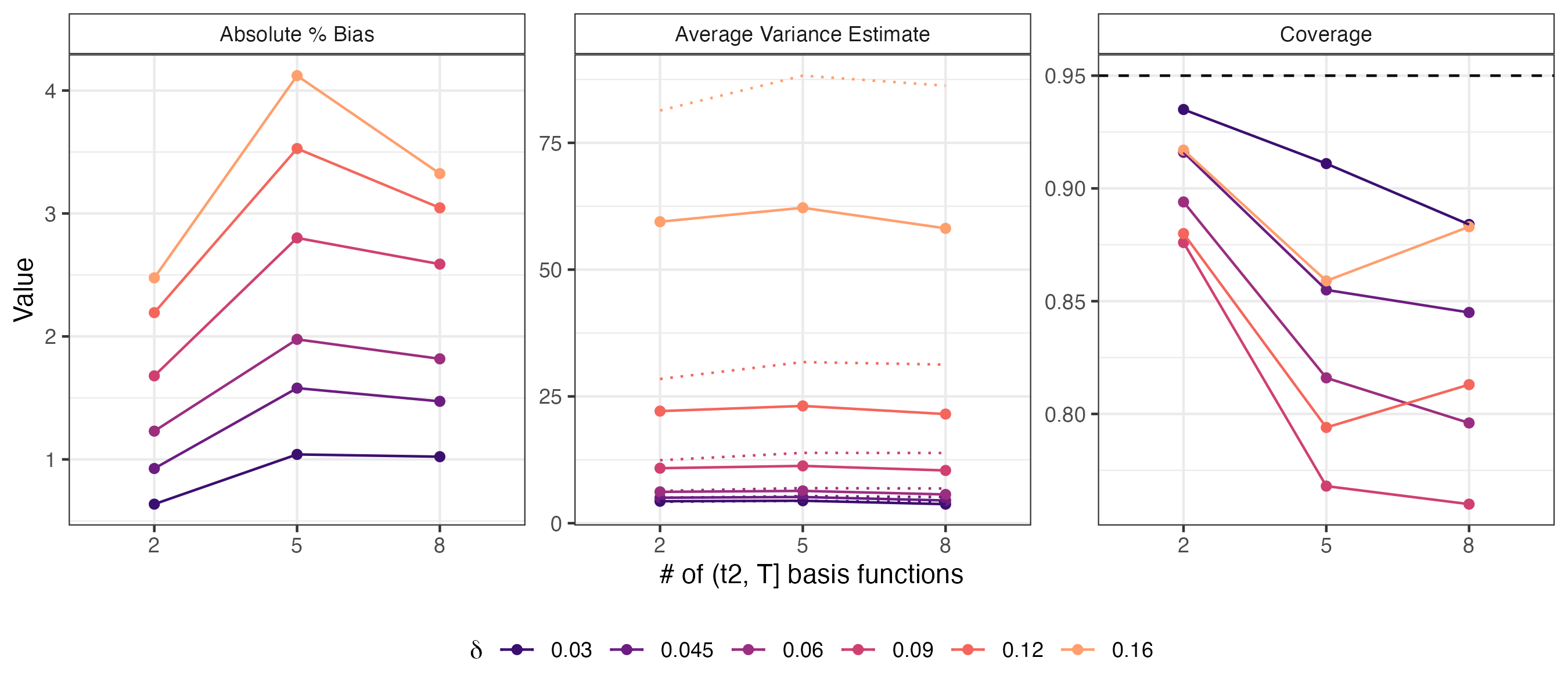}
    \caption{\footnotesize Estimator absolute percent bias, average variance estimate, and coverage by $J_3$ and $\delta$ for $\sigma = 10$, $J_1 = 8$, and 99\% of variance explained by the FPCA approximation for $\hat{\mu}^{Q(\delta)}_{\bm{J}(n)}$. The dotted line in the middle panel is the empirical estimator variance and the dashed line in the right panel is the nominal coverage level.}
    \label{fig:sim_q10_b1_8}
\end{figure}

\begin{figure}
    \centering
    \includegraphics[width=0.8\linewidth]{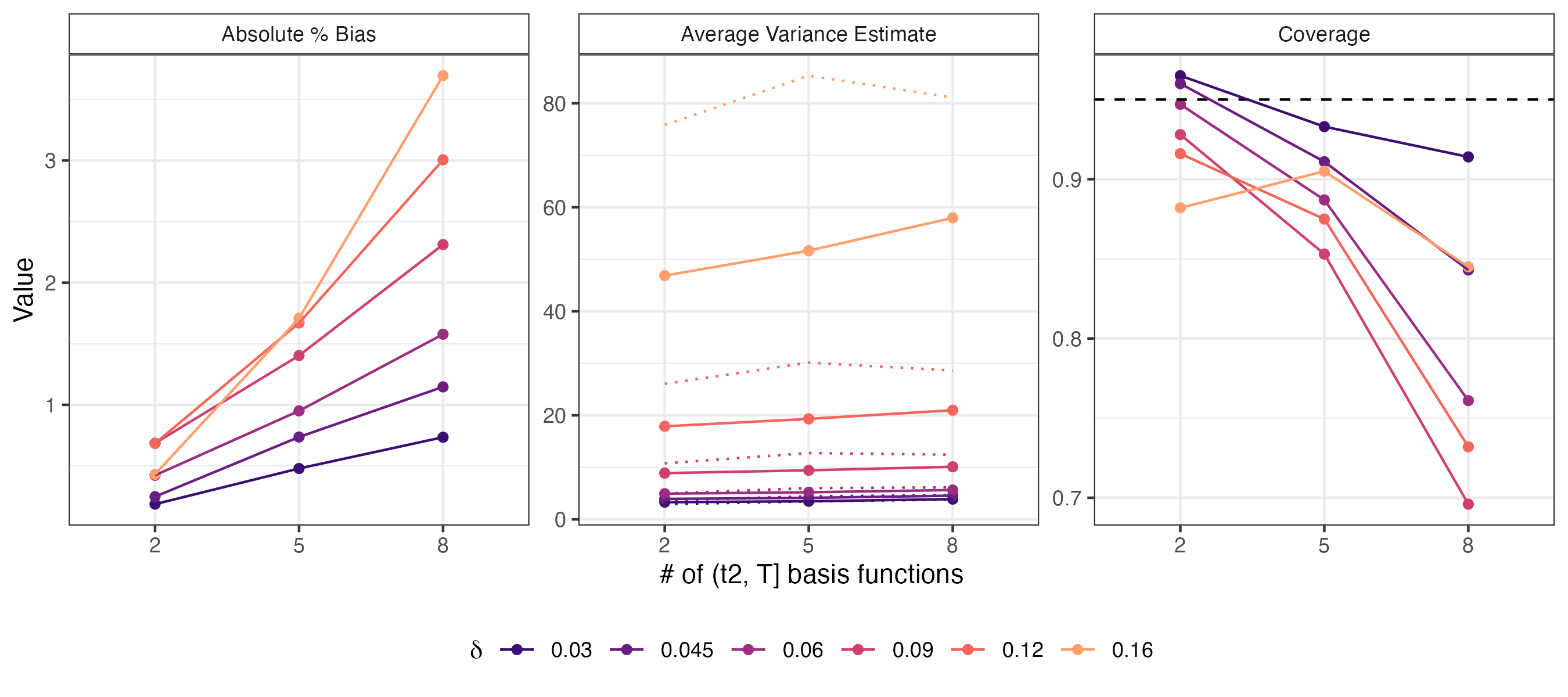}
   \caption{\footnotesize Estimator absolute percent bias, average variance estimate, and coverage by $J_3$ and $\delta$ for $\sigma = 10$, $J_1 = 2$, and 99.9\% of variance explained by the FPCA approximation for $\hat{\mu}^{Q(\delta)}_{\bm{J}(n)}$. The dotted line in the middle panel is the empirical estimator variance and the dashed line in the right panel is the nominal coverage level.}
    \label{fig:sim_q_q10_fpca999}
\end{figure}

\begin{figure}
    \centering
    \includegraphics[width=0.8\linewidth]{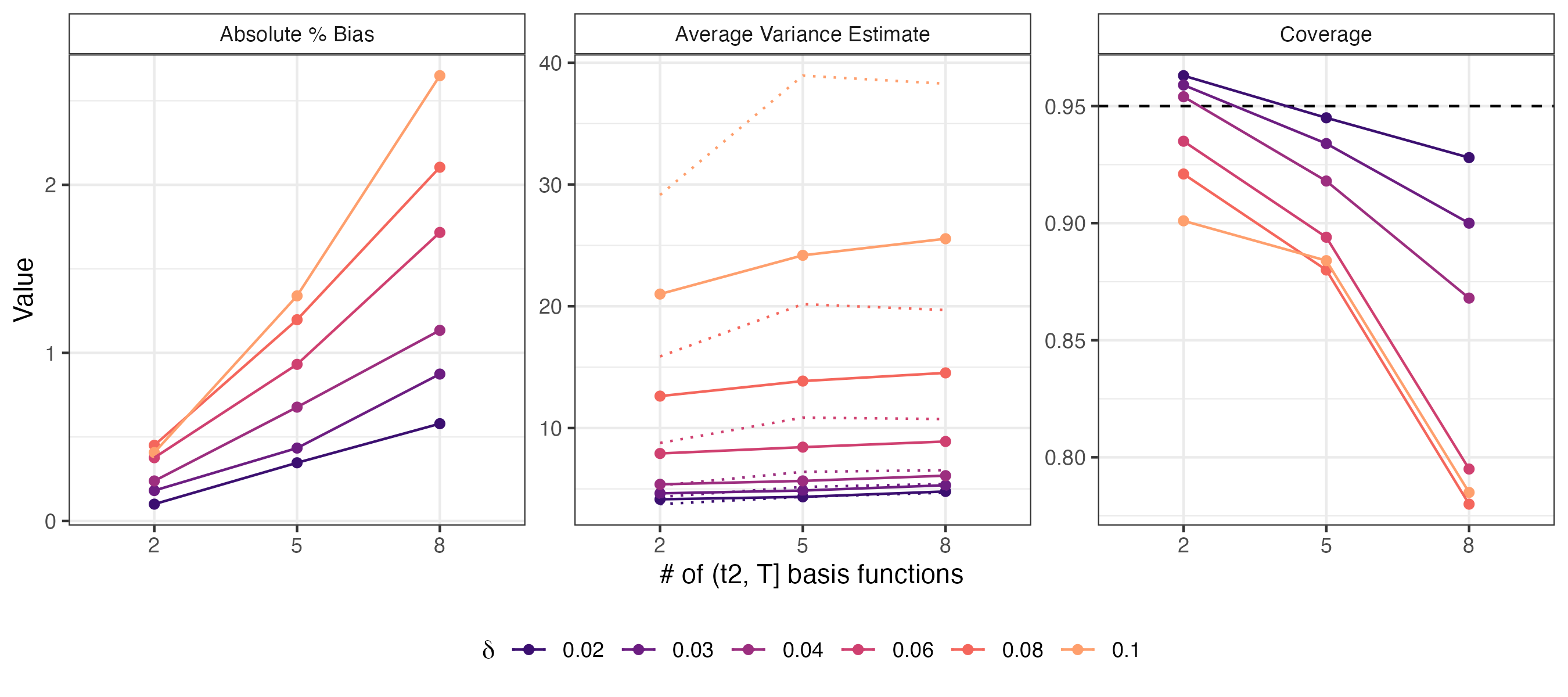}
    \caption{\footnotesize Estimator absolute percent bias, average variance estimate, and coverage by $J_3$ and $\delta$ for $\sigma = 12$, $J_1 = 2$, and 99\% of variance explained by the FPCA approximation for $\hat{\mu}^{Q(\delta)}_{\bm{J}(n)}$. The dotted line in the middle panel is the empirical estimator variance and the dashed line in the right panel is the nominal coverage level.}
    \label{fig:sim_q_q12}
\end{figure}

\begin{figure}
    \centering
    \includegraphics[width=0.8\linewidth]{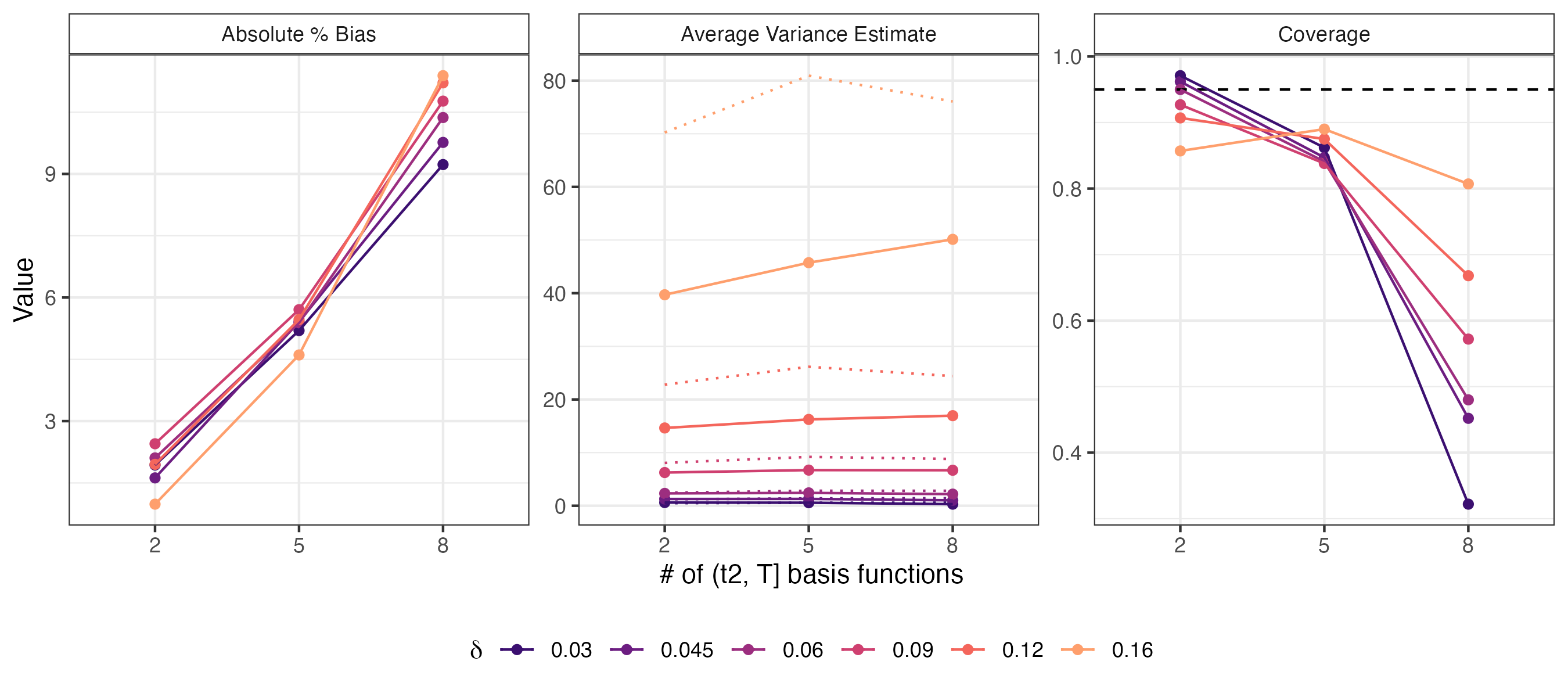}
    \caption{\footnotesize Estimator absolute percent bias, average variance estimate, and coverage by $J_3$ and $\delta$ for $\sigma = 10$, $J_1 = 2$, and 99\% of variance explained by the FPCA approximation for $\hat{\tau}^{Q(\delta)}_{\bm{J}(n)}$. The dotted line in the middle panel is the empirical estimator variance and the dashed line in the right panel is the nominal coverage level.}
     \label{fig:sim_q_q10_tau}
\end{figure}

\begin{figure}
    \centering
    \includegraphics[width=0.8\linewidth]{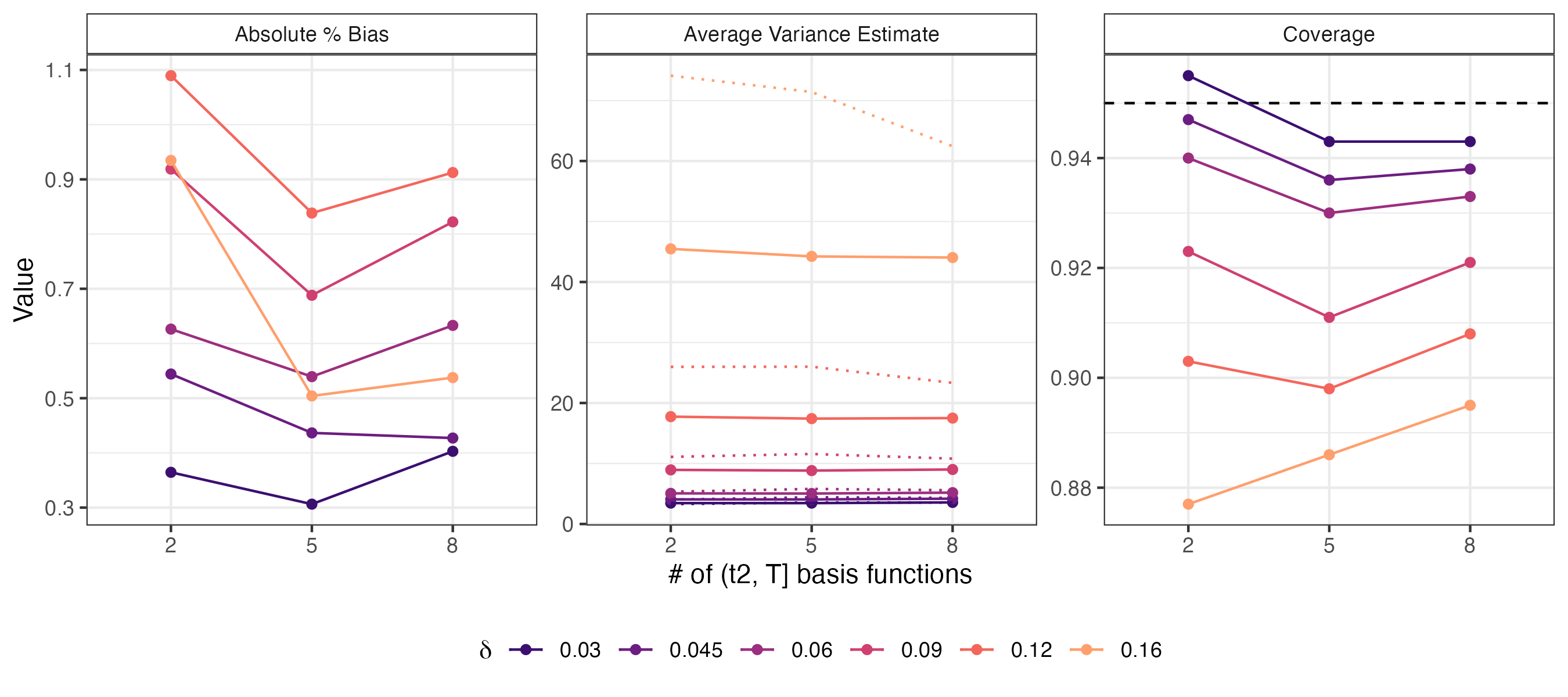}
    \caption{\footnotesize Estimator absolute percent bias, average variance estimate, and coverage by $J_3$ and $\delta$ for $\sigma = 10$, $J_1 = 2$, and 99\% of variance explained by the FPCA approximation for $\hat{\mu}^{Q^*(\delta)}_{\bm{J}(n)}$. The dotted line in the middle panel is the empirical estimator variance and the dashed line in the right panel is the nominal coverage level.}
    \label{fig:sim_qstar_b3_0}
\end{figure}

\begin{figure}
    \centering
    \includegraphics[width=0.8\linewidth]{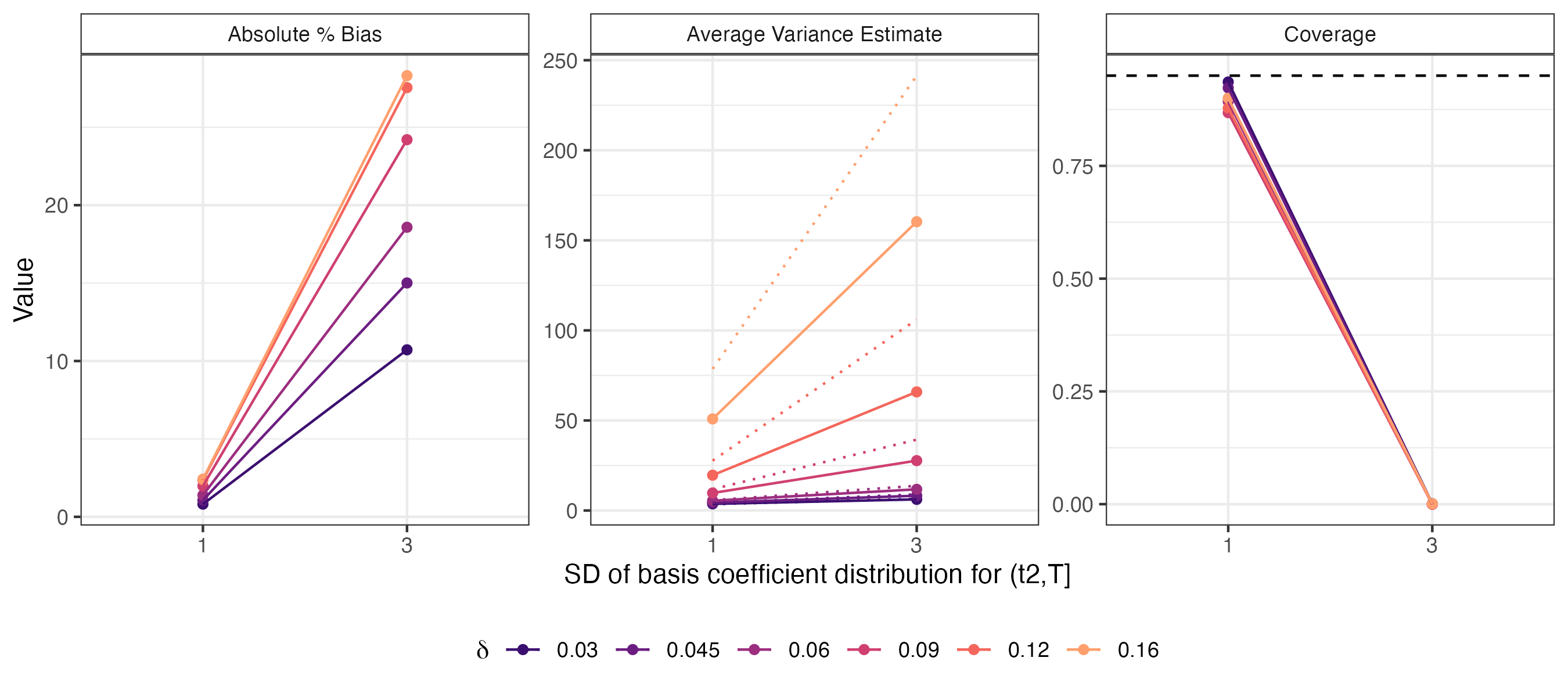}
    \caption{\footnotesize Estimator absolute percent bias, average variance estimate, and coverage by the standard deviation (SD) of distribution of $A^{(3)}_1$ and $\delta$ when assumption \textbf{A$\bm{3}^*$} does not hold for $\hat{\mu}^{Q^*(\delta)}_{\bm{J}(n)}$. Smaller SD indicates the function over $[0, t_1)$ is a stronger predictor of  $A^{(3)}_1$. The dotted line in the middle panel is the empirical estimator variance and the dashed line in the right panel is the nominal coverage level.}
     \label{fig:sim_qstar}
\end{figure}

\begin{figure}[h!]
    \centering
    \includegraphics[width=0.6\linewidth]{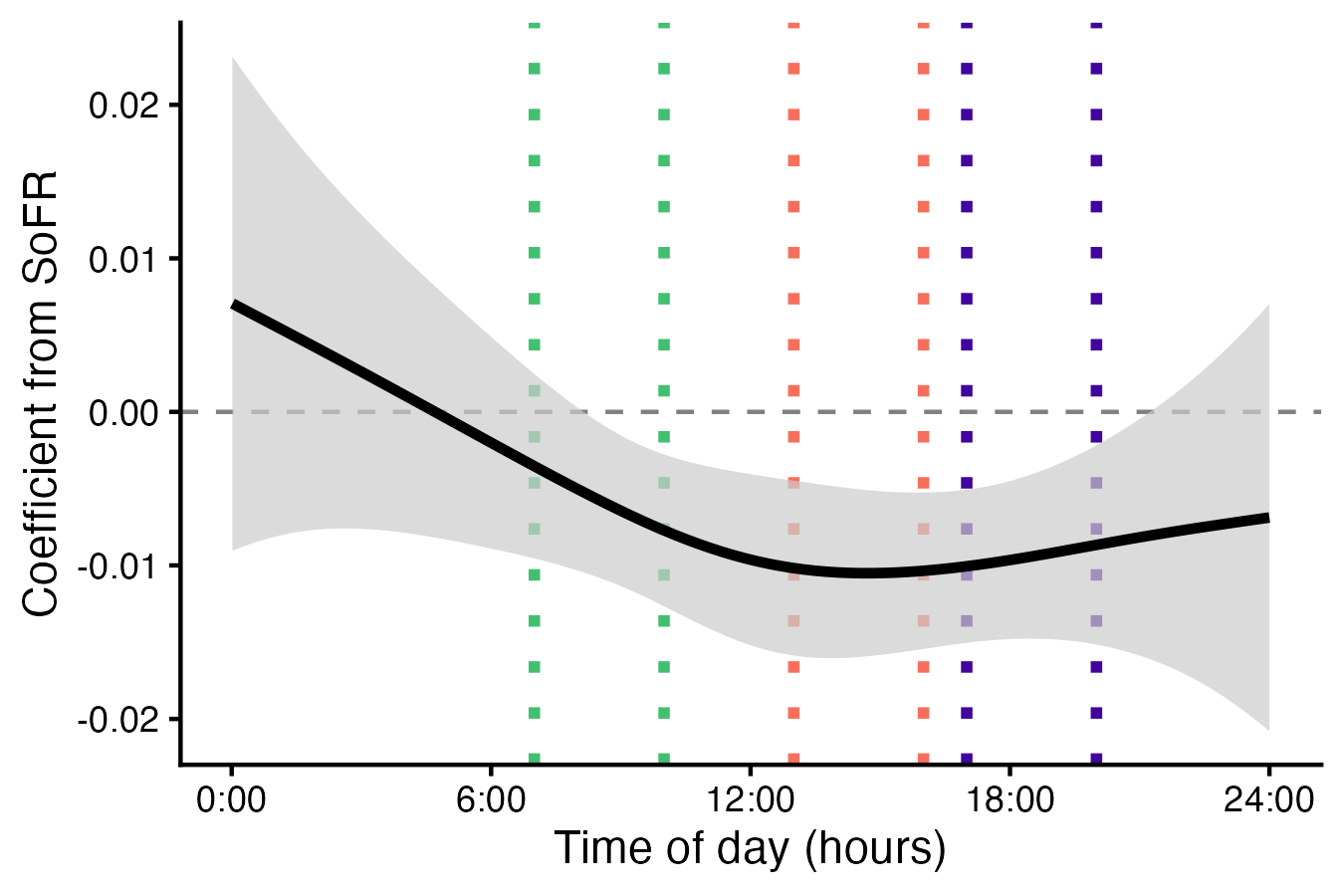}
    \caption{\footnotesize The estimated scalar-on-function regression coefficient associated with minute-level physical activity measured in MIMS. The exponentiated coefficient is the odds ratio of 5-year all-cause mortality for a one-unit increase in physical activity.}
    \label{fig:nhanes_sofr}
\end{figure}

\begin{figure}[h!]
    \centering
    \includegraphics[width=0.8\linewidth]{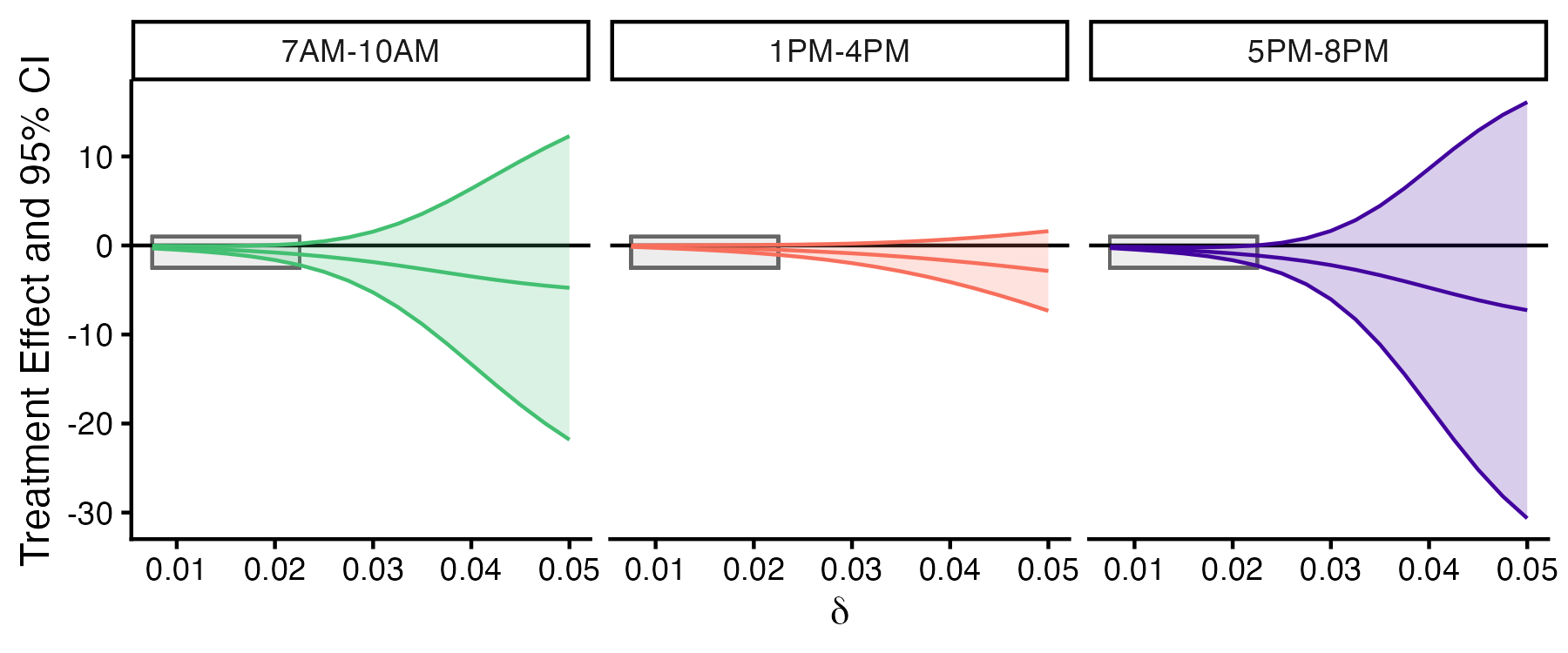}
    \caption{\footnotesize Treatment effect estimates and 95\% confidence bounds for the effect of stochastic policies that define increases in physical activity on 5-year all-cause mortality for the NHANES data. The columns indicate the time period the policy is implemented over.}
    \label{fig:nhanes_te_full_delta}
\end{figure}

\begin{figure}
    \centering
    \includegraphics[width=0.8\linewidth]{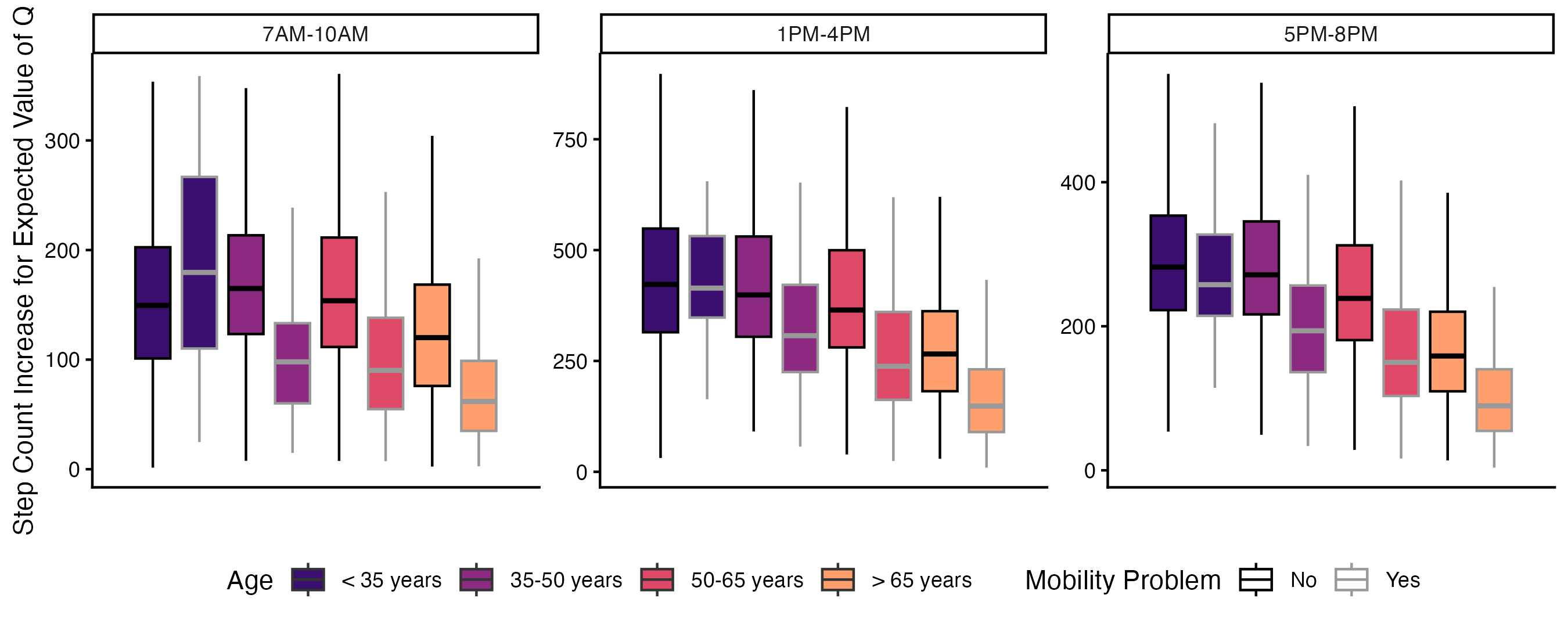}
    \caption{\footnotesize Distribution of the step count increase under the expected value of the stochastic policy for $\delta = 0.020$ by age and mobility status for the NHANES data.}
    \label{fig:nhanes_cov}
\end{figure}

\putbib
\end{bibunit}


\begin{thebibliography}{}

\bibitem[Chen et~al., 2018]{chen_national_2018}
Chen, T.-C., Parker, J.~D., Clark, J., Shin, H.-C., Rammon, J.~R., and Burt, V.~L. (2018).
\newblock National {Health} and {Nutrition} {Examination} {Survey}: {Estimation} {Procedures}, 2011-2014.
\newblock {\em Vital and Health Statistics. Series 2, Data Evaluation and Methods Research}, (177):1--26.

\bibitem[Chernozhukov et~al., 2018]{chernozhukov_doubledebiased_2018}
Chernozhukov, V., Chetverikov, D., Demirer, M., Duflo, E., Hansen, C., Newey, W., and Robins, J. (2018).
\newblock Double/debiased machine learning for treatment and structural parameters.
\newblock {\em The Econometrics Journal}, 21(1):C1--C68.

\bibitem[Crainiceanu et~al., 2024]{crainiceanu_functional_2024}
Crainiceanu, C.~M., Goldsmith, J., Leroux, A., and Cui, E. (2024).
\newblock {\em Functional {Data} {Analysis} with {R}}.
\newblock Chapman and Hall/CRC, New York.

\bibitem[Delaigle and Hall, 2010]{delaigle_defining_2010}
Delaigle, A. and Hall, P. (2010).
\newblock Defining probability density for a distribution of random functions.
\newblock {\em The Annals of Statistics}, 38(2):1171--1193.
\newblock Publisher: Institute of Mathematical Statistics.

\bibitem[Díaz and Hejazi, 2020]{diaz_causal_2020}
Díaz, I. and Hejazi, N.~S. (2020).
\newblock Causal {Mediation} {Analysis} for {Stochastic} {Interventions}.
\newblock {\em Journal of the Royal Statistical Society Series B: Statistical Methodology}, 82(3):661--683.

\bibitem[Díaz et~al., 2023]{diaz_nonparametric_2023}
Díaz, I., Williams, N., Hoffman, K.~L., and Schenck, E.~J. (2023).
\newblock Nonparametric {Causal} {Effects} {Based} on {Longitudinal} {Modified} {Treatment} {Policies}.
\newblock {\em Journal of the American Statistical Association}, 118(542):846--857.
\newblock Publisher: Taylor \& Francis \_eprint: https://doi.org/10.1080/01621459.2021.1955691.

\bibitem[Gao and Hastie, 2022]{gao_lincde_2022}
Gao, Z. and Hastie, T. (2022).
\newblock {LinCDE}: {Conditional} {Density} {Estimation} via {Lindsey}'s {Method}.
\newblock {\em Journal of Machine Learning Research}, 23(52):1--55.

\bibitem[Goldsmith et~al., 2011]{goldsmith_penalized_2011}
Goldsmith, J., Bobb, J., Crainiceanu, C.~M., Caffo, B., and Reich, D. (2011).
\newblock Penalized {Functional} {Regression}.
\newblock {\em Journal of computational and graphical statistics : a joint publication of American Statistical Association, Institute of Mathematical Statistics, Interface Foundation of North America}, 20(4):830--851.

\bibitem[Hall and Hosseini-Nasab, 2006]{hall_properties_2006}
Hall, P. and Hosseini-Nasab, M. (2006).
\newblock On properties of functional principal components analysis.
\newblock {\em Journal of the Royal Statistical Society: Series B (Statistical Methodology)}, 68(1):109--126.
\newblock \_eprint: https://rss.onlinelibrary.wiley.com/doi/pdf/10.1111/j.1467-9868.2005.00535.x.

\bibitem[Hashimoto et~al., 2025]{hashimoto_positive_2025}
Hashimoto, K., Dora, K., Murakami, Y., Matsumura, T., Yuuki, I.~W., Yang, S., and Hashimoto, T. (2025).
\newblock Positive impact of a 10-min walk immediately after glucose intake on postprandial glucose levels.
\newblock {\em Scientific Reports}, 15(1):22662.
\newblock Publisher: Nature Publishing Group.

\bibitem[Hijikata and Yamada, 2011]{hijikata_walking_2011}
Hijikata, Y. and Yamada, S. (2011).
\newblock Walking just after a meal seems to be more effective for weight loss than waiting for one hour to walk after a meal.
\newblock {\em International Journal of General Medicine}, 4:447--450.

\bibitem[Jiang et~al., 2026]{jiang_estimating_2026}
Jiang, Z., Cui, E., and Huling, J.~D. (2026).
\newblock Estimating causal effects of functional treatments with modified functional treatment policies.
\newblock arXiv:2602.09145 [stat].

\bibitem[John et~al., 2019]{john_open-source_2019}
John, D., Tang, Q., Albinali, F., and Intille, S. (2019).
\newblock An {Open}-{Source} {Monitor}-{Independent} {Movement} {Summary} for {Accelerometer} {Data} {Processing}.
\newblock {\em Journal for the measurement of physical behaviour}, 2(4):268--281.

\bibitem[Kennedy, 2019]{kennedy_nonparametric_2019}
Kennedy, E.~H. (2019).
\newblock Nonparametric {Causal} {Effects} {Based} on {Incremental} {Propensity} {Score} {Interventions}.
\newblock {\em Journal of the American Statistical Association}, 114(526):645--656.
\newblock Publisher: Taylor \& Francis \_eprint: https://doi.org/10.1080/01621459.2017.1422737.

\bibitem[Koffman and Muschelli, 2025]{PhysioNet-minute-level-step-count-nhanes-1.0.0}
Koffman, L. and Muschelli, J. (2025).
\newblock {Minute level step counts and physical activity data from the National Health and Nutrition Examination Survey (NHANES) 2011-2014}.
\newblock {\em {PhysioNet}}.
\newblock Version 1.0.0.

\bibitem[Leroux et~al., 2024]{leroux_nhanes_2024}
Leroux, A., Cui, E., Smirnova, E., Muschelli, J., Schrack, J.~A., and Crainiceanu, C.~M. (2024).
\newblock {NHANES} 2011-2014: {Objective} {Physical} {Activity} {Is} the {Strongest} {Predictor} of {All}-{Cause} {Mortality}.
\newblock {\em Medicine and Science in Sports and Exercise}, 56(10):1926--1934.

\bibitem[Leroux et~al., 2021]{leroux_quantifying_2021}
Leroux, A., Xu, S., Kundu, P., Muschelli, J., Smirnova, E., Chatterjee, N., and Crainiceanu, C. (2021).
\newblock Quantifying the {Predictive} {Performance} of {Objectively} {Measured} {Physical} {Activity} on {Mortality} in the {UK} {Biobank}.
\newblock {\em The Journals of Gerontology: Series A}, 76(8):1486--1494.

\bibitem[Robins, 1986]{robins_new_1986}
Robins, J. (1986).
\newblock A new approach to causal inference in mortality studies with a sustained exposure period—application to control of the healthy worker survivor effect.
\newblock {\em Mathematical Modelling}, 7(9):1393--1512.

\bibitem[Robins et~al., 2000]{robins_marginal_2000}
Robins, J.~M., Hernán, M.~A., and Brumback, B. (2000).
\newblock Marginal structural models and causal inference in epidemiology.
\newblock {\em Epidemiology}, 11(5):550--560.

\bibitem[Schindl et~al., 2026]{schindl_incremental_2026}
Schindl, K., Shen, S., and Kennedy, E.~H. (2026).
\newblock Incremental effects for continuous exposures.
\newblock arXiv:2409.11967 [stat].

\bibitem[Smirnova et~al., 2020]{smirnova_predictive_2020}
Smirnova, E., Leroux, A., Cao, Q., Tabacu, L., Zipunnikov, V., Crainiceanu, C., and Urbanek, J.~K. (2020).
\newblock The {Predictive} {Performance} of {Objective} {Measures} of {Physical} {Activity} {Derived} {From} {Accelerometry} {Data} for 5-{Year} {All}-{Cause} {Mortality} in {Older} {Adults}: {National} {Health} and {Nutritional} {Examination} {Survey} 2003–2006.
\newblock {\em The Journals of Gerontology: Series A}, 75(9):1779--1785.

\bibitem[Tan et~al., 2025]{tan_causal_2025}
Tan, R., Huang, W., Zhang, Z., and Yin, G. (2025).
\newblock Causal {Effect} of {Functional} {Treatment}.
\newblock {\em Journal of Machine Learning Research}, 26(91):1--39.

\bibitem[Wang et~al., 2026]{wang_flexible_2026}
Wang, J., Wong, R. K.~W., Zhang, X., and Chan, K. C.~G. (2026).
\newblock Flexible {Functional} {Treatment} {Effect} {Estimation}.
\newblock {\em Journal of Machine Learning Research}, 27(16):1--48.

\bibitem[Ying, 2024a]{ying_causality_2024-1}
Ying, A. (2024a).
\newblock Causality for {Complex} {Continuous}-time {Functional} {Longitudinal} {Studies} with {Dynamic} {Treatment} {Regimes}.
\newblock arXiv:2406.06868 [math] version: 1.

\bibitem[Ying, 2024b]{ying_causality_2024}
Ying, A. (2024b).
\newblock Causality for {Functional} {Longitudinal} {Data}.
\newblock In {\em Proceedings of the {Third} {Conference} on {Causal} {Learning} and {Reasoning}}, pages 665--687. PMLR.
\newblock ISSN: 2640-3498.

\bibitem[Zhang et~al., 2021]{zhang_covariate_2021}
Zhang, X., Xue, W., and Wang, Q. (2021).
\newblock Covariate balancing functional propensity score for functional treatments in cross-sectional observational studies.
\newblock {\em Computational Statistics \& Data Analysis}, 163:107303.

\end{thebibliography}


\begin{thebibliography}{}

\bibitem[Koffman et~al., 2025]{koffman_comparing_2025}
Koffman, L., Crainiceanu, C., and Muschelli, J. (2025).
\newblock Comparing {Step} {Counting} {Algorithms} for {High}-{Resolution} {Wrist} {Accelerometry} {Data} in {NHANES} 2011-2014.
\newblock {\em Medicine and Science in Sports and Exercise}, 57(4):746--755.

\bibitem[Koffman and Muschelli, 2025]{PhysioNet-minute-level-step-count-nhanes-1.0.0}
Koffman, L. and Muschelli, J. (2025).
\newblock {Minute level step counts and physical activity data from the National Health and Nutrition Examination Survey (NHANES) 2011-2014}.
\newblock {\em {PhysioNet}}.
\newblock Version 1.0.0.

\bibitem[Leroux et~al., 2019]{leroux_organizing_2019}
Leroux, A., Di, J., Smirnova, E., Mcguffey, E.~J., Cao, Q., Bayatmokhtari, E., Tabacu, L., Zipunnikov, V., Urbanek, J.~K., and Crainiceanu, C. (2019).
\newblock Organizing and analyzing the activity data in {NHANES}.
\newblock {\em Statistics in biosciences}, 11(2):262--287.

\bibitem[Schindl et~al., 2026]{schindl_incremental_2026}
Schindl, K., Shen, S., and Kennedy, E.~H. (2026).
\newblock Incremental effects for continuous exposures.
\newblock arXiv:2409.11967 [stat].

\bibitem[Straczkiewicz et~al., 2023]{straczkiewicz_one-size-fits-most_2023}
Straczkiewicz, M., Huang, E.~J., and Onnela, J.-P. (2023).
\newblock A “one-size-fits-most” walking recognition method for smartphones, smartwatches, and wearable accelerometers.
\newblock {\em npj Digital Medicine}, 6(1):29.
\newblock Publisher: Nature Publishing Group.

\end{thebibliography}
\end{document}